\documentclass[11pt,oneside,a4paper]{article}   	% use "amsart" instead of "article" for AMSLaTeX format
\usepackage{jheppub}
\usepackage{microtype}

\usepackage[all]{xy}
\usepackage{graphicx}				% Use pdf, png, jpg, or eps§ with pdflatex; use eps in DVI mode
\usepackage{epsfig}
\usepackage{amssymb}
\usepackage{amsthm}
\usepackage{physics}
\usepackage{enumerate}
\usepackage{float}
\usepackage{scalerel}
\usepackage{bm}
\usepackage{amsfonts} 
\usepackage{appendix}
\usepackage{verbatim}
\usepackage[T1]{fontenc}

\usepackage{mathtools}
\usepackage{tikz}
\usetikzlibrary{arrows,decorations.markings}

\usepackage[capitalise]{cleveref}
\usepackage{dsfont}

\usepackage[T1]{fontenc}

\DeclarePairedDelimiter\ceil{\lceil}{\rceil}
\DeclarePairedDelimiter\floor{\lfloor}{\rfloor}

\newtheorem{theorem}{Theorem}[section]

\newtheorem{lemma}[theorem]{Lemma}
\newtheorem{definition}{Definition}[theorem]

\crefname{case}{Case}{Cases}

\DeclareMathOperator{\poly}{poly}
\renewcommand{\paragraph}[1]{\textbf{\textit{#1}}}

\newcommand{\hyper}{\mathbb{H}^3}
\newcommand{\Hbulk}{H_\mathrm{bulk}}
\newcommand{\Hboundary}{H_\mathrm{boundary}}

\newcommand{\wmin}{w_{\min}}
\newcommand{\rmin}{r_{\min}}
\newcommand{\Hgeneric}{H_\mathrm{generic}}

\newcommand{\identity}{\mathds{1}}
\newcommand{\Hcode}{H_{\mathcal{C}}}
\newcommand{\Hnotcode}{H_{\overline{\mathcal{C}}}}
\newcommand{\boundaryHilbert}{\mathcal{H}_\mathrm{boundary}}
\newcommand{\Hbulkzero}{H_{\text{bulk},0}}
\newcommand{\Htildebulk}{\tilde{H}_{\text{bulk}}}

\title{Toy Models of Holographic Duality between local Hamiltonians}

\author{Tamara Kohler}
\author{and Toby Cubitt}
\affiliation{Department of Computer Science, University College London, UK}
\emailAdd{tamara.kohler.16@ucl.ac.uk}
\emailAdd{t.cubitt@ucl.ac.uk}


\arxivnumber{1810.08992}

\abstract{Holographic quantum error correcting codes (HQECC) have been proposed as toy models for the AdS/CFT correspondence, and exhibit many of the features of the duality.
  HQECC give a mapping of states and observables.
  However, they do not map local bulk Hamiltonians to local Hamiltonians on the boundary.
  In this work, we combine HQECC with Hamiltonian simulation theory to construct a bulk-boundary mapping between local Hamiltonians, whilst retaining all the features of the HQECC duality. %in which local bulk Hamiltonians map to local boundary Hamiltonians.
  This allows us to construct a %toy model of holographic
  duality between \emph{models}, encompassing the relationship between bulk and boundary energy scales and time dynamics.
  It also allows us to construct a map in the reverse direction: from local boundary Hamiltonians to the corresponding local Hamiltonian in the bulk.
  Under this boundary-to-bulk mapping, the bulk geometry emerges as an approximate, low-energy, effective theory living in the code-space of an (approximate) HQECC on the boundary.
  At higher energy scales, this emergent bulk geometry is modified in a way that matches the toy models of black holes proposed previously for HQECC.
  Moreover, the duality on the level of dynamics shows how these toy-model black holes can form dynamically.}

\begin{document}
\maketitle
\flushbottom

%\pagebreak

%\listoftodos

%\tableofcontents
%\enlargethispage{2\baselineskip}

%\pagebreak

\section{Introduction}

The AdS/CFT correspondence is a postulated duality between quantum gravity in $(d+1)$-dimensional, asymptotically anti-de-Sitter (AdS) space, and a conformal field theory (CFT) defined on its $d$-dimensional boundary~\cite{Maldacena:2003}.
It has provided insight into theories of quantum gravity, and has also been used as a tool for studying strongly-interacting quantum field theories.
Recently it has been shown that important insight into the emergence of bulk locality in AdS/CFT can be gained through the theory of quantum error correcting codes~\cite{Almheiri:2015}.
This idea has been used to construct holographic quantum error correcting codes (HQECC)~\cite{Pastawski:2015,Swingle:2012,Swingle:2012a,pluperfect,Hayden:2016,arpan_2016,osborne-17}, which realise many of the interesting structural features of AdS/CFT.

Holographic quantum codes give a map from bulk to boundary Hilbert space, hence also from observables in the bulk to corresponding boundary observables.
But the AdS/CFT correspondence is also a mapping between \emph{models}, not just between states and observables; it relates quantum theories of gravity in the bulk to conformal field theories in one dimension lower on the boundary.
For holographic code models, this means realising a mapping between local Hamiltonians in the bulk and local Hamiltonians on the boundary.

Since holographic quantum codes give a mapping from any bulk operator to the boundary, one can certainly map any local bulk Hamiltonian to the boundary.
But this gives a completely non-local boundary Hamiltonian, with global interactions that act on the whole boundary Hilbert space at once.
Local \emph{observables} deep in the bulk are expected to map under AdS/CFT duality to non-local boundary observables, so this is fine -- indeed, expected -- for observables.
But a global Hamiltonian acting on the entire boundary Hilbert space has lost all relation to the boundary geometry; there is no meaningful sense in which it acts in one dimension lower.
Indeed, for these toy models on finite dimensional spins, any Hamiltonian whatsoever can be realised using a global operator.
For the correspondence between bulk and boundary models to be meaningful, the local Hamiltonian describing the bulk physics needs to map to a \emph{local} Hamiltonian on the boundary.
For this reason,~\cite{Pastawski:2015,pluperfect,Hayden:2016,arpan_2016} study the mapping of observables and states in their construction, and do not apply it to Hamiltonians.

By standing on the shoulders of the holographic quantum code results, in particular the HaPPY code~\cite{Pastawski:2015}, and combining stabilizer code techniques with the recent mathematical theory of Hamiltonian simulation~\cite{cubitt:2017} and techniques from Hamiltonian complexity theory, we build on these previous results to construct a full holographic duality between quantum many-body models in 3D hyperbolic space and models living on its 2D boundary.
(We focus on 3D/2D dualities for our explicit constructions, as the smallest dimension where our simulation techniques can be applied, but the techniques extend to boundary dimensions~$\geq2$.)
This allows us to extend the toy models of holographic duality in previous HQECC to encompass local Hamiltonians, and in doing so enables us to say something about how energy scales and dynamics in the bulk are reflected in the boundary.
It also allows us to explore the duality in the other direction: from boundary to bulk.
This gives insight into how the hyperbolic bulk geometry emerges as the geometry of a low-energy effective theory, and how this effective bulk geometry gets distorted at higher energies.

The remainder of the paper is set out as follows. In \cref{summary} we present our main result, and give an overview of the proof.
In \cref{discuss} we discuss the implications of our results, including a toy model of black hole formation within these HQECC. The conclusions are presented in \cref{discussion}. The technical background and rigorous mathematical proofs of all the results are given in \cref{prelim} and \cref{technical}, respectively.

 %\cref{prelim} we introduce the background required to understand the construction, and prove a number of lemmas which are used in the main result.
%The background covered includes perfect tensors (\cref{prelim_1}), qudit stabilizer codes (\cref{prelim_2}), hyperbolic Coxeter groups (\cref{prelim_4}) and Hamiltonian simulation (\cref{prelim_5}).
%In \cref{general_construction} we set out the general procedure for constructing a HQECC in $\hyper$, and prove our main result: the holographic duality between $\hyper$ and its 2D boundary. The proof does not rely on the properties of any particular HQECC, and holds provided that there exists a HQECC in $\hyper$. In \cref{example_1} and \cref{example_2} we provide two examples of such a HQECC in $\hyper$.
\section{Main results} \label{summary}
In this paper we construct an explicit duality between quantum systems in 3D hyperbolic space, $\hyper$, and quantum systems on its 2D boundary, which encompasses states, observables, and local Hamiltonians. The map is a quantum error correcting code, where the logical Hilbert space is a set of `bulk' qudits, which are embedded in a tessellation of $\hyper$. The physical Hilbert space is a set of `boundary' qudits, which lie on the 2D boundary of $\hyper$. Every state and observable in the bulk/logical Hilbert space is mapped to a corresponding state / observable in the boundary/physical Hilbert space. The error correcting properties of the map means that it is possible to recover from erasure of part of the boundary Hilbert space, as in previous HQECC toy models.

Under our mapping, any local Hamiltonian in the bulk is mapped approximately to a 2-local, nearest-neighbour Hamiltonian in the boundary (where a $k$-local Hamiltonian is a sum over terms which each act non-trivially on at most $k$-qudits, and nearest-neighbour means the interactions are only between neighbouring). In the language of error correction, this means that the code subspace of our quantum error correcting code is approximately the low-energy subspace of a 2-local Hamiltonian $\Hboundary$, where time evolution in the code subspace is also governed by $\Hboundary$.\footnote{Note that this result does not contradict recent results in \cite{Woods:2019,Faist:2019} regarding the incompatibility of continuous symmetries and quantum error correction, as our $\Hboundary$ contains high-weight terms.}

It is important to emphasise that, as in the case of tensor network constructions of HQECC~\cite{Pastawski:2015,pluperfect,Hayden:2016,arpan_2016}, the duality we construct does not per se have anything to do with quantum gravity.
It gives a holographic duality for \emph{any} local quantum Hamiltonian, not specifically Hamiltonians modelling quantum gravity.
However, this duality does exhibit some of the structural features of the AdS/CFT correspondence.
Notably, entanglement wedge reconstruction and redundant encoding are seen in the construction.
The Ryu-Takayanagi formula is also approximately obeyed for connected bulk regions.\footnote{All these features are inherited from \cite{Pastawski:2015}, which our construction builds on.}

Therefore, one natural application of this construction is to toy models of the AdS/CFT correspondence.
This requires choosing a bulk Hamiltonian, $\Hbulk$, which models semi-classical gravity.
Applying our holographic duality to this particular choice of bulk Hamiltonian, the time dynamics and energetic properties of the toy model do then exhibit certain of the features expected of AdS/CFT, in addition to the static features inherited from the underlying HQECC construction (see \cref{discuss} for details).
However, this toy model certainly does not capture every aspect of AdS/CFT duality.
In particular, the boundary model we obtain is not a conformal field theory or Lorentz-invariant.
And it is constructed for non-relativistic quantum mechanical systems in Euclidean space, in which time appears as an external parameter, not relativistic quantum systems in Minkowski space.
We make no attempt in this work to understand whether AdS space can be embedded in some suitable way into $\hyper$.
(Indeed, a more fruitful approach for future research is likely to be to apply the \emph{techniques} we have developed in
AdS space, rather than attempting to use the duality on $\hyper$ directly.)

A complete toy model of AdS/CFT duality would have to address these and many other aspects, as well as incorporating gravity more fully.
Our holographic duality is one more step towards such a toy model, going beyond previous HQECC constructions, but as yet still a long way short of a full, mathematically rigorous construction of AdS/CFT.

\subsection{Rigorous statement of the result}
%In this paper we construct a duality between $\hyper$ and its 2D boundary that encompasses states, observables and Hamiltonians.
Our main results are encapsulated in the following theorem:

\begin{theorem} \label{main_intro}
  Let $\hyper$ denote 3D hyperbolic space, and let $B_r(x)\subset\mathbb{H}^3$ denote a ball of radius $r$ centred at $x$.
  Consider any arrangement of $n$ qudits in $\hyper$ such that, for some fixed $r$, at most $k$ qudits and at least one qudit are contained within any $B_r(x)$.
  Let $L$ denote the minimum radius ball $B_L(0)$ containing all the qudits (which wlog we can take to be centred at the origin).
  Let $\Hbulk = \sum_Z h_Z$ be any local Hamiltonian on these qudits, where each $h_Z$ acts only on qudits contained within some $B_r(x)$.

  Then we can construct a Hamiltonian $\Hboundary$ on a 2D boundary manifold $\mathcal{M}\in\mathbb{H}^3$ with the following properties:
  \begin{enumerate}
  \item%
    $\mathcal{M}$ surrounds all the qudits, has diameter $O\left(\max(1,\frac{\ln(k)}{r}) L + \log\log n\right)$, and is homeomorphic to the Euclidean 2-sphere.
  \item%
    The Hilbert space of the boundary consists of a triangulation of $\mathcal{M}$ by triangles of $O(1)$ area, with a qubit at the centre of each triangle, and a total of $O\left(n(\log n)^4\right)$ triangles/qubits.
  \item%
    Any local observable/measurement $M$ in the bulk has a set of corresponding observables/measurements $\{M'\}$ on the boundary with the same outcome. A local bulk operator $M$ can be reconstructed on a boundary region $A$ if $M$ acts within the greedy entanglement wedge of $A$, denoted $\mathcal{E}[A]$.\footnote{The entanglement wedge, $\mathcal{E}_A$ is a bulk region constructed from the minimal area surface used in the Ryu-Takayanagi formula. It has been suggested that on a given boundary region, $A$, it should be possible to reconstruct all operators which lie in $\mathcal{E}_A$~\cite{Headrick:2014}. The greedy entanglement wedge is a discretised version defined in~\cite[Definition~8]{Pastawski:2015}}
  \item%
    $\Hboundary$ consists of 2-local, nearest-neighbour interactions between the boundary qubits.
    Furthermore, $\Hboundary$ can be chosen to have full local $SU(2)$ symmetry; i.e.\ the local interactions can be chosen to all be Heisenberg interactions: $\Hboundary = \sum_{\langle i,j\rangle} \alpha_{ij} (X_iX_j + Y_iY_j + Z_iZ_j)$.
  \item%
    $\Hboundary$ is a $(\Delta_L,\epsilon,\eta)$-simulation of $\Hbulk$ in the rigorous sense of~\cite[Definition~23]{cubitt:2017}, with $\epsilon,\eta = 1/\poly(\Delta_L)$, $\Delta_L = \Omega\left(\norm{\Hbulk}^6\right)$, and where the interaction strengths in $\Hboundary$ scale as  $\max_{ij}\abs{\alpha_{ij}} = O\left(\Delta_L^{\poly(n\log(n))}\right)$.
  \end{enumerate}
\end{theorem}

This result allows us to extend toy models of holographic duality such as~\cite{Pastawski:2015,pluperfect,Hayden:2016,arpan_2016} to include a mapping between local Hamiltonians.
In doing so we show that the expected relationship between bulk and boundary energy scales can be realised by local boundary models. In particular, in our construction toy models of static black holes (as originally proposed in~\cite{Pastawski:2015}) correspond to high-energy states of the local boundary model, as would be expected in AdS/CFT.

Moreover, in our toy model we can say something about how dynamics in the bulk correspond to dynamics on the boundary. Even without writing down a specific bulk Hamiltonian, we are able to demonstrate that the formation of a (toy model) static black hole in the bulk corresponds to the boundary unitarily evolving to a state outside of the code space of the HQECC, as expected in AdS/CFT (see \cref{black_hole_formation} for details).

Finally, the Hamiltonian simulation construction allows us to derive the mapping in the other direction.
Given any local boundary Hamiltonian, one can derive a corresponding bulk Hamiltonian using rigorous formulations of perturbation theory.
Constructing boundary-to-bulk mappings is an important goal of full AdS/CFT, where the boundary CFT is better understood, and one of the aims is to understand properties of quantum gravity in the bulk which are less well understood.
Our results are a small step in this direction, though as emphasised above they are still a very long way from a full AdS/CFT model.

\subsection{Proof overview} \label{proof_overview}

To construct the bulk/boundary map between Hilbert spaces, observables and local Hamiltonians described by \cref{main_intro}, we combine new tensor network constructions of HQECC inspired by~\cite{Pastawski:2015}, with techniques originally developed in Hamiltonian complexity theory.
The key ingredient that allows us to construct bulk/boundary mappings that preserve locality of the Hamiltonians are generalisations of so-called ``perturbation gadgets'', which were originally developed to prove computational complexity results, together with the recently developed theoretical framework of analogue Hamiltonian simulation~\cite{cubitt:2017} which allows us to show that these give rise to a full duality between the bulk and boundary physics.

``Perturbation gadgets'' give a mathematically rigorous version of a concept that is well known in theoretical physics by other names. The quantum many-body models that are studied in condensed matter physics are understood to be effective theories that approximate the correct physics at low energies. For example, the Born-Oppenheimer approximation assumes the motion of atomic nuclei can be treated independently of their electron clouds, allowing effective models of just the electronic structure of molecules and materials to be derived by assuming the nuclei locations are fixed. These effective models are accurate at low energies. Similarly, since Wilson's seminal work~\cite{Wilson1,Wilson2}, quantum field theories in high-energy physics are believed to be effective, low-energy theories that emerge from some deeper, underlying model. In atomic and optical physics, one frequently performs perturbation expansions to some finite order to derive effective interaction Hamiltonians.

Perturbation gadgets apply the same general idea to build up Hamiltonians out of one type of interaction, that give rise to low-energy effective Hamiltonians composed of a different type of interaction. The main difference to the standard perturbation theory taught at undergraduate physics is to keep track of rigorous bounds on the approximation errors, rather than to simply truncate a perturbation series at finite order. A typical example of such a gadget is the Hamiltonian depicted in \cref{fig:gadget}, which consists only of two-body interactions, but whose low-energy effective Hamiltonian approximates to high accuracy a many-body interaction.

\begin{figure}
\centering
\begin{tikzpicture}
 \filldraw[black] (0,-0.5) circle (2pt) node[below]{A};
 \draw (1,0) circle (2pt) node[below]{$w$};
 \filldraw[black] (2,-0.50) circle (2pt) node[below]{$B$};
 \draw (0,-0.5) -- (0.9,0); \draw (1.1,0) -- (2,-0.5) ;
 \filldraw[black] (1,1) circle (2pt) node[above]{$C$};
 \draw (0,-0.5)--(1,1);
 \draw (2,-0.5)--(1,1);
 \draw (1,1) -- (1,0.1);
 \draw (0,-0.5) -- (2,-0.5);
 \node at (1,-1.5) {(a)};

 \filldraw[black] (6,-0.5) circle (2pt) node[below]{A};
 \filldraw[black] (8,-0.50) circle (2pt) node[below]{$B$};
 \filldraw[black] (7,1) circle (2pt) node[above]{$C$};
\draw (7,0) ellipse (1.5cm and 1.5cm);

\end{tikzpicture}
\caption{(a) The interaction graph of a Hamiltonian, $H$, consisting only of two body interactions, whose low-energy effective Hamiltonian approximates to high accuracy the 3-body interaction depicted in (b).
The white vertex in (a) represents an `ancilla qubit' - these don't appear in the low-energy effective Hamiltonian as they are projected into a one-dimensional subspace in this regime.
Let $H = \Delta^{\frac{1}{3}}\frac{(-A + B)^2}{2} + \frac{(A^2 + B^2)\otimes C}{2} - \Delta^{\frac{2}{3}}\left( C \otimes \ket{1}\bra{1}_w +  \frac{(-A+B)\otimes X_w}{\sqrt{2}}\right)$ and let $H_{\text{eff}} = A \otimes B \otimes C$.
If we restrict to the subspace with energy below $\Delta$, $H_{\text{eff}} \simeq H$ up to errors of order $\frac{1}{\Delta}$. (More precisely: $|| H|_\Delta - H_{\text{eff}} || \leq \frac{1}{\Delta}$.)
The perturbation gadgets used to construct this effective Hamiltonian were developed in \cite{oliveira:2005}.} \label{fig:gadget}
\end{figure}
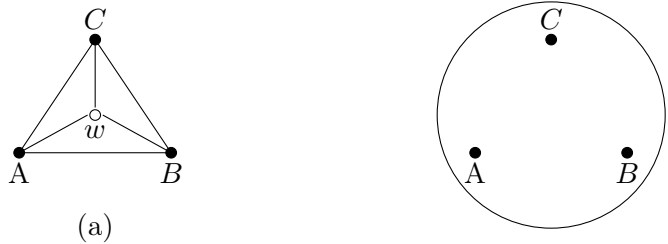

In our result, we use similar perturbation gadget methods to show that the highly non-local Hamiltonian that results from mapping a bulk Hamiltonian to the boundary using tensor network constructions, can be approximated to arbitrarily high accuracy as the emergent, low-energy effective Hamiltonian arising from a two-body, nearest-neighbour, local Hamiltonian on the boundary. The Hamiltonian simulation theory developed in~\cite{cubitt:2017} allows us to prove that this approximates the entire physics of the bulk.

\cite{Pastawski:2015} constructs a HQECC by building a tensor network composed out of perfect tensors, arranged in a tessellation of hyperbolic 2-space by pentagons. This gives a map from 2D bulk to 1D boundary.
However, the perturbative Hamiltonian simulation constructions of~\cite{cubitt:2017} only work in 2D or higher, which means we require at least a 3D bulk and 2D boundary.\footnote{The simulation techniques from \cite{cubitt:2017} cannot be used in 1D as they require a 2D interaction graph. See \cref{appendix_6} for details of the interaction graphs involved.}
We must therefore generalise the holographic tensor network codes to a space with a $\geq$ 2D boundary -- so a $\geq$ 3D bulk -- as a first step.
As it is the smallest dimension in which our techniques work, we focus on constructing explicit 3D/2D dualities. But the techniques we have developed readily extend to any boundary dimension~$\geq2$.
% In order to prove that our tensor network constructions meet the criteria to be a HQECC it is necessary to analyse properties of the tessellations the perfect tensors are embedded in.
When working in $\mathbb{H}^2$ it is possible to use the Poincare disc model to visualise the tessellations and determine their properties.
However, in $\mathbb{H}^3$ this is more difficult, and generalising the HQECC to 3D and higher requires a more systematic approach.
We use hyperbolic Coxeter groups\footnote{Coxeter groups were previously used in~\cite{arpan_2016} to describe tensor networks in $\mathbb{H}^2$ for toy models of holographic dualities.} to analyse honeycombings (higher-dimensional tessellations) of $\mathbb{H}^3$.
(These techniques also generalise beyond 3D.)
A Coxeter system is a pair $(W,S)$, where $W$ is a group, generated by a set $S \subset W$ of involutions, subject only to relations of the form $(s_is_j)^{m_{ij}} = 1$ where $m_{ii} = 1$, $m_{ij} \in \left(\mathbb{N} \setminus 1 \right) \cup \{\infty\}$ for $i \neq j$.
Coxeter groups admit a geometric representation as groups generated by reflections.
Associated to every hyperbolic Coxeter system is a Coxeter polytope $P \subseteq \mathbb{H}^d$, where $P$ tessellates $\mathbb{H}^d$.
All of the properties of the tessellation can be determined directly from the Coxeter system $(W,S)$ using combinatorics of Coxeter groups.
For example, we use the Coxeter relations to prove that the boundary of the HQECC is homeomorphic to the Euclidean 2-sphere.

Generalising the method in~\cite{Pastawski:2015}, we construct tensor networks by taking a Coxeter system $(W,S)$ with Coxeter polytope $P \subseteq \mathbb{H}^3$, and placing perfect tensors in each polyhedral cell of (a finite portion of) the tessellation of $\mathbb{H}^3$ by $P$.
Each perfect tensor in the interior of the tessellation has one free index, corresponding to a bulk qudit; the other indices are contracted with neighbouring tensors.
Tensors at the outer edge can be shown, again using the Coxeter relations, to have between $\ceil*{\frac{t}{2}}$ and $t-2$ additional free indices (where the perfect tensor has a total of $t$ indices), which correspond to qudits on the boundary.
We can show that if the tessellation of $\mathbb{H}^3$ associated to a Coxeter system $(W,S)$ has the properties required for a HQECC, then the associated Coxeter polytope $P$ has at least 7 faces, which means we require perfect tensors with at least 8 indices.
There are no qu\textit{b}it perfect tensors with $\geq6$ indices~\cite{Gour:2010,Rains:1999,Huber:2017}, so we must use qu\textit{d}it perfect tensors.

In order to later generate a local boundary model using perturbation gadgets, we need the tensor network to preserve the Pauli rank of operators.
(As we are working with qudits rather than qubits, we mean \emph{generalised} Pauli operators on qu\emph{d}its, rather than qubit Paulis, and we choose prime-dimensional qudits.)
% The generalised Pauli group on $n$ $p$-dimensional qudits, $\mathcal{G}_{n,p}$ is generated by the non-Hermitian clock and shift matrices which generalise the Pauli $Z$ and $X$ operators to higher dimensional qudits.\footnote{We restrict our attention to prime dimensional qudits.}
We use perfect tensors which describe qudit stabilizer absolutely maximally entangled states (AMES), constructed via the method in~\cite{Helwig:2013a} from classical Reed-Solomon codes.
Using properties of stabilizer groups, we show that tensor networks composed of these qudit stabilizer perfect tensors preserve the generalised Pauli rank of operators.

This Coxeter polytope qudit perfect tensor network gives a HQECC in $\hyper$.
The non-local boundary Hamiltonian is given by $\Hboundary' = H'+ \Delta_S H_S$, where $H_S$ is zero on the code-subspace of the HQECC and at least one on its orthogonal complement, $V$ is the encoding isometry of the HQECC and $H'$ satisfies $V\Hbulk V^\dagger = H'\Pi_{\mathcal{C}} = H'VV^\dagger$.\footnote{$\Hboundary'$ is not unique, as expected in AdS/CFT} Comparing with the classification of Hamiltonian simulations in~\cite{cubitt:2017}, this mapping is an example of a simulation. (In fact, a perfect simulation in the terminology of~\cite{cubitt:2017}.)

In order to construct a local boundary Hamiltonian we first determine the distribution of Pauli weights of the terms in $\Hboundary'$ from the properties of the Coxeter system.
We then use perturbation gadgets to reduce the boundary Hamiltonian to a 2-local planar Hamiltonian. This requires introducing a number of ancilla qudits in the boundary system.
The techniques we use follow the methods from~\cite{oliveira:2005}, however the perturbation gadgets derived in~\cite{oliveira:2005} can't be used in our construction as the generalised Pauli operators aren't Hermitian.
We therefore generalise those to qudit perturbation gadgets which act on operators of the form $P_A + P_A^\dagger$, where $P_A \in \mathcal{G}_{n,p}$.
% As far as we are aware they are the first perturbation gadgets designed specifically to be used with qudits.
% They are based on the qubit perturbation gadgets introduced in~\cite{Oliveira:2005}, and
These gadgets meet the requirements in~\cite{cubitt:2017,Bravyi:2017} to be perturbative simulations.
Finally we use simulation techniques from~\cite{cubitt:2017} to simulate the planar 2-local qudit Hamiltonian with a qubit Hamiltonian on a triangular lattice with full local SU(2) symmetry.
(The full technical details and proof are given in \cref{main_theorem}.)

\section{Discussion} \label{discuss}

\subsection{Main result}

In our bulk/boundary mapping, local Hamiltonians in 3D hyperbolic space, $\hyper$, are mapped to local Hamiltonians on its boundary.
At first glance this may appear to be at odds with the bulk reconstruction expected in AdS/CFT, where observables deep in the bulk are expected to map to non-local observables on the boundary CFT.
However, while the local simulation in our construction ensures that the boundary \emph{Hamiltonian} is local, it does not affect the locality of \emph{observables}.
As in the HaPPY code, observables deep in the bulk in our construction map to observables which require a large fraction of the boundary to be reconstructed, while observables near the boundary of the HQECC can be reconstructed on smaller fractions of the boundary (see point \labelcref{main:iii} from \cref{main_intro} for details).
This includes local Hamiltonian terms in the bulk when viewed as energy \emph{observables}.
For a local Hamiltonian term deep in the bulk, the corresponding energy observable on the boundary is not a single local term in the boundary Hamiltonian, but is made up a sum over many local terms acting across a large area of the boundary.

Point~\labelcref{main:i} from \cref{main_intro} demonstrates that the boundary surface in our construction really is a boundary geometrically.
The radius of the boundary surface is at a distance $\log\log n$ from the $n$ bulk qudits.
In \cref{other-geometries} we compare this with the spherical and Euclidean case.

Point \labelcref{main:iv} from \cref{main_intro}, which follows immediately from work in~\cite{cubitt:2017}, says that we can always choose that the boundary Hamiltonian in our holographic duality has full local SU(2) symmetry.
This hints at the possiblity of systematically incorporate local symmetries into the construction, such as gauge symmetries.% and, in particular, Lorentz invariance and conformal symmetry.
Doing so would involve tailoring our general construction to specific bulk models of interest, which is an intriguing possibility that we leave to future work.

Finally it is worth commenting on the energy scales in the construction.
There are two large energy scales.
The first, $\Delta_S$, is the energy penalty applied to boundary states which violate stabilizers of the HQECC.
Above this energy scale, the geometry of the corresponding tensor network in the bulk is modified in a way that corresponds to toy models of black holes proposed previously \cite{Pastawski:2015}.
We discuss this more fully in \cref{black_hole_formation}.
The second, $\Delta_L$, is the energy scale at which the local simulation of point \labelcref{main:v} from \cref{main_intro} breaks down.
At energies above $\Delta_L$ there is no longer any meaningful duality between bulk and boundary.

\subsection{Boundary to bulk mapping}

So far throughout this paper we have considered the tensor network as a map from bulk to boundary.
But one of the holy grails of AdS/CFT is to construct a mapping in the opposite direction: from boundary to bulk, as that opens up the possibility of studying bulk quantum gravity via the better-understood boundary CFT.
Our construction allows us to construct a toy model of such a boundary-to-bulk mapping.

Consider the boundary Hamiltonian $\Hboundary$ dual to some $k$-local bulk Hamiltonian $\Hbulk$ on $n_{\text{bulk}}$ qudits, from \cref{main_intro}.
Whatever the form of $\Hbulk$, $\Hboundary$ can always be decomposed in the form:
\begin{equation} \label{boundary-general}
  \Hboundary =  \Delta_L H_L + \Delta_S \tilde{H}_S + \tilde{H}_{\text{bulk}}
\end{equation}
where $H_L = \sum_{i\in\mathcal{A}} H_{i0}$ contains 1-body terms which act exclusively on the ancilla qubits $\mathcal{A}$, with interaction strengths $\geq\Delta_L$ (these arise from the perturbation gadget techniques); $\tilde{H}_S$ contains all the terms arising from the perturbative simulation of the stabilizer Hamiltonian $H_S$ (apart from the 1-body terms already included in $H_L$); and $\tilde{H}_{\text{bulk}}$ contains all the remaining terms arising from the perturbative simulation of $\Hbulk$ (again, apart from the 1-body terms included in $H_L$).

Let:
\begin{equation} \label{boundary-generic}
\Hgeneric = \Delta_L H_L + \Delta_S \tilde{H}_S,
\end{equation}
which is the boundary Hamiltonian dual to the zero Hamiltonian in the bulk.
We can recover a geometric interpretation of the bulk from \cref{boundary-generic}.
Consider decomposing $\Hgeneric|_{\frac{\Delta_L}{2}}$ into subspaces $\mathcal{H}_n$ of energy $(n-1/2) \Delta_S \leq E \leq (n+1/2)\Delta_S$ for $n \in \mathbb{N}$ such that $E \leq \frac{\Delta_L}{2}$.
Note that \cref{propertiesHS} and the fact that $\tilde{H}_S$ is a simulation (\cref{approxsim}) of $H_S$, $\Hgeneric|_{\frac{\Delta_L}{2}}$ is block-diagonal with respect to the Hilbert space decomposition $\mathcal{H}_{\text{boundary}} = \bigoplus_n \mathcal{H}_{n}$.

A boundary state $\ket{\psi}_{\text{boundary}}$ with support only on the subspace $\mathcal{H}_0$ corresponds to the bulk geometry of an unperturbed tensor network.
This subspace $\mathcal{H}_0$ is precisely the image of the isometry defined by the full tensor network.
Thus the bulk state $\ket{\psi}_{\text{bulk}}$ dual to $\ket{\psi}_{\text{boundary}}$ can be recovered by applying the inverse of the encoding isometry $V_0$ for the unbroken tensor network: $\ket{\psi}_{\text{bulk}} = V_0^\dagger\ket{\psi}_{\text{boundary}}$.

\begin{figure}
\centering
\includegraphics[scale=0.29]{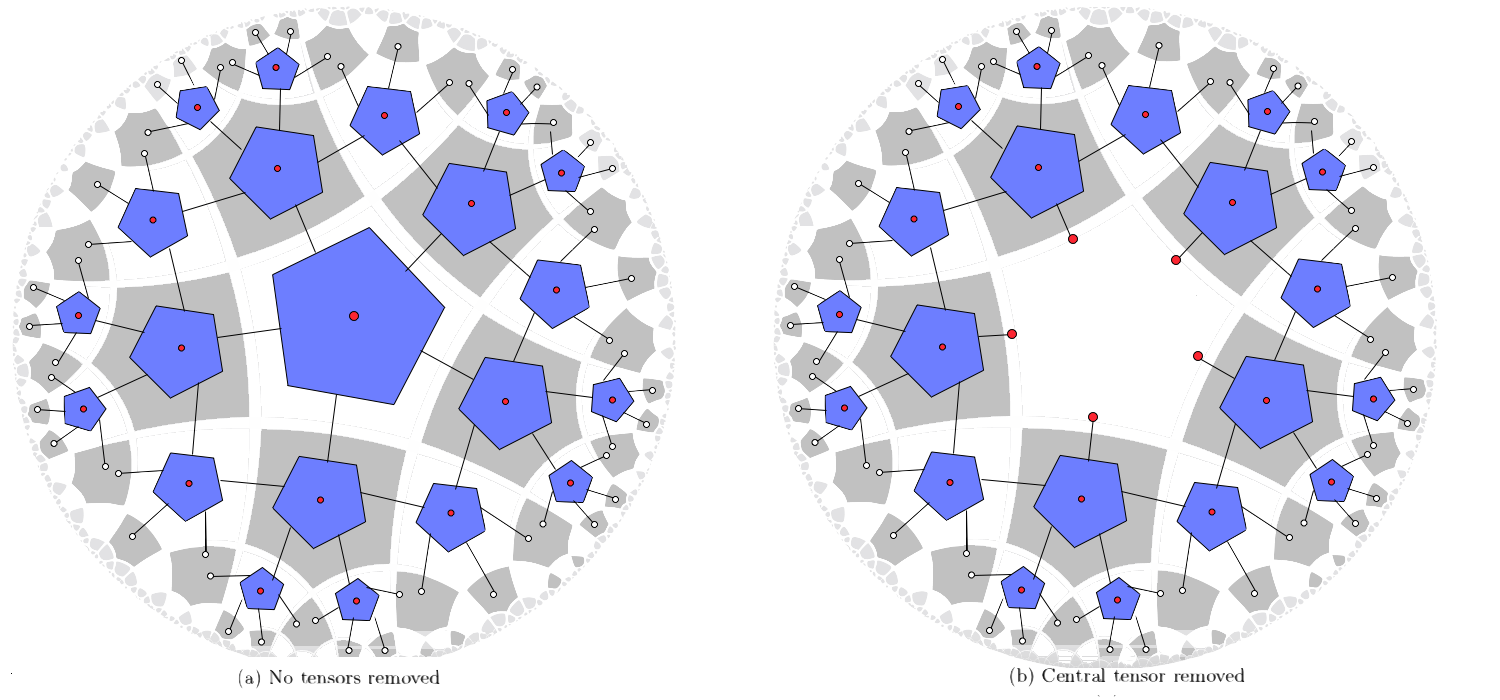}
\caption{(a) The pentagon code from~\cite{Pastawski:2015}. Red dots indicate bulk (logical) indices, white dots indicate boundary (physical) indices. (b) The pentagon code from~\cite{Pastawski:2015} with the central tensor removed. The indices from neighbouring tensors which were contracted with the central tensor are now logical indices.}\label{tensor-remove-fig}
\end{figure}

A boundary state $\ket{\psi}_{\text{boundary}} \in \mathcal{H}_n$ with $n\geq 1$ corresponds to a bulk geometry where one or more of the tensors in the network has been removed (see \cref{tensor-remove-fig} for details).
To see this, note that a state on the boundary is in $\mathcal{H}_{n\geq 1}$ iff it has violated one of the stabilizer terms of the HQECC (see \cref{propertiesHS}).
In the bulk it isn't meaningful to talk about states violating stabilizer terms, as the stabilizers don't act on the same Hilbert space as the bulk indices.
However, if a tensor is removed from the network, the stabilizer terms associated with that tensor \emph{do} act on the bulk indices of this modified tensor network.
Therefore, for any boundary state $\ket{\psi}_{\text{boundary}}\in\mathcal{H}_n$, it is possible to determine whether it is associated with a bulk geometry which contains holes in the tensor network ($n>0$), and how many ($n$), by considering just $\Hgeneric$.
Moreover, if the bulk geometry does contain holes, the location of holes can be inferred from which stabilizer terms in $\Hgeneric$ are violated by $\ket{\psi}_{\text{boundary}}$.
Since (see \cref{propertiesHS}) states violating different stabilizer terms are orthogonal, the subspace $\mathcal{H}_n$ corresponding to $n$ holes in the bulk further decomposes into $\mathcal{H}_{n} = \bigoplus_c \mathcal{H}_{n,c}$, where the sum is over all configurations $c$ of $n$ holes in the tensor network.
%Therefore subspace $\mathcal{H}_n$ corresponding to $n$ holes in the bulk further decomposes into $\mathcal{H}_{n} = \bigoplus_i \mathcal{H}_{n,i}$, where the $\mathcal{H}_{n,i}$ correspond to bulk geometries with all possible locations of $n$ holes in the tensor network.
The dual bulk state can be recovered by applying the inverse of the encoding isometry $V_{n,c}$ of the tensor network with holes in the appropriate locations: $\ket{\psi}_{\text{bulk}} = V_{n,c}^\dagger\ket{\psi}_{\text{boundary}}$.

By linearity, states $\ket{\psi}_{\text{boundary}}$ with support across multiple subspaces $\bigoplus_n\bigoplus_i\mathcal{H}_{n,c}$ correspond to coherent superpositions of states with different bulk geometries, and the dual state in the bulk can be recovered via: $\ket{\psi}_{\text{bulk}} = \bigoplus_n\bigoplus_c V_{n,c}^\dagger\ket{\psi}_{\text{boundary}}$.
All of this also extends to observables and operators on $\mathcal{H}_{\text{boundary}}$ in the obvious way.

A very similar analysis applies to general boundary Hamiltonians of the form \cref{boundary-general}.
$\Hgeneric$ determines the subspace decomposition $\mathcal{H}_{\text{physical}} = \bigoplus_n\bigoplus_c \mathcal{H}_{n,c}$ as before, which is independent of $\tilde{H}_{\text{bulk}}$, giving exactly the same bulk geometric interpretation of states in (and operators on) $\mathcal{H}|_{\frac{\Delta_L}{2}} = \mathcal{H}_{\text{physical}} = \bigoplus_n\bigoplus_c \mathcal{H}_{n,c}$.
The only new aspect is how to recover the bulk Hamiltonian $\Hbulk$ dual to $\tilde{H}_{\text{bulk}}$.

If we consider $\tilde{H}_{\text{bulk}}|_{\frac{\Delta_L}{2}}$, all ancilla qudits are projected onto a one-dimensional subspace by $\Delta_L H_L$, so do not appear.
Thus the resulting Hamiltonian only acts on the `physical' qudits in the boundary theory, $\mathcal{H}_{\text{physical}}$.
(I.e.\ the same Hilbert space as the non-local boundary Hamiltonian $\Hboundary' = H' + \Delta_S H_S$ obtained by pushing bulk interactions and stabilizers through the tensor network, see \cref{proof_overview} for discussion, or Step~1 of \cref{main} for full details.)

We can make the relationship between this Hamiltonian and the bulk geometry explicit, by considering how it looks with respect to the Hilbert space decomposition $\mathcal{H}_{\text{physical}} = \bigoplus_n\bigoplus_c \mathcal{H}_{n,c}$.
For example, its action on the Hilbert space of the unbroken tensor network is given by:
\begin{equation}
H_{\text{bulk},0} = V_0 \tilde{H}_{\text{bulk}}|_{\mathcal{H}_0} V_0^\dagger
\end{equation}
%If we denote this isometry by $W$ then we can construct the bulk Hamiltonian as:
%\begin{equation}
%H_{\text{bulk}(1)} = W^\dagger H_1 W
%\end{equation}

For general $\tilde{H}_{\text{bulk}}$, the resulting $H_{\text{bulk},0}$ will not be have any particular local structure.
However, there do exist $\tilde{H}_{\text{bulk}}$ that give rise to every $k$-local $H_{\text{bulk},0}$.
Indeed, we know exactly what form the $\tilde{H}_{\text{bulk}}$ corresponding to $k$-local $H_{\text{bulk},0}$ take, because these are precisely the Hamiltonians we constructed in \cref{main_intro} going in the other direction!
Moreover, if we start with a $\tilde{H}_{\text{bulk}}$ which is dual to a $k$-local $H_{\text{bulk},0}$, and add weak terms coupling e.g.\ $\mathcal{H}_0$ and $\mathcal{H}_1$, then the resulting $\Hbulk$ will have a $H_{\text{bulk},0}$ as a low-energy effective theory for energies $<\Delta_S$.
But it will now be possible for a state $\ket{\psi} \in \mathcal{H}_0$ to evolve to a state $\ket{\varphi} \in \mathcal{H}_1$ under the action of $\Hbulk$.

\subsection{Black hole formation in HQECC} \label{black_hole_formation}

In~\cite{Pastawski:2015} it was suggested that black holes can be incorporated into HQECC models of AdS/CFT by removing tensors from the tensor network. Recall, if a tensor is removed from the bulk then its one logical index is replaced by $t-1$ logical indices, corresponding to the indices that were previously contracted with the missing tensor (see \cref{black_hole_fig}).

\begin{figure}
\centering
\includegraphics[scale=0.29]{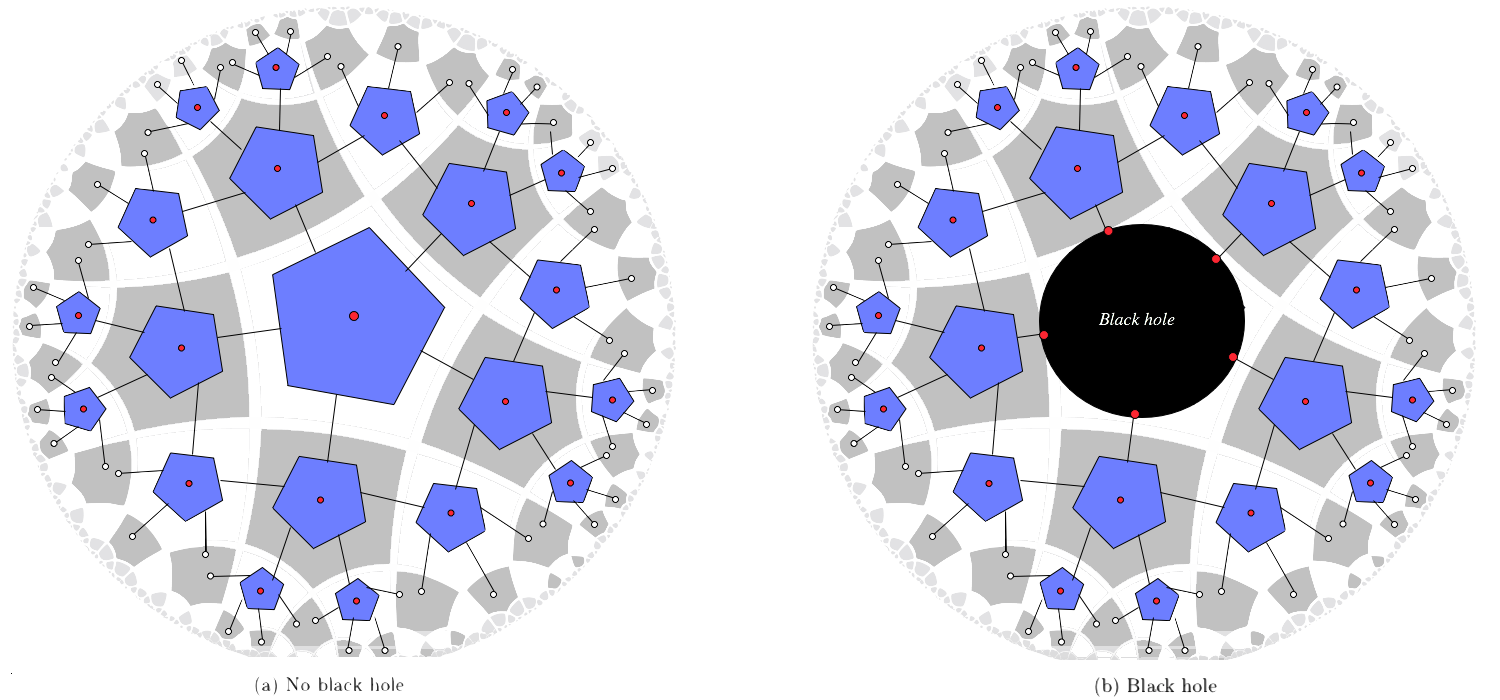}
\caption{(a) The pentagon code from~\cite{Pastawski:2015}. Red dots indicate bulk (logical) indices, white dots indicate boundary (physical) indices. (b) The pentagon code from~\cite{Pastawski:2015} with a central black hole. The central tensor has been removed, and the indices from neighbouring tensors which were contracted with the central tensor are now logical indices.}\label{black_hole_fig}
\end{figure}

This increases the code subspace of the boundary Hilbert space, and~\cite{Pastawski:2015} suggested that this can be interpreted as describing bulk configurations which contain a black hole. It is noted in~\cite{Pastawski:2015} that this model ensures that every boundary state is dual to a bulk state, and that black hole entropy scales with area, as expected from the Beckenstein-Hawking bound~\cite{Beckenstein,hawking:1975}.

This method of incorporating black holes into HQECC toy models of holography may at first appear ad hoc and artificial. However, by extending the toy models of holographic duality to encompass Hamiltonians, it emerges more naturally, and can also be extended to toy models of black hole \emph{formation}. Indeed, by considering the boundary dynamics dual to black hole formation in the bulk, we will see that removing a tensor from the network is the only way to preserve energy under unitary dynamics.

%%%%%By extending the toy models of holographic duality to encompass Hamiltonians we can extend this to consider toy models of black hole \emph{formation}. We can gain a qualitative picture of the boundary dynamics dual to black hole formation in the bulk from general considerations, even without considering a specific, discretised, toy model model of gravity in the bulk.

Consider a HQECC, with boundary Hamiltonian:
\begin{equation}
  \Hboundary =  \Hgeneric + \Htildebulk
\end{equation}
where $\Hgeneric$ is as defined in the previous section.
We will choose $\tilde{H}_{\text{bulk}}$ to ensure that $\Hbulkzero$  is some local Hamiltonian which models semi-classical gravity, given by:
\begin{equation}
\Hbulkzero = \sum_Z h_Z
\end{equation}
where the parameter $\Delta_{S}$ in $\Hgeneric$ satisfies: $\Delta_{S} \gg ||h_Z||$, but $\Delta_{S} < \sum_Z ||h_Z||$. To see that it is always possible to choose $\Htildebulk$ which gives $\Hbulkzero$ of this form, note that we could push $\Hbulkzero$ through the tensor network and construct such a $\Htildebulk$. Here we are allowing ourselves the freedom to add some additional perturbation to $\Htildebulk$.

%$\Htildebulk'$. $\Hbulkzero$ would then exactly describe the bulk dynamics dual to the boundary Hamiltonian $\Hboundary' = \Hgeneric + \Htildebulk'$. Here we are allowing ourselves the freedom to add some additional perturbation to $\Htildebulk'$. $\Hbulkzero$ is then an approximation to the low energy bulk dynamics which are dual to $\Hboundary$.

Consider the boundary state:
\begin{equation}
\ket{\psi_1} = \otimes_x \overline{A}_x \ket{\psi_0}
\end{equation}
where $\psi_0$ is the vacuum state (the ground state), $\overline{A}_x = W^\dagger A_x W$ is the boundary operator dual to some local bulk operator $A_x$, and the tensor product is over $O(n)$ boundary operators, which correspond to bulk local operators acting on a shell of $O(n)$ qudits near the boundary.\footnote{The hyperbolic geometry ensures that there are $O(n)$ qudits in a shell near the boundary.} This boundary state corresponds to a shell of matter in the bulk near the boundary.

Each bulk local excitation will pick up energy, $\delta_E$, from only a few of the local $h_Z$ terms, but the overall state will have large energy from summing over all these contributions. The energy on the boundary is equal to the energy in the bulk theory, so we must have:
\begin{equation}
\bra{\psi_1}\Hboundary\ket{\psi_1} = \sum_{x=0}^{O(n)} \delta_E = O(n \delta_E).
\end{equation}
For suitably chosen $\Delta_S$, the total energy of this configuration $E(\ket{\psi_1}) =  O(n \delta_E) > \Delta_S$. However, the local operators $\overline{A}_x$ are encoded versions of bulk operators, so acting on $\ket{\psi_0}$ with $\overline{A}_x$ will not take the state outside of the code subspace. Therefore, $\ket{\psi_1}$ is in the code subspace of the HQECC.

The boundary will evolve under $\Hboundary$. The bulk time dynamics can be approximated by $\Hbulkzero$. This will lead to an error $\epsilon t$ that increases only linearly in $t$~\cite[Proposition~29]{cubitt:2017}. So this approximation will be valid for sufficiently long times.

%To model a shell of in-falling matter in the bulk, we can start from the vacuum state (the ground state) and apply local operators to a shell of $O(n)$ bulk qudits near the boundary, where $n$ is the total number of qudits in the bulk:\footnote{The hyperbolic geometry ensures that there are $O(n)$ qudits in a shell near the boundary.}
%\begin{equation}
%\ket{\psi_1} = \otimes_x A_x \ket{\psi_0}
%\end{equation}
%for some spatially local operators $A_x$.
%Each local excitation will pick up energy, $\delta_E$, from only a few of the local $h_Z$ terms, but the overall state will have large energy from summing over all these contributions:
%\begin{equation}
%\Hbulk \ket{\psi_1} = \left(\sum_{x=0}^{O(n)} \delta_E\right) \ket{\psi_1}
%\end{equation}
%Thus, for suitable $\Delta_S$, the total energy of this configuration $E =  O(n \delta_E) > \Delta_S$, but these local excitations still correspond to a boundary state which is within the code space.

If we assume that under the action of $\Hbulkzero$ this shell of matter collapses inwards towards the centre of the HQECC (as would be expected from a Hamiltonian that models gravity), then the bulk will unitarily evolve to a configuration where most regions are in a low energy state (with respect to the Hamiltonians that act there), and most of the energy comes from a few ($O(1)$) terms near the centre of the HQECC. Denote the boundary state dual to this bulk configuration by $\ket{\psi_2}$.

The evolution is unitary, so we must have that:
\begin{equation}
\bra{\psi_2}\Hboundary\ket{\psi_2}= \bra{\psi_1}\Hboundary\ket{\psi_1}
\end{equation}
where $\bra{\psi_1}\Hboundary\ket{\psi_1}= O(n \delta_E) > \Delta_S$. The bulk must have the same energy as the boundary. But the maximum energy the bulk could have picked up from $O(1)$ $h_Z$ terms is given by:
\begin{equation}
\sum_{x=0}^{O(1)} ||h_Z|| = O\left(||h_Z||\right).
\end{equation}
By assumption $\Delta_{S} \gg ||h_Z||$, so  it is not possible for the bulk to pick up energy greater than $\Delta_S$ from $O(1)$ $h_Z$ terms.

If we consider the boundary system, it is possible for the boundary to pick up energy greater than $\Delta_S$ either from (the encoded version) of many $h_Z$ terms, or by violating one of the stabilizers of the HQECC. Since the bulk state dual to $\ket{\psi_2}$ cannot pick up energy from many $h_Z$ terms, the only way for $\ket{\psi_2}$ to pick up energy greater than $\Delta_S$ is to violate a stabilizer term.

Violating a stabilizer corresponds to picking up energy from the $\tilde{H}_S$ term in the boundary Hamiltonian. On the boundary it is clear that if we begin in a state inside the code space with energy greater than $\Delta_S$ it is possible to unitarily evolve to a state which is outside the code space and violates a stabilizer (provided that $\Htildebulk$ does not commute with $\tilde{H}_S$ -- see discussion below).

In the undisturbed bulk geometry it is not meaningful to talk about violating a stabilizer, as the stabilizers do not act on the same Hilbert space as the bulk logical indices. If, however, one of the bulk tensors has been removed, as in the models of black holes from~\cite{Pastawski:2015}, then the stabilizers corresponding to the removed tensor \emph{do} act on the Hilbert space of the $t-1$ new logical indices, and it is meaningful to talk about these stabilizers being violated.

%The only other way the system can pick up energy greater than $\Delta_S$ is by violating one of the stabilizers of the HQECC. On the boundary it is clear that if we begin in a state with energy greater than $\Delta_S$ it is possible to unitarily evolve to a state which is outside the code space and violates a stabilizer. But in the bulk it is not meaningful to talk about violating a stabilizer, as the stabilizers do not act on the same Hilbert space as the bulk logical indices. If, however, one of the bulk tensors has been removed, as in the models of black holes from~\cite{Pastawski:2015}, then the stabilizers corresponding to the removed tensor \emph{do} act on the Hilbert space of the $t-1$ new logical indices, and it is meaningful to talk about these stabilizers being violated.

Therefore, the only way for the system to conserve energy under these dynamics is for the tensor network corresponding to the boundary state to be `broken', and for at least one of the stabilizers corresponding to the missing tensor to be violated. This process therefore predicts the dynamical formation of a toy model black hole as proposed in~\cite{Pastawski:2015}.

In order for this process to occur we must have:
\begin{equation} \label{comm_eq}
[\tilde{H}_S,\Htildebulk] \neq 0
\end{equation}
(If we had some $\Htildebulk'$ which commuted with $\tilde{H}_S$ then there would be no coupling between the code-space, $\mathcal{C} \in \boundaryHilbert$, and the rest of the boundary, $\overline{\mathcal{C}} \in \boundaryHilbert$.
% We could write the total Hamiltonian $\tilde{H}'$ as:
% \begin{equation}
% \tilde{H}' = H_{\Delta}(\{a_i\}) + V(\{a_i\}) + V_{\text{commuting}}' = \Hcode \otimes \identity + \identity \otimes \Hnotcode
% \end{equation}
% Any state could be written as:
% \begin{equation}
% \rho =  \rho_{\mathcal{C}} \otimes \identity + \identity \otimes \rho_{\overline{\mathcal{C}}}
% \end{equation}
% Clearly if the boundary begins in a state which is in $\mathcal{C}$, then it will remain in $\mathcal{C}$ after evolution under $\tilde{H}'$. In particular,
So it would not be possible for $\ket{\psi_1} \in \mathcal{C}$ to unitarily evolve to $\ket{\psi_2} \in \overline{\mathcal{C}}$.)

In the bulk, \cref{comm_eq} implies that $\Hbulkzero$ is a low energy effective theory. The full bulk Hamiltonian includes some coupling between $\mathcal{H}_0$ and $\mathcal{H}_{\geq 1}$ (where $\mathcal{H}_n$ is as defined in the previous section). It therefore doesn't act only on the logical indices of the unbroken tensor network, and so the tensor network is always an approximation to the actual bulk theory. Another recent paper which examined bulk geometries containing black holes also showed that the bulk reconstruction in AdS/CFT is necessarily only approximate~\cite{hayden:2018}.

\cref{comm_eq} also implies that the toy model black hole degrees of freedom will (in general) be entangled with the rest of the tensor network. We can write the boundary Hamiltonian as:
\begin{equation}
\Hboundary|_{\frac{\Delta_L}{2}} =  \Hcode \otimes \identity + \identity \otimes \Hnotcode + \text{coupling terms}
\end{equation}
where \cref{comm_eq} ensures that the coupling terms are non-zero. Since the coupling terms are non-zero, any boundary state which is separable across the $\mathcal{C}$ / $\overline{\mathcal{C}}$ partition (equivalently, any bulk state which is separable across the black hole boundary) is not a stationary state of $\Hboundary$, so separable states will always evolve to entangled states. Therefore, the black hole degrees of freedom (if we trace out the rest of the tensor network) will always evolve to a mixed state. Similarly, on the boundary tracing out part of the system will lead to a mixed state. In this sense the black hole (and the equivalent state on the boundary) has thermalized.

It also follows from this discussion that the toy models of black holes correspond to high energy states on the boundary, which pick up their energy from a small number of high energy terms in the Hamiltonian.

Therefore, with an appropriately chosen $\Htildebulk$, we can model black hole formation in our HQECC. No information is lost in this process (as the dynamics are unitary we can always reverse them). But the isometry which takes bulk states to boundary states will have changed, so the `dictionary' for reconstructing the bulk state from the boundary state will be different. In particular, the fraction of the boundary needed to reconstruct an operator acting on the central bulk region will increase in the presence of a toy-model black hole.

To see this consider a central black hole where one tensor is removed from the HQECC. An operator which acts on the degrees of freedom representing the black hole acts on all $t-1$ of the logical indices, so will need to be pushed through all of the $t-1$ tensors at radius 1 in the HQECC. In the absence of any black hole an operator acting on the central bulk index could be reconstructed via pushing through just $\ceil*{\frac{t}{2}}$ of the tensors at radius 1 in the HQECC. Thus the fraction of the boundary needed for bulk reconstruction of the centre has increased.

Throughout this discussion we have concentrated on a black hole in the centre of the bulk for clarity, but these qualitative conclusions apply equally well to black holes situated at any point in the bulk.

\subsection{Other geometries} \label{other-geometries}

We have constructed a duality between quantum many-body models in $\hyper$ and models living on its 2D boundary, as a toy model for a duality between Anti-de Sitter space and its boundary. From a cosmological perspective it would also be interesting to consider toy models of dualities between positively curved / flat geometries, and their boundaries.

There is no reason to suspect that the error correcting properties of AdS/CFT should be recreated in such dualities, so it is not clear that the error correcting code constructions of HQECC will be relevant. However, the analogue Hamiltonian simulation theory from~\cite{cubitt:2017} can be applied in any geometry. It follows immediately from the results in~\cite{cubitt:2017} that it is possible to construct a duality between Euclidean or spherical geometry in dimension~3, and a 2D `boundary' surface. However, it is not clear whether such a `boundary' surface can be considered a geometric boundary in any meaningful sense.

In Euclidean geometry, results from~\cite{cubitt:2017} imply that in order to simulate $n$ bulk qudits in $\mathbb{E}^3$ with a local boundary model requires $O(\poly(n))$ boundary qudits. If we maintain the density of qudits from the bulk on the boundary, this implies that if the bulk qudits were contained in a ball of radius $R$, then the boundary surface would be at a radius $R' = O(R + \poly(n))$. So the distance between the bulk qudits and the boundary surface would increase polynomially with $n$.

The situation in the positively curved case is worse. $\mathbb{S}^3$ is finite, so the boundary surface required to simulate $n$ qudits which lie in $\mathbb{S}^3$ might not itself lie in $\mathbb{S}^3$.
Therefore, while it is possible to construct a duality between $\mathbb{E}^3$ or $\mathbb{S}^3$ and a 2D surface, it is not clear whether such a duality could be considered a bulk / boundary mapping.

\section{Conclusions} \label{discussion}

Even in the absence of a duality at the level of Hamiltonians, holographic quantum codes such as~\cite{Pastawski:2015} already provide a simple, tractable toy model of many of the interesting static features expected of the real AdS/CFT ``dictionary'', such as redundant encoding and complementary recovery of information on the boundary~\cite{Almheiri:2015,Harlow:2016}, entropic relations such as the Ryu-Takayanagi formula~\cite{Ryu:2006a,Ryu:2006}, and even toy models of (static) black holes satisfying the Beckenstein-Hawking bound~\cite{Beckenstein,hawking:1975}.
However, without a holographic mapping between local Hamiltonians, these toy models give more limited insight into the relationship between bulk and boundary energy scales -- a key aspect of AdS/CFT, where non-classical bulk spacetime geometries are believed to correspond to high-energy boundary state.
More importantly, HQECC alone gave no insight into how dynamics in the bulk is reflected in the boundary.

By extending the toy models of holographic duality to encompass Hamiltonians, we show one way to complete this ``dictionary''.
For example, it follows almost immediately from our construction that the toy models of static black holes proposed in~\cite{Pastawski:2015} do indeed correspond to high-energy states of the local boundary model, which moreover pick up their energy from a small number of high-energy terms in the boundary Hamiltonian.

More intriguingly, our construction allows these toy models to say something about how dynamics in the bulk is reflected in the boundary.
Even without writing down any specific local bulk Hamiltonian, the structure of the bulk/boundary mapping we construct implies that dynamics in the bulk is dual to boundary dynamics with some of the qualitative features of AdS/CFT duality. In particular we show that the formation of a (toy model) black hole in the bulk dynamics is dual to a boundary dynamics in which local excitations unitarily evolve to a non-local excitation that lives outside the code space.

On the other hand, our construction shows that \emph{any} local Hamiltonian in the bulk has a corresponding local boundary model.
This implies that, at least in these toy models, the holographic duality has little to do with quantum gravity per se, but is entirely a consequence of the hyperbolic geometry.

Our construction also inherits some of the drawbacks of the HaPPY code. In particular, the Ryu-Takayanagi formula is not obeyed exactly for arbitrary bulk regions; there exist certain pathological choices of boundary region, $A$, for which there are bulk operators that are not recoverable on $A$ nor on $A^c$ (violating complementary recovery). Like the HaPPY code, the tensor network cannot describe sub-AdS geometry as it is only defined at scales larger than the AdS radius. A number of holographic codes have been constructed which build on the HaPPY code and do not have these drawbacks. Notable examples include bidirectional holographic codes (BHC) composed of pluperfect tensors~\cite{pluperfect}, and random tensor network constructions~\cite{Hayden:2016}. It should be feasible to apply the framework developed in this paper to stabilizer BHCs or stabilizer random tensor networks~\cite{Nezami:2017}, to construct a toy model of holographic duality which remedies these defects.

Another way to complete the holographic `dictionary' was suggested in~\cite{osborne-17}, where it is argued that the dynamics for a particular holographic state should be the unitary representation of Thompson's group (a discrete analogue of the conformal group).\footnote{A holographic state is a holographic code with no bulk logical indices.}
While~\cite{osborne-17} concentrates on a particular holographic state, they discuss how to generalise their results.
The key advantage of the method in~\cite{osborne-17} is that it gives a boundary system which is conformally invariant, as would be expected in AdS/CFT.
However, the results in~\cite{osborne-17} apply to holographic states, not holographic codes.
A subspace of the boundary Hilbert space is identified as the bulk Hilbert space in~\cite{osborne-17}, but it is not clear that this is redundantly encoded in the boundary, as would be expected in AdS/CFT.
Even if one identified a good boundary Hilbert space and a symmetry subgroup to identify with time dynamics, there is no reason to expect the generators of this time dynamics will be local.
In our construction we have not attempted to include conformal invariance in the boundary theory (although it in certain cases it does exhibit `block translational invariance' - see \cref{appendix_ti} for details).
An interesting avenue of further research would be to look into combining the work done in this paper, with the methods in~\cite{osborne-17} to construct a bulk-boundary correspondence between Hamiltonians which has the error-correction properties of holographic codes, as well as a conformally invariant boundary.

The main limitation of our result is the usual one stemming from the use of perturbation gadgets: the coupling strengths $\alpha_{ij}$ in the boundary Hamiltonian $\Hbulk$ are very far from uniform.
Indeed, some coupling strengths will be $O(1)$ whilst others are $O\left(\Delta_L^{\poly(n\log(n))}\right)$, where $\Delta_L = \Omega\left(\norm{\Hbulk}^6\right)$. High-energy interactions on the boundary perhaps matter less here than in Hamiltonian complexity results, since the motivation for holographic duality is to model high-energy physics phenomena.
Nonetheless, it would be interesting to understand if a large range of interaction energy scales is a necessary feature of toy models of holographic duality, or an artefact of our proof techniques. Recent results in \cite{Woods:2019,Faist:2019} indicate that this feature may be inherent to any mapping to a local boundary Hamiltonian.

\section{Technical preliminaries} \label{prelim}

\subsection{Perfect tensors and pseudo-perfect tensors} \label{prelim_1}
Perfect tensors were first introduced in~\cite{Pastawski:2015}, where they were used in the construction of HQECC from a 2D bulk to a 1D boundary.

\begin{definition}[Perfect tensors, definition 2 from~\cite{Pastawski:2015}]
A $2m$-index tensor $T_{a_1a_2...a_{2m}}$ is a perfect tensor if, for any bipartition of its indices into a set $A$ and a complementary set $A^c$ with $|A| \leq |A^c|$, $T$ is proportional to an isometric tensor from $A$ to $A^c$.
\end{definition}
This definition is equivalent to requiring that the tensor is a unitary from any set of $m$ legs to the complementary set.

For one of the constructions in this work we introduce a generalisation of perfect tensors:
\begin{definition}[Pseudo-perfect tensors]
A $2m+1$-index tensor $T_{a_1a_2...a_{2m+1}}$ is a pseudo-perfect tensor if, for any bipartition of its indices into a set $A$ and a complementary set $A^c$ with $|A| < |A^c|$, $T$ is proportional to an isometric tensor from $A$ to $A^c$.
\end{definition}

The states described by (pseudo-)perfect tensors are absolutely maximally entangled (AME) (see \cref{appendix_1} for details). Furthermore, viewed as an isometry from $k$~indices to $n$~indices, a $t$-index (pseudo-)perfect tensor is the encoding isometry of a $[n,k,d]$ code, where $t = n+k$ and $d = \floor*{\frac{t}{2}}-k+1$ (see \cref{appendix_2} for proof).

\subsubsection{Stabilizer (pseudo-)perfect tensors}

For this work we will be specifically interested in stabilizer (pseudo-)perfect tensors:

\begin{definition}[Stabilizer (pseudo-)perfect tensors] \label{stab_tensor_def}
Stabilizer (pseudo-)perfect tensors describe stabilizer AME states.\footnote{See \cref{appendix_2} for definition of stabilizer state and stabilizer code.}
%A (pseudo-)perfect tensor is a stabilizer (pseudo-)perfect tensor if the AME state described by the tensor is a stabilizer state.\footnote{See \cref{appendix_2} for definition of stabilizer state and stabilizer code.}
\end{definition}

In \cref{appendix_3} we demonstrate that stabilizer (pseudo-)perfect tensors describe stabilizer QECC. Furthermore, they map Pauli rank one operators to Pauli rank one operators in a consistent basis.
% (see \cref{appendix_3} for details).

It is possible to construct a $t$-index (pseudo-)perfect stabilizer tensor for arbitrarily large $t$ by increasing the local Hilbert space dimension. Details of the construction are given in \cref{appendix_4}.

\subsection{Hyperbolic Coxeter groups} \label{prelim_4}
\subsubsection{Coxeter systems}

The HQECC presented in this paper are tensor networks embedded in tessellations of $\mathbb{H}^3$. We use Coxeter systems to analyse these tessellations.\footnote{An overview of hyperbolic Coxeter groups can be found at~\cite{dur-webpage}.}

 \begin{definition}[Coxeter system~\cite{Tits:1961}]
 Let $S = \{s_i\}_{i \in I}$, be a finite set. Let $M = \left(m_{i,j} \right)_{i,j \in I}$ be a matrix such that:
 \begin{itemize}
 \item $m_{ii} = 1$,   $ \forall i \in I$
 \item $m_{ij} = m_{ji} $,   $\forall i,j \in I, i \neq j$
 \item $m_{ij} \in \left( \mathbb{N} \setminus \{1\} \right) \cup \{\infty   \}$,   $ \forall i,j \in I, i \neq j$
 \end{itemize}
 $M$ is called the Coxeter matrix. The associated Coxeter group, $W$, is defined by the presentation:\footnote{A group presentation $\langle S\mid R \rangle$, where $S$ is a set of generators and $R$ is a set of relations between the generators, defines a group which is (informally) the largest group which is generated by $S$ and in which all the relations in $R$ hold.}
 \begin{equation}
 W = \langle S \mid \left(s_is_j \right)^{m_{ij}} = 1 \forall i,j \in I  \rangle
 \end{equation}
 The pair $(W,S)$ is called a Coxeter system.
 \end{definition}

To understand the connection between Coxeter systems and tesselations of hyperbolic space we need to introduce the notion of a Coxeter polytope.

\begin{definition}
A convex polytope in $\mathbb{X}^d = \mathbb{S}^d, \mathbb{E}^d$ or $\mathbb{H}^d$ is a convex intersection of a finite number of half spaces.
A Coxeter polytope $P \subseteq \mathbb{X}^d$  is a polytope with all dihedral angles integer submultiples of $\pi$.
\end{definition}

A Coxeter system can be associated to every Coxeter polytope. Let $(F_i)_{i \in I}$ be the facets of $P$, and if $F_i \cap F_j \neq \emptyset$ set $m_{ij} = \frac{\pi}{\alpha_{ij}}$, where $\alpha_{ij}$ is the dihedral angle between $F_i$ and $F_j$. Set $m_{ii} = 1$, and if  $F_i \cap F_j = \emptyset$ set $m_{ij} = \infty$. Let $s_i$ be the reflection in $F_i$. The Coxeter group with Coxeter matrix $(m_{ij})_{i,j \in I}$ is a discrete subgroup of $Isom(\mathbb{X}^d)$, generated by reflections in the facets of $P$, and $P$ tiles $\mathbb{X}^d$~\cite{Davis:2007}.

Coxeter systems can be represented by Coxeter diagrams, where a vertex is associated to every $s_i$ (or equivalently to every facet in the corresponding Coxeter polytope). Vertices are connected by edges in the following manner:
\begin{itemize}
\item If $m_{ij} = 2$ (i.e. facets $F_i$ and $F_j$ in the Coxeter polytope are orthogonal) there is no edge between the vertices representing $s_i$ and $s_j$
\item If $m_{ij} = 3$ (i.e. the dihedral angle between $F_i$ and $F_j$ is $\frac{\pi}{3}$ ) there is an unlabelled edge between vertices representing $s_i$ and $s_j$
\item If $m_{ij} \in \mathbb{N} \setminus \{1,2,3\}$ (i.e. the dihedral angle between $F_i$ and $F_j$ is $\frac{\pi}{m_{ij}}$) there is an edge labelled with $m_{ij}$ between vertices representing $s_i$ and $s_j$
\item If $m_{ij} = \infty$ (i.e. facets $F_i$ and $F_j$ in the Coxeter polytope diverge) there is a dashed edge between the vertices representing $s_i$ and $s_j$
\end{itemize}

A Coxeter group is irreducible if its Coxeter diagram is connected.

Faces of $P$ correspond to subsets of $S$ that generate finite Coxeter groups:\footnote{This does not apply to ideal vertices (vertices at the boundary of $\mathbb{X}^d$) however in this paper we are only concerned with compact polyhedra, which do not have vertices at infinity.}

\begin{lemma}[From~\cite{Vinberg:1985}] \label{Vinberg}
$f = \cup_{i \in I} F_i $ is a codimension $|I|$ face of  $P$ if and only if $\{s_i \mid i \in I\}$ generates a finite Coxeter group.
\end{lemma}

\subsubsection{Combinatorics of Coxeter groups} \label{combinatorics}

In this section we briefly introduce the notions which are used later in the paper.

Let $(W,S)$ be a Coxeter system. Every element $w \in W$ can be written as a product of generators:
\begin{equation}
w = s_1s_2...s_k \mbox{  for  } s_i \in S
\end{equation}
This description is not unique. We can define a length function with respect to the generating set $S$ such that $l_S(1) = 0$, and:
\begin{equation}
l_S(w) = \min\{l \in \mathbb{N} \mid s_1s_2...s_l = w \}
\end{equation}
An expression for $w$ with the minimum number of generators, $s_1s_2...s_{l_S(w)}$ is called a reduced word for $w$.

Coxeter groups satisfy the Deletion Condition:

\begin{definition}[Deletion Condition] \label{deletion}
Let $(W,S)$ be a pair where $W$ is a group and $S$ is a generating set for $W$ consisting entirely of elements of order two. We say that this pair satisfies the Deletion Condition if for any non reduced word $s_1...s_r$ over $S$ there are two indices $i$ and $j$ such that:
\begin{equation}
s_1...s_r = s_1...\hat{s}_i...\hat{s}_j...s_r
\end{equation}
where the carets indicate omission.
\end{definition}

The length function on Coxeter groups has a number of important properties:
\begin{enumerate}[(i)]
\item $l_S(ws) = l_S(w) \pm 1$ for all $s \in S$
\item $l_S(sw) = l_S(w) \pm 1$ for all $s \in S$
\item $l_S(w^{-1}) = l_S(w)$ for all $w \in W$
\item $|l_S(u) - l_S(w)| \leq l_S(uw) \leq l_S(u) + l_S(w)$ for all $u,w \in W$
\item $d(u,w) = l_S(u^{-1}w)$ for $u,w \in W $ is a metric on $W$ (referred to as the word metric)
\end{enumerate}

By conditions (i) and (ii), if we define the following sets:
\begin{equation}
\begin{split}
& \mathcal{D}_R(w) = \{s \in S \mid l_S(ws) = l_S(w) - 1 \}\\
& \mathcal{A}_R(w) = \{s \in S \mid l_S(ws) = l_S(w) + 1 \}\\
& \mathcal{D}_L(w) = \{s \in S \mid l_S(sw) = l_S(w) - 1 \}\\
& \mathcal{A}_L(w) = \{s \in S \mid l_S(sw) = l_S(w) + 1 \}\\
\end{split}
\end{equation}
then we have $\mathcal{D}_R(w) \cup \mathcal{A}_R(w) = \mathcal{D}_L(w) \cup \mathcal{A}_L(w) = S $ and $\mathcal{D}_R(w) \cap \mathcal{A}_R(w) = \mathcal{D}_L(w) \cap \mathcal{A}_L(w) = \{\}$. We refer to $\mathcal{D}_R(w)$ and $\mathcal{D}_L(w)$  ($\mathcal{A}_R(w), \mathcal{A}_L(w)$) as the right and left descent sets (ascent sets) of $w$ respectively.

\begin{lemma}[Corollary 2.18 from~\cite{Abramenko:2008}] \label{finite_descent}
For all $w \in W$, the Coxeter groups generated by $\mathcal{D}_R(w)$ and $ \mathcal{D}_L(w)$ are finite.\footnote{All subsets of $S$ generate a Coxeter group.}
\end{lemma}

The irreducible finite Coxeter groups are classified in \cref{finite}. A general Coxeter group is finite if and only if each connected component of the Coxeter graph generates a finite group.

\begin{table}
\begin{center}
\begin{tabular}{c|c}
Name & Coxeter diagram \\
\hline
$A_n$  ($n \geq 1$) &
\begin{tikzpicture}
\filldraw[black] (0,0) circle (2pt) node[anchor=north] {$1$};
\filldraw[black] (1,0) circle (2pt) node[anchor=north] {$2$};
\filldraw[black] (2,0) circle (2pt) node[anchor=north] {$3$};
\filldraw[black] (3,0) circle (2pt) node[anchor=north] {$n-1$};
\filldraw[black] (4,0) circle (2pt) node[anchor=north] {$n$};
\draw (0,0) -- (1,0);
\draw (2,0) -- (1,0);
\draw[dotted] (2,0) -- (3,0);
%\node at (2.5,0) {\ldots};
\draw (3,0) -- (4,0);
\end{tikzpicture} \\
$B_n = C_n$  ($n \geq 3$) & \begin{tikzpicture}
\filldraw[black] (0,0) circle (2pt) node[anchor=north] {$1$};
\filldraw[black] (1,0) circle (2pt) node[anchor=north] {$2$};
\filldraw[black] (2,0) circle (2pt) node[anchor=north] {$n-2$};
\filldraw[black] (3.5,0) circle (2pt) node[anchor=north] {$n-1$};
\filldraw[black] (5,0) circle (2pt) node[anchor=north] {$n$};
\draw (0,0) -- (1,0);
\draw[dotted] (2,0) -- (1,0);
\draw (2,0) -- (3.5,0);
\draw (3.5,0) --  (5,0) node[midway,above] {4};
\end{tikzpicture} \\

$D_n$  ($n \geq 4$) & \begin{tikzpicture}
\filldraw[black] (0,0) circle (2pt) node[anchor=north] {$1$};
\filldraw[black] (1,0) circle (2pt) node[anchor=north] {$2$};
\filldraw[black] (2,0) circle (2pt) node[anchor=north] {$n-3$};
\filldraw[black] (3.5,0) circle (2pt) node[anchor=north] {$n-2$};
 \filldraw[black] (3.5,1) circle (2pt) node[anchor=south] {$n-1$};
\filldraw[black] (5,0) circle (2pt) node[anchor=north] {$n$};
\draw (0,0) -- (1,0);
\draw[dotted] (2,0) -- (1,0);
\draw (2,0) -- (3.5,0);
\draw (3.5,0) --  (5,0);
\draw(3.5,0) -- (3.5,1);
\end{tikzpicture} \\

$E_6$  & \begin{tikzpicture}
\filldraw[black] (0,0) circle (2pt) node[anchor=north] {$1$};
\filldraw[black] (1,0) circle (2pt) node[anchor=north] {$3$};
\filldraw[black] (2,0) circle (2pt) node[anchor=north] {$4$};
\filldraw[black] (3,0) circle (2pt) node[anchor=north] {$5$};
 \filldraw[black] (2,1) circle (2pt) node[anchor=south] {$2$};
\filldraw[black] (4,0) circle (2pt) node[anchor=north] {$6$};
\draw (0,0) -- (1,0);
\draw (2,0) -- (1,0);
\draw (2,0) -- (3,0);
\draw (3,0) --  (4,0);
\draw(2,0) -- (2,1);
\end{tikzpicture} \\

$E_7$  & \begin{tikzpicture}
\filldraw[black] (0,0) circle (2pt) node[anchor=north] {$1$};
\filldraw[black] (1,0) circle (2pt) node[anchor=north] {$3$};
\filldraw[black] (2,0) circle (2pt) node[anchor=north] {$4$};
\filldraw[black] (3,0) circle (2pt) node[anchor=north] {$5$};
 \filldraw[black] (2,1) circle (2pt) node[anchor=south] {$2$};
\filldraw[black] (4,0) circle (2pt) node[anchor=north] {$6$};
\filldraw[black] (5,0) circle (2pt) node[anchor=north] {$7$};

\draw (0,0) -- (1,0);
\draw (2,0) -- (1,0);
\draw (2,0) -- (3,0);
\draw (3,0) --  (4,0);
\draw(2,0) -- (2,1);
\draw (5,0) --  (4,0);

\end{tikzpicture} \\

$E_8$  & \begin{tikzpicture}
\filldraw[black] (0,0) circle (2pt) node[anchor=north] {$1$};
\filldraw[black] (1,0) circle (2pt) node[anchor=north] {$3$};
\filldraw[black] (2,0) circle (2pt) node[anchor=north] {$4$};
\filldraw[black] (3,0) circle (2pt) node[anchor=north] {$5$};
 \filldraw[black] (2,1) circle (2pt) node[anchor=south] {$2$};
\filldraw[black] (4,0) circle (2pt) node[anchor=north] {$6$};
\filldraw[black] (5,0) circle (2pt) node[anchor=north] {$7$};
\filldraw[black] (6,0) circle (2pt) node[anchor=north] {$8$};

\draw (0,0) -- (1,0);
\draw (2,0) -- (1,0);
\draw (2,0) -- (3,0);
\draw (3,0) --  (4,0);
\draw(2,0) -- (2,1);
\draw (5,0) --  (4,0);
\draw (5,0) --  (6,0);

\end{tikzpicture} \\

$F_4$ & \begin{tikzpicture}
\filldraw[black] (0,0) circle (2pt) node[anchor=north] {$1$};
\filldraw[black] (1,0) circle (2pt) node[anchor=north] {$2$};
\filldraw[black] (2,0) circle (2pt) node[anchor=north] {$3$};
\filldraw[black] (3,0) circle (2pt) node[anchor=north] {$4$};
\draw (0,0) -- (1,0);
\draw (2,0) -- (1,0) node[midway,above] {4};
\draw (2,0) -- (3,0);
\end{tikzpicture} \\

$G_2$ & \begin{tikzpicture}
\filldraw[black] (0,0) circle (2pt) node[anchor=north] {};
\filldraw[black] (1,0) circle (2pt) node[anchor=north] {};
\draw (0,0) -- (1,0) node[midway,above] {6};
\end{tikzpicture} \\

$H_3$ & \begin{tikzpicture}
\filldraw[black] (0,0) circle (2pt) node[anchor=north] {$1$};
\filldraw[black] (1,0) circle (2pt) node[anchor=north] {$2$};
\filldraw[black] (2,0) circle (2pt) node[anchor=north] {$3$};
\draw (0,0) -- (1,0);
\draw (2,0) -- (1,0) node[midway,above] {5};
\end{tikzpicture} \\

$H_4$ & \begin{tikzpicture}
\filldraw[black] (0,0) circle (2pt) node[anchor=north] {$1$};
\filldraw[black] (1,0) circle (2pt) node[anchor=north] {$2$};
\filldraw[black] (2,0) circle (2pt) node[anchor=north] {$3$};
\filldraw[black] (3,0) circle (2pt) node[anchor=north] {$4$};
\draw (0,0) -- (1,0);
\draw (2,0) -- (1,0);
\draw (2,0) -- (3,0) node[midway,above] {5};
\end{tikzpicture} \\

$I^{(m)}_2$  ($m \geq 3$) & \begin{tikzpicture}
\filldraw[black] (0,0) circle (2pt) node[anchor=north] {};
\filldraw[black] (1,0) circle (2pt) node[anchor=north] {};
\draw (0,0) -- (1,0) node[midway,above] {$m$};
\end{tikzpicture} \\

\end{tabular}
\end{center}
\caption{Diagrams of irreducible finite Coxeter systems. Table reproduced from~\cite{Cohen}.}   \label{finite}

\end{table}

Finally, we note that if $s \in \mathcal{D}_R(w)$ ($s \in \mathcal{D}_L(w)$) there is a reduced word for $w$ that ends in $s$ (begins with $s$).

\subsubsection{Growth rates of Coxeter groups}

The growth series of a Coxeter group with respect to a set of generators $S$ is defined as:
\begin{equation}
f_S(x) = \sum_{w \in W} x^{l_S(w)} = 1 + Sx + ... = 1 + \sum_{i \geq 1} a_ix^i
\end{equation}
where $a_i$ is the number of $w \in W$ satisfying $l_S(w) = i$. The growth rate is given by:
\begin{equation}
\tau = \lim \sup_{n \rightarrow \infty} \sqrt[i]{a_i}
\end{equation}

Spherical and Euclidean Coxeter groups have growth rate 0 and 1 respectively. Hyperbolic Coxeter groups have $\tau > 1$.

%
%\subsubsection{Coxeter groups in $\mathbb{H}^3$}
%
%We will be interested in compact Coxeter polytopes in hyperbolic 3-space, $\mathbb{H}^3$, and their associated Coxeter systems. Andreev's theorem provides a classification of such polytopes:
%
%\begin{theorem}[Andreev's theorem, from~\cite{Andreev:1970}]\label{andreev}
%A compact acute-angled polytope $P$ of a given simple combinatorial type with given dihedral angles exists in $\mathbb{H}^3$ exists if and only if the following conditions holds:
%\begin{enumerate}
%\item if three faces $F_i,F_j,F_k$ pass through the same vertex then $\alpha_{ij}+\alpha_{jk}+\alpha_{ki} > \pi$;
%\item if the faces  $F_i,F_j,F_k$ form a three circuit then $\alpha_{ij}+\alpha_{jk}+\alpha_{ki} < \pi$, where faces form a $k$-circuit if $F_i$ shares a common edge with $F_{i+1}$ for $i = 1...k$ where $F_{k+1} = F_1$;
%\item if the faces  $F_i,F_j,F_k,F_l$ form a four circuit then $\alpha_{ij}+\alpha_{jk}+\alpha_{kl}+\alpha_{li} < 2\pi$;
%\item if $P$ is a triangular prism then at least one dihedral angle between a triangular side and a quadrilateral side is not equal to $\frac{\pi}{2}$
%\item if $P$ is a tetrahedon then the determinant of the Gram matrix should be negative
%\end{enumerate}
%\end{theorem}
%
%For a given Coxeter diagram it is straightforward to check whether the associated Coxeter polytope meets the conditions of Andreev's theorem, and therefore whether the Coxeter polytope tiles $\mathbb{H}^3$.

\subsection{Hamiltonian simulation} \label{prelim_5}
\cite{cubitt:2017} introduced a mathematical theory of analogue Hamiltonian simulation, characterising precisely when a Hamiltonian reproduces the same physics as another.
They then applied perturbation gadget techniques from Hamiltonian complexity theory to construct examples of universal local Hamiltonians, able to simulate any other model in this rigorous sense.
We make use of these techniques to construct holographic versions of Hamiltonian simulation which, together with the HQECC's, give the desired holographic dualities at the level of the Hamiltonians.

\subsubsection{Hamiltonian encodings}

In~\cite{cubitt:2017} it is shown that if an encoding $H' = \mathcal{E}(H)$ has the following three properties:
\begin{enumerate}
\item $\mathcal{E}(A) = \mathcal{E}(A)^\dagger$ for all $A \in \mbox{Herm}_n$
\item $\mbox{spec}(\mathcal{E}(A)) = \mbox{spec}(A)$ for all $A \in \mbox{Herm}_n$
\item $\mathcal{E}(pA + (1-p)B) = p\mathcal{E}(A)+(1-p)\mathcal{E}(B)$ for all $A, B \in \mbox{Herm}_n$ and all $p \in [0,1]$
\end{enumerate}
then it must be of the form:
\begin{equation}
\mathcal{E}(M) = U(M \otimes P + \overline{M} \otimes Q)U^\dagger
\end{equation}
for orthogonal projectors $P,Q$ such that $P+Q = \identity$ where $U$ is a unitary, and $\overline{M}$ denotes complex conjugation.
Furthermore, it is shown that under any such encoding, $H'$ will preserve the partition function, measurement outcomes and time evolution of $H$.

If the encoding $\mathcal{E}(H)$ only acts within a subspace of the simulator system $\mathcal{H}'$ then the unitary $U$ is replaced with an isometry $W$:
\begin{equation}
\mathcal{E}(M) = W(M \otimes P + \overline{M} \otimes Q)W^\dagger
\end{equation}
The stabilizer codes discussed in \cref{codes_section} are an example of subspace encodings with the particularly simple structure $\mathcal{E}(M) = WMW^\dagger$.

A local encoding is an encoding which maps local observables to local observables.
\begin{definition}[Local encoding into a subspace (this is a generalisation of definition 13 from~\cite{cubitt:2017}] Let $\mathcal{E}: \mathcal{B}(\otimes_{i=1}^n \mathcal{H}_i) \rightarrow \mathcal{B}(\otimes_{i=1}^{n'} \mathcal{H}_i')$ be a subspace encoding, and let $\{S_i'\}_{i=1}^{n}$ be subsets of $[n']$.
  We say that $\mathcal{E}$ is local with respect to $\{S_i'\}$ if for any operator $A_i \in \mbox{\emph{Herm}}(\mathcal{H}_i)$ there exists $A_i' \in \mbox{\emph{Herm}}(\otimes_{i=1}^{n'} \mathcal{H}_i')$ which acts non-trivially only on $\{S_i'\}$ such that:
  \begin{equation}
    \mathcal{E}(A_i \otimes \identity) = (A_i' \otimes \identity)\mathcal{E}(\identity)
  \end{equation}
\end{definition}

\subsubsection{Hamiltonian simulation}

If $H'$ perfectly simulates $H$ then it reproduces the physics of $H$ below some energy cut-off $\Delta$, where $\Delta$ can be made arbitrarily large.

\begin{definition}[Exact simulation, definition 20 from~\cite{cubitt:2017}] We say that $H'$ perfectly simulates $H$ below energy $\Delta$ if there is a local encoding $\mathcal{E}$ into the subspace $S_{\mathcal{E}}$ such that
  \begin{enumerate}[i.]
  \item $S_{\mathcal{E}} = S_{\leq\Delta(H')}$ (or equivalently $\mathcal{E}(\identity) = P_{\leq\Delta(H')}$)
  \item $H'|_{\leq \Delta} = \mathcal{E}(H)|_{S_\mathcal{E}}$
  \end{enumerate}
\end{definition}

If $\mathcal{E}(M) = WMW^\dagger$, where $W$ is the encoding isometry of some stabilizer code and $\Pi = WW^\dagger$ is the projector onto the code-space, then $H' = \mathcal{E}(H) + \Delta \Pi$ simulates $H$ below energy $\Delta$.

We can also consider the case where the the simulation is only approximate:

\begin{definition}[Approximate simulation, definition 23 from~\cite{cubitt:2017}]\label{approxsim}
  We say that $H'$ is a $(\Delta,\eta,\epsilon)$-simulation of $H$ if there exists a local encoding $\mathcal{E}(M) = W(M \otimes P + \overline{M}\otimes Q)W^\dagger$ such that:
  \begin{enumerate}[i.]
  \item There exists an encoding $\tilde{\mathcal{E}}(M)= \tilde{W}(M \otimes P + \overline{M}\otimes Q)\tilde{W}^\dagger$ such that $S_{\tilde{\mathcal{E}}} = S_{\leq \Delta(H')}$ and $||W - \tilde{W}||_{\infty} \leq \eta$;
  \item $|| H'_{\leq \Delta} - \tilde{\mathcal{E}}(H) ||_{\infty} \leq \epsilon$
  \end{enumerate}
\end{definition}
In~\cite{cubitt:2017} it is shown that approximate Hamiltonian simulation preserves important physical properties:
\begin{theorem}[From~\cite{cubitt:2017}] \label{physical-properties} Let $H$ act on $\left(\mathbb{C}^d\right)^{\otimes n}$.
  Let $H'$ acting on $\left(\mathbb{C}^{d'}\right)^{\otimes m}$ be a $(\Delta, \eta, \epsilon)$-simulation of $H$ with corresponding local encoding $\mathcal{E}(M) = W(M \otimes P + \overline{M} \otimes Q)W^\dagger$.
  Let $p = \rank(P)$ and $q = \rank(Q)$.
  Then:

  \begin{enumerate}[i.]
  \item Let $\lambda_i(H)$ (resp.\ $\lambda_i(H')$) be the $i^{\text{th}}$ smallest eigenvalue of $H$ (resp.\ $H'$), then for all $1 \leq i \leq d^n$, and all $(i-1)(p+q) \leq j \leq i (p+q)$, $|\lambda_i(H) - \lambda_j(H')| \leq \epsilon$.
  \item The relative error in the partition function evaluated at $\beta$ satisfies:
    \begin{equation}
      \frac{|\mathcal{Z}_{H'}(\beta) - (p+q)\mathcal{Z}_H(\beta) |}{(p+q)\mathcal{Z}_H(\beta)} \leq \frac{(d')^m e^{-\beta \Delta}}{(p+q)d^n e^{-\beta ||H||}} + (e^{\epsilon \beta} - 1)
    \end{equation}
  \item For any density matrix $\rho'$ in the encoded subspace so that $\mathcal{E}(\identity)\rho' = \rho'$:
    \begin{equation}
      ||e^{-iH't}\rho'e^{iH't} - e^{-i\mathcal{E}(H)t}\rho'e^{i\mathcal{E}(H)t}||_1 \leq 2\epsilon t + 4\eta
    \end{equation}
  \end{enumerate}
\end{theorem}

\subsubsection{Perturbative gadgets} \label{perturbative_simulations}

The following lemmas were shown in~\cite{Bravyi:2017}, and can be used to construct simulations perturbatively. See \cref{appendix_6} for further details.

Let $\mathcal{H}$ be a Hilbert space decomposed as $\mathcal{H} = \mathcal{H}_- \oplus \mathcal{H}_+$. Let $\Pi_\pm$ be the projectors onto $\mathcal{H}_\pm$. For arbitrary operator $M$ define $M_{++} = \Pi_+M\Pi_+$, $M_{--} = \Pi_-M\Pi_-$, $M_{+-} = \Pi_+M\Pi_-$, and $M_{-+} = \Pi_-M\Pi_+$.
Consider an unperturbed Hamiltonian $H = \Delta H_0$, where $H_0$ is block-diagonal with respect to the split $\mathcal{H} = \mathcal{H}_- \oplus \mathcal{H}_+$, $(H_0)_{--} = 0$, $\lambda_{\text{{min}}}\left((H_0)_{++} \right) \geq 1$.

\begin{lemma}[Second order simulation~\cite{Bravyi:2017}] \label{second_order} Let $V = H_1 + \Delta^{\frac{1}{2}}H_2$ be a perturbation acting on the same space as $H_0$ such that $\max(||H_1||,||H_2||) \leq \Lambda$; $H_1$ is block diagonal with respect to the split $\mathcal{H} = \mathcal{H}_- \oplus \mathcal{H}_+$ and $(H_2)_{--} = 0$.
  Suppose there exists an isometry $W$ such that $\text{\emph{Im}}(W) = \mathcal{H}_-$ and:
  \begin{equation} \label{second_order_eq}
    || W H_{\text{\emph{target}}} W^\dagger - (H_1)_{--} + (H_2)_{-+}H_0^{-1}(H_2)_{+-}||_{\infty} \leq \frac{\epsilon}{2}
  \end{equation} Then $\tilde{H} = H + V$ $(\frac{\Delta}{2}, \eta, \epsilon)$ simulates $H_{\text{\emph{target}}}$, provided that $\Delta \geq O(\frac{\Lambda^6}{\epsilon^2} + \frac{\Lambda^2}{\eta^2})$.
\end{lemma}

\begin{lemma}[Third order simulation~\cite{Bravyi:2017}] \label{third_order}
 Let $V = H_1 + \Delta^{\frac{1}{3}} H_1' + \Delta^{\frac{2}{3}}H_2$ be a perturbation acting on the same space as $H_0$ such that $\max(||H_1||,||H_1'||,||H_2||) \leq \Lambda$; $H_1$ and $H_1'$ are block diagonal with respect to the split $\mathcal{H} = \mathcal{H}_- \oplus \mathcal{H}_+$ and $(H_2)_{--} = 0$.
  Suppose there exists an isometry $W$ such that $\text{\emph{Im}}(W) = \mathcal{H}_-$ and:
  \begin{equation} \label{third_order_eq1}
    || W H_{\text{\emph{target}}} W^\dagger - (H_1)_{--} + (H_2)_{-+}H_0^{-1}(H_2)_{++}H_0^{-1}(H_2)_{+-}||_{\infty} \leq \frac{\epsilon}{2}
  \end{equation} and also that:
  \begin{equation} \label{third_order_eq2}
    (H_1')_{--} = (H_2)_{-+}H_0^{-1}(H_2)_{+-}
  \end{equation} Then $\tilde{H} = H + V$ $(\frac{\Delta}{2}, \eta, \epsilon)$ simulates $H_{\text{\emph{target}}}$, provided that $\Delta \geq O(\frac{\Lambda^{12}}{\epsilon^3} + \frac{\Lambda^3}{\eta^3})$.
\end{lemma}

%\subsubsection{Perturbation gadgets}

We derive a number of qudit perturbation gadgets (based on qubit perturbation gadgets from~\cite{oliveira:2005}) for use in our construction. Using \cref{second_order,third_order} we can show that all of these gadgets are simulations (for appropriate choices of $\Delta$).
Details are given in \cref{appendix_6}.
The main results are collected here:

\noindent\paragraph{Qudit subdivision gadget}\\
The qudit subdivision gadget is used to simulate a $k$-local interaction by two $\ceil*{\frac{k}{2}}+1$-local interactions, by introducing a mediator qudit.
The resulting interaction pattern is shown in \cref{subdivision-fig}.

\begin{figure}[H] \centering
  \begin{tikzpicture} \filldraw[black] (0,0) circle (2pt) node[below]{$A$}; \filldraw[black] (3,0) circle (2pt) node[below]{$B$}; \draw (0,0) -- (3,0) node[midway, above]{$M_1$}; \filldraw[black] (5,0) circle (2pt) node[below]{A}; \draw (7,0) circle (2pt) node[below]{$w$}; \filldraw[black] (9,0) circle (2pt) node[below]{$B$}; \draw (5,0) -- (6.9,0) node[midway, above]{$M_2$}; \draw (7.1,0) -- (9,0) node[midway, above]{$M_3$};
  \end{tikzpicture}
  \caption{Subdivision gadget.
    The $k$-local interaction on the left is simulated by the two $\ceil*{\frac{k}{2}}+1$-local interactions on the right by introducing a mediator qudit, $w$.
    The interactions are given by $M_1 = P_{A}\otimes P_{B} + P_{A}^\dagger\otimes P_{B}^\dagger$, $M_2 = P_{A}\otimes X_w + P_A^\dagger \otimes X_w^\dagger$, and $M_3 = P_B\otimes X_w^\dagger + P_B^\dagger\otimes X_w$.}
  \label{subdivision-fig}
\end{figure}
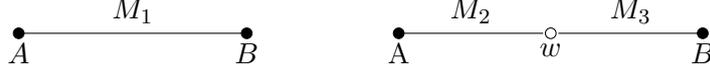

\noindent\paragraph{Qudit 3-2 gadget}\\
The 3-2 gadget is used to simulate a 3-local interaction with six 2-local interactions, by introducing a mediator qudit.
The resulting interaction pattern is shown in \cref{three-two-fig}.

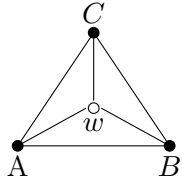
\begin{figure}[H] \centering
  \begin{tikzpicture} \filldraw[black] (0,-0.5) circle (2pt) node[below]{A}; \draw (1,0) circle (2pt) node[below]{$w$}; \filldraw[black] (2,-0.50) circle (2pt) node[below]{$B$}; \draw (0,-0.5) -- (0.9,0); \draw (1.1,0) -- (2,-0.5) ; \filldraw[black] (1,1) circle (2pt) node[above]{$C$}; \draw (0,-0.5)--(1,1); \draw (2,-0.5)--(1,1); \draw (1,1) -- (1,0.1); \draw (0,-0.5) -- (2,-0.5);
  \end{tikzpicture}
  \caption{3-2 gadget: The three body interaction between $A$, $B$ and $C$ $\left(P_A\otimes P_B\otimes P_C + P_A^\dagger \otimes P_B^\dagger \otimes P_C^\dagger \right)$ is simulated by the interaction pattern shown in the figure.}
  \label{three-two-fig}
\end{figure}

\noindent\paragraph{Qudit crossing gadget}\\
The crossing gadget is used to remove crossings in an interaction graph by introducing a mediator qudit.
The resulting interaction pattern is shown in \cref{crossing-fig}.

\begin{figure}[H] \centering
  \begin{tikzpicture} \filldraw[black] (0.5,-0.5) circle (2pt) node[below]{$A$}; \filldraw[black] (2.5,-0.5) circle (2pt) node[below]{$B$}; \draw (0.5,-0.5) -- (2.5,-2.5); \draw (2.5,-0.5) -- (0.5,-2.5); \filldraw[black] (0.5,-2.5) circle (2pt) node[below]{$C$}; \filldraw[black] (2.5,-2.5) circle (2pt) node[below]{$D$}; \filldraw[black] (5.5,-0.5) circle (2pt) node[above]{$A$}; \filldraw[black] (7.5,-0.5) circle (2pt) node[above]{$B$}; \draw (5.5,-0.5) -- (6.45,-1.45) ; \draw (7.5,-0.5) -- (6.55,-1.45); \draw (6.55,-1.55) -- (7.5,-2.5); \draw (6.45,-1.55) -- (5.5,-2.5); \draw (5.5,-0.5) -- (7.5,-0.5); \draw (5.5,-0.5) -- (5.5,-2.5); \draw (5.5,-2.5) -- (7.5,-2.5); \draw (7.5,-0.5) -- (7.5,-2.5); \filldraw[black] (5.5,-2.5) circle (2pt) node[below]{$C$}; \filldraw[black] (7.5,-2.5) circle (2pt) node[below]{$D$}; \draw (6.5,-1.5) circle(2pt) node[left]{$w$};
  \end{tikzpicture}
  \caption{Crossing gadget.
    The interaction pattern on the left is simulated by the interaction pattern on the right.}
  \label{crossing-fig}
\end{figure}
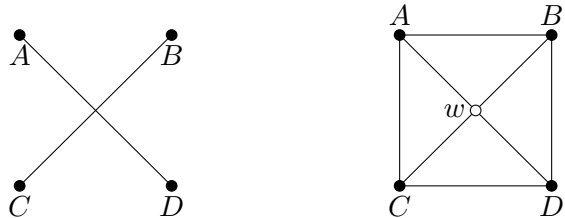

\noindent\paragraph{Qudit fork gadget}\\
The fork gadget is used to reduce the degree of a vertex in the interaction graph by introducing a mediator qudit.
The resulting interaction pattern is shown in \cref{fork-fig}.

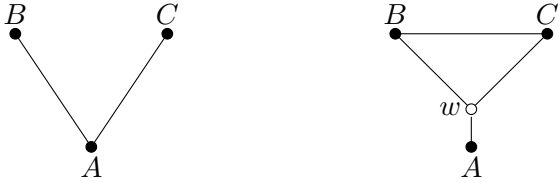
\begin{figure}[H] \centering
  \begin{tikzpicture} \filldraw[black] (1.5,0) circle (2pt) node[below]{$A$}; \filldraw[black] (0.5,1.5) circle (2pt) node[above]{$B$}; \filldraw[black] (2.5,1.5) circle (2pt) node[above]{$C$}; \draw (1.5,0) -- (0.5,1.5); \draw (1.5,0) -- (2.5,1.5); \filldraw[black] (6.5,0) circle (2pt) node[below]{$A$}; \filldraw[black] (5.5,1.5) circle (2pt) node[above]{$B$}; \filldraw[black] (7.5,1.5) circle (2pt) node[above]{$C$}; \draw (6.5,0.5) circle (2pt) node[left]{$w$}; \draw (6.5,0) -- (6.5,0.4); \draw (6.55,0.55) -- (7.5,1.5); \draw (6.45,0.55) -- (5.5,1.5); \draw (7.5,1.5) -- (5.5,1.5);
  \end{tikzpicture}
  \caption{Fork gadget.
    The interaction pattern on the left is simulated by the interaction pattern on the right.}
  \label{fork-fig}
\end{figure}

\noindent\paragraph{Qudit triangle gadget}\\
The qudit triangle gadget is formed by first applying the qudit subdivision gadget, then the qudit fork gadget, in the same way as it is formed for qubits in~\cite{oliveira:2005}.

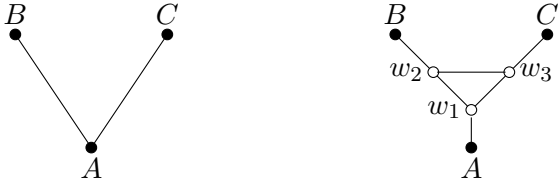
\begin{figure}[H] \centering
  \begin{tikzpicture} \filldraw[black] (1.5,0) circle (2pt) node[below]{$A$}; \filldraw[black] (0.5,1.5) circle (2pt) node[above]{$B$}; \filldraw[black] (2.5,1.5) circle (2pt) node[above]{$C$}; \draw (1.5,0) -- (0.5,1.5); \draw (1.5,0) -- (2.5,1.5); \filldraw[black] (6.5,0) circle (2pt) node[below]{$A$}; \filldraw[black] (5.5,1.5) circle (2pt) node[above]{$B$}; \filldraw[black] (7.5,1.5) circle (2pt) node[above]{$C$}; \draw (6.5,0.5) circle (2pt) node[left]{$w_1$}; \draw (6.5,0) -- (6.5,0.4); \draw (6.55,0.55) -- (6.95,0.95); \draw (6.45,0.55) -- (6.05,0.95); \draw (7.5,1.5) -- (7.05,1.05); \draw (5.5,1.5) -- (5.95,1.05); \draw (6,1) circle (2pt) node[left]{$w_2$}; \draw (7,1) circle (2pt) node[right]{$w_3$}; \draw (6.05,1) -- (6.95,1);
  \end{tikzpicture}
  \caption{Triangle gadget.
    The interaction pattern on the left is simulated by the interaction pattern on the right by first applying the subdivision gadget to edges $AB$ and $AC$, and then applying the fork gadget to qudit $A$.}
\end{figure}

The 3-2 gadget is a third order simulation. The other gadgets are second order simulations.

In~\cite{oliveira:2005} it is demonstrated that the qubit perturbation gadgets can be used at many places in an interaction graph in parallel, and that they do not interact with each other.
The same arguments follow for the qudit perturbation gadgets introduced here.
%
%For each of the perturbation gadgets, the resulting Hamiltonian is of the form:
%
%\begin{equation}
%  \tilde{H} = \Delta \Pi_+ + V
%\end{equation} where $\Pi_+$ is a projector which acts only on the mediator qudit, and $\Delta$ is some large parameter (the scaling of $\Delta$ depends on whether the resulting simulation is second order or third order, see \cref{perturbative_simulations}).

\section{Full technical details} \label{technical}

\subsection{General construction} \label{general_construction}

In this section we demonstrate the general procedure for constructing a HQECC using Coxeter groups and (pseudo-)perfect stabilizer tensors with particular properties.
In particular in \cref{main_theorem} we prove our main result: a full holographic duality between quantum many-body models in 3D hyperbolic space and models living on its 2D boundary.
In \cref{example_1,example_2} we construct two examples of sets of Coxeter groups and tensors which have the required properties.

\subsubsection{Notation}

Let $(W,S)$ be a Coxeter system with Coxeter polytope $P \subseteq \hyper$.
$F_a$ denotes the face of $P$ corresponding to the generator $s_a \in S$.
$E_{ab}$ denotes the edge of $P$ between $F_a$ and $F_b$.

$P^{(w)}$ denotes the polyhedral cell in the tessellation of $\hyper$ which corresponds to element $w$ of the Coxeter group.
Similarly $F_a^{(w)}$ and $E_{ab}^{(w)}$ refer to faces / edges of $P^{(w)}$.
$F_a^{A}$ and $E_{ab}^{A}$ refer to specific faces / edges in the tessellation of $\hyper$ which are shared by the polyhedral cells associated to the sets of elements $A \subseteq W$.

A bulk qudit which is associated to the polyhedral cell $P^{(w)}$ will be labelled by $q^{(w)}$.
A boundary qudit which is associated to the uncontracted tensor index through the face $F_a^{(w)}$ will be labelled by $q_a^{(w)}$.

\subsubsection{Holographic quantum error correcting codes}

The procedure we use to construct the tensor network is based on that in~\cite{Pastawski:2015}, where perfect tensors are embedded in tessellations of $\mathbb{H}^2$.
We take a Coxeter system $(W,S)$ with Coxeter polytope, $P \subseteq \mathbb{H}^3$ where $|S| = t-1$, so $P$ has $t-1$ faces.
Take the tessellation of $\mathbb{H}^3$ by $P$, and embed a (pseudo-)perfect tensor, $T$, with $t$ legs in each polyhedral cell.
$t-1$ legs of each tensor are contracted with legs of neighbouring tensors at shared faces of the polyhedra, and a logical, or input, qudit for the tensor network is associated with the uncontracted tensor leg in each polyhedral cell.
Cut off the tessellation at some radius $R$, and the uncontracted tensor legs on the boundary are the physical qudits of the tensor network.

A HQECC is defined as a tensor network composed of (pseudo)-perfect tensors which gives rise to an isometric map from bulk legs to boundary legs~\cite{Pastawski:2015}.
This is equivalent to requiring that the number of output indices from every tensor is greater than or equal to the number of input indices, where the input indices are the indices coming from the previous layer of the tessellation plus the logical index.

We are working in negatively curved geometry, so a majority of the tensors will have more output indices than input indices, but this doesn't guarantee it is true for every tensor.
%on average the number of output legs will be greater than the number of input legs, but this doesn't guarantee that every tensor will have more output legs then input legs.
For example, consider the triangulation of $\mathbb{H}^2$ with Schl\"{a}fli symbol $\{3,8\}$ (\cref{eight_three}).
This is the tiling which corresponds to the Coxeter diagram shown in \cref{eight_three_coxeter}.

It can be seen that there are triangular cells in the tessellation which share edges with two triangles from the previous layer, and only one in the subsequent layer.
If we put a four-index perfect tensor in each cell of this tessellation, then there would be some tensors with three input legs, and only one output leg.
These tensors would not be isometries, so it is not obvious that the overall tensor network would be an isometry.
In order to ensure that the tensor network is a HQECC we derive a condition to enforce that every tensor has at least as many output indices as input indices.
This is a sufficient condition for the tensor network to be a HQECC, but it may not be necessary.

\begin{figure} \centering \includegraphics[width=0.5\textwidth]{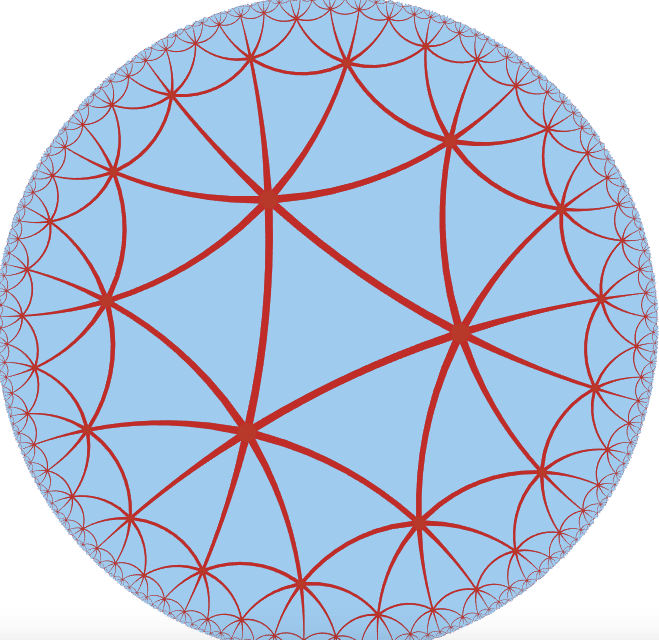}
  \caption{The triangulation of $\mathbb{H}^2$ with Schl\"{a}fli symbol $\{3,8\}$.
    There are triangular cells in the tessellation which share edges with two triangles from the previous layer, and only one in the subsequent layer.
    Figure produced via the software~\cite{Kaleidotile}.}
  \label{eight_three}
\end{figure}

\begin{figure} \centering
  \begin{tikzpicture} \draw(0,0) -- (1,1.5) node[midway,above]{4}; \draw(1,1.5) -- (2,0) node[midway,above]{4}; \draw(0,0)--(2,0) node[midway,below]{4}; \filldraw[black] (1,1.5) circle (2pt) node[anchor=west]{}; \filldraw[black] (0,0) circle (2pt) node[anchor=west]{}; \filldraw[black] (2,0) circle (2pt) node[anchor=west]{};
  \end{tikzpicture}
  \caption{The Coxeter diagram for the triangulation of $\mathbb{H}^2$ with Schl\"{a}fli symbol $\{3,8\}$} \label{eight_three_coxeter}
\end{figure}
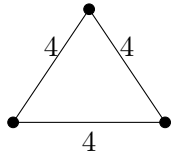

\begin{theorem} \label{isometry_thm} Consider a tensor network constructed as above, defined by Coxeter system $(W,S)$ and perfect tensor $T$ with $t$ indices.
  Define $\mathcal{F} = \{J \subseteq S \mid W_J \text{ is finite} \}$.
  The tensor network will be a HQECC, if $\forall J \in \mathcal{F}: |J| \leq \floor*{\frac{t-2}{2}}$.
\end{theorem}

\begin{proof} If we order the tensors into layers labelled by the value of the length function $l_S(w)$ at $P^{(w)}$, and include the uncontracted tensor leg in each polyhedral cell as an input leg, then the number of input legs for a tensor embedded in $P^{(w)}$ is $\mathcal{D}_R(w) + 1$.
  By \cref{finite_descent}, $\mathcal{D}_R(w) \in \mathcal{F}$.
  Therefore the maximum number of input legs to any tensor in the tensor network is $\max(|J| \mid J \in \mathcal{F}) + 1$.
  We therefore require $\max(|J| \mid J \in \mathcal{F}) + 1 \leq \floor*{\frac{t}{2}}$.
\end{proof}

\cref{isometry_thm} gives a sufficient condition for every tensor in the tensor network to have at least as many output indices as input indices.

The requirements of \cref{isometry_thm} dictate that we will not be able to use qubit stabilizer tensors to construct HQECC in dimensions greater than two.
To see this recall that \cref{Vinberg} stated that $\{s_i \mid i \in I\}$ generates a finite Coxeter group if and only if $f = \cup_{i \in I} F_i $ is a codimension $|I|$ face of $P$.
In dimension $d$ there will exist codimension $d$ faces, so $\max(|J| \mid J \in \mathcal{F}) \geq d$.
We therefore require that $\floor*{\frac{t}{2}} \geq d+1$.
For $d \geq 3$ this enforces $t \geq 8$, and there are no qubit perfect tensors with $t > 6$~\cite{Gour:2010,Rains:1999,Huber:2017}.

HQECC which are constructed in this way inherit all the properties of the 2-dimensional HQECC constructed in~\cite{Pastawski:2015}.

We call HQECC constructed in this way from Coxeter honeycombings and perfect tensors ``Coxeter HQECCs''.

\subsubsection{Surface of the HQECC} Define the boundary of the HQECC as the faces in the tessellation which correspond to the uncontracted tensor legs.
More precisely:
\begin{definition} The boundary, M, of a Coxeter HQECC of radius $R$ is given by:
  \begin{equation}
    M = \bigcup_{F_a^{(w)}\in\mathcal{M}} F_{a}^{(w)}
  \end{equation} where $\mathcal{M} = \{ F_{a}^{(w)} \mid l_S(w) = R, s_a \in \mathcal{A}_R(w) \}$.
\end{definition}

The boundary of hyperbolic $n$-space is an $n-1$ dimensional sphere.
For our HQECC we are cutting off the tessellation of $\mathbb{H}^3$ at some finite radius $R$, but it is still possible to demonstrate that the boundary is homeomorphic to a 2-sphere.

In order to reason about the boundary we need two lemmas about edges in the tessellation of $\hyper$ by $P$:

\begin{lemma} \label{edge-lemma-1} Consider an edge, $E_{ab}^{A}$, in the tessellation of $\hyper$ by a Coxeter polytope, $P$.
  If $l_S(w_1) = l_S(w_2) = L$ for distinct $w_1,w_2\in A$, then $\mathcal{D}_R(w_1)$ and $\mathcal{D}_R(w_2)$ contain at least one of $s_a$ or $s_b$.
\end{lemma}

\begin{proof} Recall that an edge $E_{ab}$ corresponds to the finite Coxeter subgroup generated by $s_a$ and $s_b$: $\langle s_a,s_b\rangle = \{\langle s_a,s_b\rangle^x | x\in[0,m_{ab})\}$, where $\langle s_a,s_b\rangle^x$ denotes a string of alternating $s_a$ and $s_b$ of length $x$.
  $A$ is set of Coxeter group elements corresponding to the polyhedra that meet at the common edge $E_{ab}^A$, so $A = \{ws | s\in \langle s_a,s_b\rangle\}$ for any fixed element $w\in A$.
  Therefore, we have that $w_2 = w_1\langle s_a,s_b\rangle^x$ for some $x\in[1,m_{ab})$; we take $x$ to be the minimum such value.
  Since $l_S(w_1\langle s_a,s_b\rangle^x) = l_S(w_2) + x > l_S(w_2)$, the deletion condition (\cref{deletion}) implies that there are two generators in the word $w_1\langle s_a,s_b\rangle^x$ which we can delete to get a shorter word for $w_2$.

  By minimality of $x$, they cannot both be deleted from the $\langle s_a,s_b\rangle^x$ part of this word.
  If they were both deleted from the $w_1$ part, so that $w_2 = s_1\dots \hat{s}_i\dots \hat{s}_j\dots s_{L}\langle s_as_b\rangle^x$, we would have $w_1 = w_2 (\langle s_a,s_b\rangle^x)^{-1} = w_2 \langle s_b,s_a\rangle^x = s_1\dots\hat{s}_i\dots\hat{s}_j\dots s_L$ which has length $L-2$, contradicting $l_S(w_1) = L$.
  Therefore, one generator must be deleted from the $w_1$ part, the other from $\langle s_a,s_b\rangle^x$.
  Thus $w_2 = s_1\dots\hat{s}_i\dots s_r\langle s_a,s_b\rangle^{x-1}$.

  This word for $w_2$ has length $L+x-2$.
  By the deletion condition, we must be able to delete a further $x-2$ generators to reach a reduced word for $w_2$.
  But $\langle s_a,s_b\rangle^{x-1}$ contains $x-1$ generators, so at least one of these must remain.
  Thus either $w_2 = u s_a$ or $w_w = u s_b$ for some $u \in A$ of length $l_S(u) = L-1$.
  Hence at least one of $s_a$ or $s_b$ is in $\mathcal{D}_R(w_2)$.

  The $w_1$ case follows by an analogous argument.
\end{proof}

\begin{lemma} \label{edge_lemma} Consider an edge, $E_{ab}^{A}$, in the tessellation of $\hyper$ by a Coxeter polytope, $P$.
  The set of elements $A$ associated with the polyhedral cells that share the edge $E_{ab}^{A}$ has the following properties:
  \begin{enumerate}[(i)]
  \item There is a unique minimum length element $w_{\min} \in A$ which has length, $l_S(w_{\min}) = r_{\min}$.
    \label{edge_lemma:i}
  \item For $0 \leq x < m_{ab}$, $l_S(\wmin \langle s_a,s_b\rangle^{x+1}) = l_S(\wmin \langle s_a,s_b\rangle^{x})+1$.
    \label{edge_lemma:ii}
  \item For $m_{ab} \leq x < 2m_{ab}$, $l_S(\wmin \langle s_a,s_b\rangle^{x+1}) = l_S(\wmin \langle s_a,s_b\rangle^{x})-1$.
    \label{edge_lemma:iii}
  \item There is a unique maximum length element $w_{\max} \in A$ which has length $l_S(w_{\max}) = r_{\min} + m_{ab}$.
    \label{edge_lemma:iv}
  \item For $r_{\min} < i < r_{\min} + m_{ab}$ there are exactly two elements $w_i,w_i' \in A$ which satisfy $l_S(w_i) = l_S(w_i') = i$.
    \label{edge_lemma:v}
  \end{enumerate} where $\langle s_a,s_b\rangle^x$ denotes a string of alternating $s_a$ and $s_b$ of length $x$ (i.e.
  $\langle s_a,s_b\rangle^3 = s_as_bs_a$).
\end{lemma}

\begin{proof}[Proof of \cref{edge_lemma}] \labelcref{edge_lemma:i}.
  Assume there are two minimum length elements in $A$, $w_{\min}$ and $w_{\min}'$ such that $l_S(w_{\min})= l_S(w_{\min}') = r_{\min}$.
  By \cref{edge-lemma-1} either $s_a$ or $s_b$ is in the descent set of $w_{\min}$ and $w_{\min}'$.
  This implies that there is at least one element in $A$ with length $\rmin - 1$, contradicting our assumption.% that $w_{\min}$ and $w_{\min}'$ are the minimum length elements in $A$.

  \labelcref{edge_lemma:ii}.
  Assume there is some $x < m_{ab}$ such that\linebreak[2] $l_S(\wmin \langle s_a,s_b\rangle^{x+1}) = l_S(\wmin\langle s_a,s_b\rangle^{x})-1 = L$.
  If we let $u = \wmin \langle s_a,s_b\rangle^{x+1}$ and assume (wlog) that $x$ is even, it follows that $s_a \in \mathcal{A}_R(u)$.
  Note that $l_S(\wmin) = \rmin < L$, $l_S(\wmin\langle s_a,s_b\rangle^x) = L+1$ and $l_S(\wmin\langle s_a,s_b\rangle^{x+1}) = L$.
  But each generator $s_a$ or $s_b$ that we multiply $\wmin$ by can only change the length by $\pm 1$.
  So $u$ is not the only element of length $L$ in $A$.
  By \cref{edge-lemma-1} this implies that at least one of $s_a$ or $s_b$ must be in the descent set of $u$.
  Therefore $s_b \in \mathcal{D}_R(u)$.

  If we let $v = us_b$ then $s_b \in \mathcal{A}_R(v)$.
  By a similar argument, $s_a \in \mathcal{D}_R(v)$.
  If we continue this argument we find that the length of the element $\wmin\langle s_a,s_b\rangle^{2x}$ is $\rmin$, which is not possible as $\wmin$ is the unique element of $A$ with length $\rmin$, and by assumption $x < m_{ab}$ so $\langle s_a,s_b\rangle^{2x} = (s_as_b)^x \neq I$ by definition of $m_{ab}$.
  Therefore there is no $x < m_{ab}$ such that $l_S(\wmin \langle s_a,s_b\rangle^{x+1}) = l_S(\wmin\langle s_a,s_b\rangle^{x})-1 = L$.

  \labelcref{edge_lemma:iii}.
  From \labelcref{edge_lemma:ii} it follows that $l_S(\wmin\langle s_a,s_b\rangle^{m_{ab}}) = \rmin + m_{ab}$.
  We have that $\langle s_a,s_b\rangle^{2m_{ab}} = (s_as_b)^{m_{ab}} = I$, thus $l_S(\wmin\langle s_a,s_b\rangle^{2m_{ab}}) = l_S(\wmin) = \rmin$.
  As each generator can only change the length of an element by $\pm 1$, for $m_{ab} < x \leq 2m_{ab}$ we must have that $l_S(\wmin \langle s_a,s_b\rangle^{x+1}) = l_S(\wmin \langle s_a,s_b\rangle^{x})-1$.

  \labelcref{edge_lemma:iv,edge_lemma:v} follow at once from points \labelcref{edge_lemma:ii,edge_lemma:iii}.
\end{proof}

An example of the set of polyhedra associated with an edge $E_{ab}$ where $m_{ab} = 4$ is shown in \cref{edge-pic}.
\cref{edge_lemma} ensures that the lengths associated to the polyhedra around any edge follow the same pattern.

We can now consider the boundary of the HQECC.

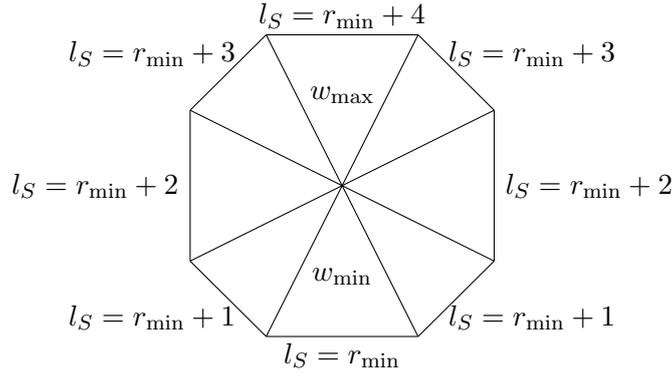
\begin{figure} \centering
  \begin{tikzpicture} \draw(0,0) -- (1,-2); \draw (1,-2) -- (-1,-2); \draw(0,0) -- (-1,-2); \draw(0,0) -- (1,2); \draw (1,2) -- (-1,2); \draw (1,-2) -- (2,-1); \draw (1,2) -- (2,1); \draw (-1,-2) -- (-2,-1); \draw (-1,2) -- (-2,1); \draw(2,1) -- (2,-1); \draw(-2,1) -- (-2,-1); \draw(0,0) -- (-1,2); \draw(0,0) -- (2,-1); \draw(0,0) -- (-2,-1); \draw(0,0) -- (2,1); \draw(0,0) -- (-2,1); \node at (0,-2.25) {$l_S = \rmin$}; \node at (-3.25,0) {$l_S = \rmin + 2$}; \node at (-2.5,-1.75) {$l_S = \rmin + 1$}; \node at (-2.5,1.75) {$l_S = \rmin + 3$}; \node at (0,2.25) {$l_S = \rmin + 4$}; \node at (2.5,1.75) {$l_S = \rmin + 3$}; \node at (3.25,0) {$l_S = \rmin + 2$}; \node at (2.5,-1.75) {$l_S = \rmin + 1$}; \node at (0,-1.2) {$\wmin$}; \node at (0,1.2) {$w_{\max}$};
  \end{tikzpicture}
  \caption{A cross section view of the polyhedral cells in the tessellation which meet around a common edge, $E_{ab}$, for $m_{ab} = 4$.}
  \label{edge-pic}
\end{figure}

\begin{lemma}\label{surface} The boundary $M$ of a Coxeter HQECC in $\mathbb{H}^3$ is a surface (a 2D manifold).
\end{lemma}
\begin{proof} Within the faces that make up the boundary $M$ is clearly locally Euclidean, and the same will be true at the edges where the faces meet provided no more than two faces meet at an edge.
  Point~(v) from \cref{edge_lemma} shows this is indeed the case.
\end{proof}

\begin{lemma} \label{closed} The boundary manifold, ${M}$, of a Coxeter HQECC in $\mathbb{H}^3$ is closed - i.e.
  compact and with no boundary.
\end{lemma}
\begin{proof} Assume $M$ has a boundary.
  This implies that $\exists w \in W$ such that $l_S(w) = R$ where $F_a^{(w)} \in {M}$ and $E_{a,b}^{(w)} \in \partial {M}$.

  Since $E_{a,b}$ is an edge of $P$ we must have that $\{s_a,s_b\} \in \mathcal{F}$ where $\mathcal{F} = \{ J \subseteq S | W_J \mbox{ is finite}\}$.
  This implies that $\exists m_{ab} \in \mathbb{N} \setminus \{1\}$ such that $(s_as_b)^{m_{ab}} = (s_bs_a)^{m_{ab}} = I$.

  $E_{a,b}^{(w)} \in \partial {M}$ implies that $F_{b}^{(w)} \notin {M}$, and therefore $s_b \in \mathcal{D}_R(w)$.
  This gives $w = us_b$ where $u \in W$ and $l_S(u) = R-1$.
  We also have that $l_S(ws_a) = l_S(us_bs_a) = R+1$ (because by assumption $s_a \in \mathcal{A}_R(w)$).

  Putting everything together we find that:
  \begin{equation}
    l_S(u) = R-1
  \end{equation}
  \begin{equation}
    l_S(us_bs_a) = R+1
  \end{equation}
  \begin{equation}
    l\left[u (s_bs_a)^{m_{ab}}  \right] = l_S(u) = R-1
  \end{equation} Therefore, at least one of the following must be true:
  \begin{enumerate}
  \item $\exists x$ such that $1 \leq x < m_{ab}$ where $l[u (s_bs_a)^x] = R+1$ and $l[u (s_bs_a)^xs_b] = R$
  \item $\exists x$ such that $1 \leq x < m_{ab}$ where $l[u (s_bs_a)^xs_b] = R+1$ and $l[u (s_bs_a)^{x+1}] = R$
  \end{enumerate}

  The second case would imply that $l[u (s_bs_a)^xs_b] = R+1 = l_S(us_bs_a)$ but this cannot occur as it is not possible for two elements of the Coxeter group with the same word length to be related by an odd number of generators.
If the first case occurs then $F_v^{(b)} \in {M}$ for $v = u (s_bs_a)^xs_b$ and shares edge $E_{ab}$ with $F_w^{(a)}$, so $E_{ab}^{(w)} \notin \partial {M}$.

  Therefore ${M}$ does not have a boundary.
  The boundary of every polyhedron face is included in $M$ so $M$ is compact.
\end{proof}

We now prove that the boundary surface is orientable.
A smooth surface is orientable if a continuously varying normal vector can be defined at every point on the surface.
This normal vector defines the positive side of the surface (the side the normal vector is pointing to) and a negative side (the side the normal vector points away from).
If the surface has a boundary, the normal vector defines an \emph{orientation} on the boundary curve, with the following convention: standing on the positive side of the surface, and walking around the boundary curve in the direction of the orientation, the surface is always on our left.

\begin{lemma} \label{orientable} The boundary surface, $M$, of a Coxeter HQECC in $\mathbb{H}^3$ is orientable.
\end{lemma}
\begin{proof} A piecewise smooth manifold (such as $M$) is orientable if, whenever two smooth component surfaces join along a common boundary, they induce \emph{opposite} orientation on the common boundary.

  Define the unit normal vector, $\hat{n}$, to a face $F_a^{(w)} \in {M}$ to point away from $P^{(w)}$ (i.e.
  it points into $P^{(v)}$ where $v = ws_A$).

  Consider the two possible configurations that could occur when two faces meet at a common edge.
  If the two faces always induce opposite orientation on the common edge (as in \cref{orient-fig} (a)) then $M$ is orientable.
  If the two faces ever induce the same orientation on the common edge (as in \cref{orient-fig} (b)) then $M$ is not orientable.

  If two faces which meet at a common edge of $M$ are part of the same polyhedral cell of the tessellation, i.e.
  they are faces $F_a^{(u)}$ and $F_b^{(u)}$, then it is guaranteed that the orientation of the surfaces will correspond to that shown in \cref{orient-fig} (a) as $\hat{n}$ is defined to point away from $P^{(u)}$.

  If two faces which meet at a common edge of $M$ are part of different polyhedral cells then parts (ii) and (iii) of \cref{edge_lemma} enforce that the orientation of the surfaces will always correspond to that shown in \cref{orient-fig} (a).

  Therefore $M$ is orientable.
\end{proof}

\begin{figure} \centering
  \begin{tikzpicture}

    \draw(2,1) -- (0,0.75); \draw(2,1) -- (2.25,3); \draw(2.25,3) -- (0.25,2.75); \draw (0,0.75) -- (0.25,2.75); \draw[->] (1.1,1.8) -- (1.1,2.1); \draw (1.05,1.8) -- (1.15,1.8); \draw(2,1) -- (4,0.75); \draw(2.25,3) -- (4.25,2.75); \draw(4.25,2.75) -- (4,0.75); \draw[->] (3.1,1.8) -- (3.1,2.1); \draw (3.05,1.8) -- (3.15,1.8); \draw (1.1,1.8) circle [radius = 0.75]; \draw[->] (1.85,1.8)--(1.85,1.801); \draw (3.1,1.8) circle [radius = 0.75]; \draw[->] (2.35,1.801)--(2.35,1.8); \node at (2,0) {(a)};

    \draw(8,1) -- (6,0.75); \draw(8,1) -- (8.25,3); \draw(8.25,3) -- (6.25,2.75); \draw (6,0.75) -- (6.25,2.75); \draw[->] (7.1,1.8) -- (7.1,2.1); \draw (7.05,1.8) -- (7.15,1.8); \draw(8,1) -- (10,0.75); \draw(8.25,3) -- (10.25,2.75); \draw(10.25,2.75) -- (10,0.75); \draw[->] (9.1, 1.8) -- (9.1,1.5); \draw (9.05,1.8) -- (9.15,1.8); \draw (7.1,1.8) circle [radius = 0.75]; \draw[->] (7.85,1.8)--(7.85,1.801); \draw (9.1,1.8) circle [radius = 0.75]; \draw[->] (8.35,1.8)--(8.35,1.801); \node at (8,0) {(b)};

  \end{tikzpicture}
  \caption{A cross section image of two possibilities for the orientation of faces that meet at a common edge in $M$.
    In (a) the two faces will induce opposite orientation on the common edge $E$, while in (b) the two faces will induce the same orientation on the common edge $E$.}
  \label{orient-fig}
\end{figure}
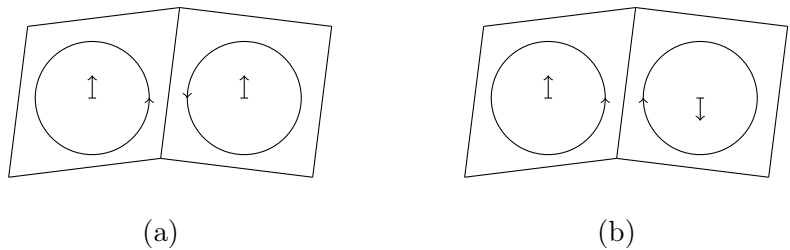

\begin{lemma}\label{connected} The boundary surface, $M$, of a Coxeter HQECC in $\mathbb{H}^3$ is connected.
\end{lemma}
\begin{proof} Let the boundary surface, $M$ be a sum of connected components, where we denote the $i^{th}$ connected component by $M^{(i)}$.
  $M$ is closed and orientable so by the classification of surface theorem it is the sum of spheres and connected sums of tori.
  Both spheres and tori have well defined interiors and exteriors, so we can define the interior and exterior of each $M^{(i)}$.

  Define the interior of $M^{(i)}$ to be the region that $\hat{n}^{(i)}$ points away from (i.e.
  the interior of $M^{(i)}$ contains $P^{(u)}$ for $F_a^{(u)} \in {M^{(i)}}$, $l_S(u) = R$).
  The exterior of $M^{(i)}$ is then the region that $\hat{n}^{(i)}$ points into to (i.e.
  the exterior of $M^{(i)}$ contains $P^{(v)}$ for $v = us_a$ where $F_a^{(u)} \in {M^{(i)}}$, $l_S(v) = R+1$).

  It follows from this definition that the exterior of each $M^{(i)}$ must be unbounded.
  To see this, note that for an infinite Coxeter group, $W$, every $w \in W$ has a non-empty $\mathcal{A}_R(w)$.
  This means that for arbitrary $w \in W$ there exists $s_a \in S$ such that $l_S(ws_a) = l_S(w)+1$.
  In terms of the HQECC this implies that the number of polyhedra in the exterior of any $M^{(i)}$ is infinite, so the exterior of $M^{(i)}$ is unbounded.

  Assume that $M = \cup_i M^{(i)}$ is composed of more than one connected component $M^{(i)}$.
  Consider any two components $M^{(1)}$ and $M^{(2)}$.
  There are three possible configurations:
  \begin{enumerate}
  \item $M^{(1)}$ and $M^{(2)}$ intersect.
    \label[case]{connected:intersect}
  \item $M^{(2)}$ is in the interior of $M^{(1)}$ (see \cref{case1}).
    \label[case]{connected:interior}
  \item $M^{(2)}$ is in the exterior of $M^{(1)}$ (see \cref{case2}).
    \label[case]{connected:exterior}
  \end{enumerate}

  However, \cref{connected:intersect} would imply that $M$ is not a surface, contradicting \cref{surface}.
  Thus we only need to consider \cref{connected:interior,connected:exterior}.

  The Coxeter group, and therefore the HQECC, contains a unique identity element of length $l_S(I)=0$.
  For any $ v \in W$ such that $l_S(v) = R$ we can write a reduced word for $v$ as $v = s_1^{(v)}s_2^{(v)}...s_R^{(v)}$.
  Using the fact that all the generators are involutions, it follows that $vs_R^{(v)}...s_2^{(v)}s_1^{(v)} = I$.
  We started with an element of length $R$, and applied $R$ generators to reach an element of length 0.
  Since each generator can only change the length of the previous element by $\pm1$ it follows that each generator must have decreased the length of the element by 1.
  Therefore in the HQECC there is a path through the tessellation from a polyhedra associated with an element of length $R$ to $P^{(I)}$ which passes through $R$ polyhedra, all associated with elements of length less than $R$.

  From the definition of the interior and exterior of $M^{(i)}$ it is clear that all polyhedra which lie directly on the interior of a given $M^{(i)}$ (i.e.
  those that are in the interior of $M^{(i)}$ and touching $M^{(i)}$) are associated with Coxeter group elements of length $R$.
  While all polyhedra which lie directly on the exterior of a given $M^{(i)}$ (i.e.
  those that are in the exterior of $M^{(i)}$ and touching $M^{(i)}$) are associated with Coxeter group elements of length $R+1$.

  Therefore there must always be a path from polyhedra directly on the interior of a $M^{(i)}$ to $P^{(I)}$ which doesn't cross $M^{(i)}$.
  In \cref{case1} and \cref{case2} it is clear that there is no location for $P^{(I)}$ which meets this condition.
  Therefore $M$ cannot be made up of more than one connected component.
\end{proof}

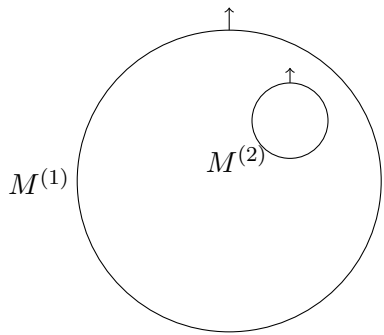
\begin{figure} \centering
  \begin{tikzpicture} \draw (0,0) circle [radius = 2]; \draw [->] (0,2) -- (0,2.3); \draw (0.8,0.8) circle [radius = 0.5]; \draw [->] (0.8,1.3) -- (0.8,1.5); \node at (-2.5,0) {$M^{(1)}$}; \node at (0.1,0.3) {$M^{(2)}$};
  \end{tikzpicture}
  \caption{The configuration of $M^{(1)}$ and $M^{(2)}$ when $M^{(2)}$ is in the interior of $M^{(1)}$.
    The arrows are the $\hat{n}^{(i)}$ which point into the exterior of each surface.
    The polyhedra which lie directly on the interior of a given $M^{(i)}$ are associated with Coxeter group elements of length $R$, while the polyhedra which lie directly on the exterior of a given $M^{(i)}$ are associated with elements of length $R+1$.
    Although we have drawn $M^{(1)}$ and $M^{(2)}$ as circles we are not assuming they are spherical, they could be tori, all we are assuming is that they have a well defined interior and exterior.}\label{case1}
\end{figure}

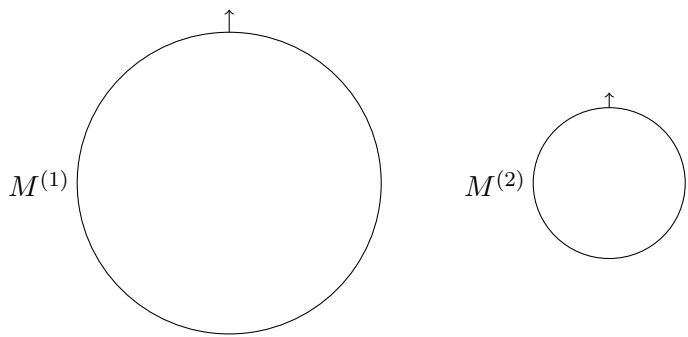
\begin{figure} \centering
  \begin{tikzpicture} \draw (0,0) circle [radius = 2]; \draw [->] (0,2) -- (0,2.3); \draw (5,0) circle [radius = 1]; \draw [->] (5,1) -- (5,1.2); \node at (-2.5,0) {$M^{(1)}$}; \node at (3.5,0) {$M^{(2)}$};
  \end{tikzpicture}
  \caption{The configuration of $M^{(1)}$ and $M^{(2)}$ when $M^{(2)}$ is in the exterior of $M^{(1)}$.
    The arrows are the $\hat{n}^{(i)}$ which point into the exterior of each surface.
    The polyhedra which lie directly on the interior of a given $M^{(i)}$ are associated with Coxeter group elements of length $R$, while the polyhedra which lie directly on the exterior of a given $M^{(i)}$ are associated with Coxeter group elements of length $R+1$.
    Although we have drawn $M^{(1)}$ and $M^{(2)}$ as circles we are not assuming they are spherical, they could be tori, all we are assuming is that they have a well defined interior and exterior.}\label{case2}
\end{figure}

\begin{lemma} \label{2_Sphere} The boundary surface, $M$, of a Coxeter HQECC in $\mathbb{H}^3$ is homeomorphic to the 2-sphere.
\end{lemma}
\begin{proof}

  By the classification theorem of closed surfaces, every connected closed surface is homeomorphic to either the sphere, a connected sum of tori, or a connected sum of real projective planes.
  Since $M$ is orientable, it is either homeomorphic to a sphere, or a connected sum of tori.

  Consider a loop $\mathcal{C}$, on the surface $M$.
  This loop is across faces which make up $M$, which are all associated to polyhedral cells of the tessellation corresponding to Coxeter group elements of length $R$.
  Since $M$ is closed (\cref{closed}), it cannot pinch down to a single point anywhere.
  Thus for any polyhedral vertex in $M$, $M$ must contain at least two faces that meet along one of the edges touching that vertex.
  Any loop which passes between adjacent faces which only touch via a common vertex can therefore be continuously deformed into a nearby loop which passes through those faces.
  Thus we can assume wlog neighbouring faces which contain adjacent sections of $\mathcal{C}$ share a common edge.

  There are four possible ways the polyhedral cells associated to these neighbouring faces could be connected (see \cref{genus-fig} for an illustration):
  \begin{enumerate}
  \item The neighbouring faces are associated with the same polyhedral cell.
    \label[case]{2_Sphere:same}
  \item A pair of neighbouring polyhedral cells share a single common edge.
    \label[case]{2_Sphere:edge}
  \item A pair of neighbouring polyhedral cells share a common face (and therefore also common edges).
    \label[case]{2_Sphere:face}
  \end{enumerate}

  \Cref{2_Sphere:face} is not possible because, if two polyhedra $P^{(u')}$ and $P^{(v')}$ where $l_S(u') = l_S(v') = r$ meet at a face, there would exist $s_a \in S$ such that $l_S(u's_a) = l_S(v') = r = l_S(u')$, contradicting property (i) of the length function of Coxeter groups (see \cref{combinatorics}).
  Therefore, for every pair of neighbouring faces containing adjacent sections of the non-contractible loop, either \cref{2_Sphere:same} or \cref{2_Sphere:edge} must hold.

  The surface of the HQECC at radius $R-1$ (which we will denote $M'$) is contained inside $M$, where ``inside'' is well-defined as $M$ is orientable by \cref{orientable}.
  Consider continuously deforming $\mathcal{C}$ so that it lies on $M'$, by the following procedure.
  Take a section of $\mathcal{C}$ which lies on the faces of a single polyhedral cell, $P^{(u)}$, and deform it so that it lies on the faces of $P^{(u)}$ associated with $\mathcal{D}_R(u)$ whilst leaving its endpoints unchanged.
  To see that this can always be done, note that at the edge where faces $F_a^{(u)}$ and $F_b^{(v)}$ from two polyhedral cells $P^{(u)}$ and $P^{(v)}$ meet (see \cref{genus-fig}(b)), the faces $F_b^{(u)}$ and $F_a^{(v)}$ are associated with $\mathcal{D}_R(u)$ by \cref{edge_lemma}\labelcref{edge_lemma:ii,edge_lemma:iii}.
  The faces of $P^{(u)}$ associated with $\mathcal{D}_R(u)$ share common edges, so this deformation can be carried out while leaving the curve intact.

  We can repeat this contraction procedure until the loop $\mathcal{C}$ lies on the faces of $P^{(I)}$.
  At that point we can continuously contract $\mathcal{C}$ through $P^{(I)}$ to a point.
  Therefore every loop on $M$ can be contracted through the bulk of the tessellation to a point.

  Any torus (or connected sum of tori) contains curves which cannot be continuously contracted to a point through the solid torus forming its interior.
  Therefore $M$ cannot be homeomorphic to the connected sum of tori.
  Thus $M$ is homeomorphic to a 2-sphere.
\end{proof}

\begin{figure}
  \begin{tikzpicture} \draw (-2,2.5)--(0,3) node[midway,above] {$F_a^{(u)}$}; \draw (2,2.5)--(0,3) node[midway,above] {$F_b^{(u)}$}; \draw[dashed] (-2,2.5) -- (-2,1); \draw[dashed] (2,1) -- (-2,1); \draw[dashed] (2,1) -- (2,2.5); \node at (0,1.75) {$P^{(u)}$}; \node at (0,0) {(a)};

    \draw (3,2.5)--(5,3) node[midway,above] {$F_a^{(u)}$}; \draw (7,2.5)--(5,3) node[midway,above] {$F_b^{(v)}$}; \draw[dashed] (3,2.5) -- (3,1); \draw[dashed] (5,3) -- (3,1) node[midway,below] {$F_b^{(u)}$}; \draw[dashed] (7,2.5) -- (7,1); \draw[dashed] (5,3) -- (7,1) node[midway,below] {$F_a^{(v)}$}; \node at (3.5,2.25) {$P^{(u)}$}; \node at (6.5,2.25) {$P^{(v)}$}; \node at (5,0) {(b)};

    \draw (8,2.5)--(10,3) node[midway,above] {$F_a^{(u)}$}; \draw (12,2.5)--(10,3) node[midway,above] {$F_b^{(v)}$}; \draw[dashed] (8,2.5) -- (8,1); \draw[dashed] (12,2.5) -- (12,1); \draw[dashed] (10,3) -- (10,1); \draw[dashed] (8,1) -- (12,1); \node at (9,1.75) {$P^{(u)}$}; \node at (11,1.75) {$P^{(v)}$}; \node at (10,0) {(c)};

  \end{tikzpicture}
  \caption{A cross section image of the three possible ways which neighbouring faces in a loop $\mathcal{C}$ on $M$ could be connected.
    In (a) the neighbouring faces are associated with a single polyhedral cell.
    In (b) neighbouring faces are associated with polyhedral cells which share a single common edge.
    In (c) the neighbouring faces are associated with polyhedral cells which share a common edge.
    The figures all represent cross-sections of the tessellations, and in all figures dashed edges represent faces of polyhedra which do not form part of $M$, while solid edges represent faces of polyhedra which form part of $M$.}
  \label{genus-fig}
\end{figure}
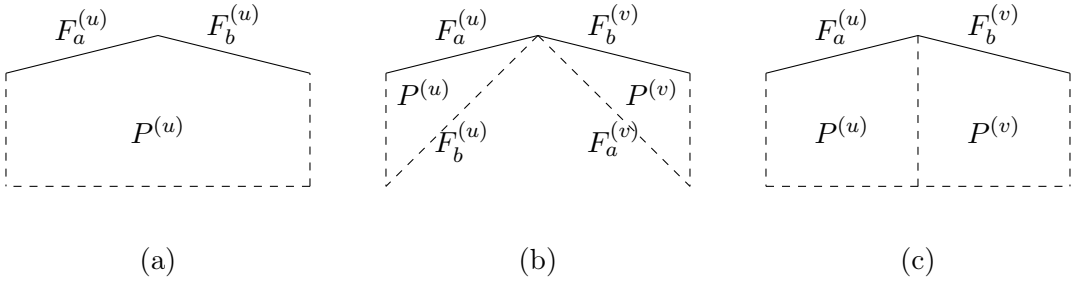

\subsubsection{The metric on the boundary surface of the HQECC}

We can upper-bound the distance between qudits on the boundary surface by the distance according the word metric between the corresponding elements of the Coxeter group.
Consider two boundary qudits, $q_a^{(u)}$ and $q_b^{(v)}$.
The Coxeter polytopes we use in the HQECC are of size $O(1)$ in every direction, so the distance between $q_a^{(u)}$ and $q_b^{(v)}$ is upper-bounded by $c d(u,v) = cl_S(u^{-1}v)$, for some constant $c$.

If $q_a^{(u)}$ and $q_b^{(v)}$ are nearest-neighbour qudits on the boundary surface of a HQECC then they are separated by distance $O(1)$.
This follows because the boundary surface of the HQECC is connected, so $F^{(u)}_a$ and $F^{(v)}_b$ must share a common edge.
The number of Coxeter polyhedra which fit round this edge is upper-bounded by $m_{ab}$, so $l_S(u^{-1}v) \leq m_{ab} = O(1).$

\subsubsection{Operators on the boundary surface of the HQECC} \label{ops}

In order to determine the overhead required to simulate the boundary Hamiltonian with a local model, we need to determine the distribution of operator weights that results from pushing the bulk Hamiltonian through the tensor network.
In~\cite{Pastawski:2015} it is shown that an operator $M$ can be reconstructed on a boundary region $A$ if $M$ lies in the greedy entanglement wedge of $A$, denoted $\mathcal{E}[A]$, where the greedy entanglement wedge is defined as below:

\begin{definition}[Greedy entanglement wedge, definition 8 from~\cite{Pastawski:2015}] Suppose $A$ is a (not necessarily connected) boundary region.
  The greedy entanglement wedge of $A$, denoted $\mathcal{E}[A]$, is the set of bulk points reached by applying the greedy algorithm to all connected components of $A$ simultaneously.
\end{definition}

The greedy algorithm is a simple procedure for finding bulk regions which can be reconstructed on a given boundary region.
It considers a sequence of cuts $\{c_\alpha\}$ through the tensor network which are bounded by $\partial A$ and which correspond to a set of isometries $\{P_\alpha\}$ from bulk indices to boundary indices.
The algorithm begins with the cut $c_1 = A$, corresponding to $P_1 = \identity$.
Each cut in the sequence is obtained from the previous one by identifying a (pseudo)-perfect tensor which has at least half of its indices contracted with $P_\alpha$, and adding that tensor to $P_\alpha$ to construct $P_{\alpha +1}$.
In this way $P_{\alpha +1}$ is guaranteed to be an isometry if $P_\alpha$ is.
The algorithm terminates when there are no tensors which have at least half their indices contracted with $P_\alpha$.
(See~\cite{Pastawski:2015} for details).
The greedy algorithm only relies on the properties of perfect tensors, so we can apply it to our HQECC in $\hyper$.

A given bulk point will be in the greedy entanglement wedge of many boundary regions.
As we are interested in minimising the operator weights of the boundary Hamiltonian we want to calculate the smallest boundary region needed to reconstruct a bulk operator.

Consider a HQECC described by a perfect tensor, $T$, and a Coxeter system $(W,S)$ with associated Coxeter polyhedra $P \subseteq \mathbb{H}^3$.
Let the growth rate of $W$ with respect to $S$ be $\tau$, and let the radius of the HQECC be $R$.

By the definition of the growth rate, the number of boundary qudits, $N$, scales as $O(\tau^R)$.
Reconstructing an operator which acts on the central bulk qudit requires an $O(1)$ fraction of the boundary, so requires $O(\tau^R)$ boundary qudits.\footnote{In theory it is possible to work out the value of the $O(1)$ constant from the properties of the (pseudo-)perfect tensor and Coxeter system used in a given HQECC, however as we are concerned with asymptotic scaling of weights we don't provide an example of this calculation.}

Consider the number of boundary qudits required to reconstruct on operator which acts on a qudit, $q^{(v)}$, at distance $x$ from the centre.
By assumption $\mathcal{A}_R(w) > \mathcal{D}_R(w)$ for all $w \in W$, and hence for all polyhedral cells in our tessellation of $\mathbb{H}^3$.
Therefore if we take an operator acting on $q^{(v)}$, we can push it to the boundary while at each step moving outwards in the tensor network - i.e.
we are guaranteed to be able to reconstruct the operator on the boundary using only qudits which are a distance $R-x$ from $q^{(v)}$.
If we consider shifting the centre of the tensor network to $q^{(v)}$ we can see that there are $O(\tau^{R-x})$ qudits which are at distance $R-x$ from $q^{(v)}$.
Not all of these lie on the boundary, but we can upper-bound the number of qudits needed for boundary reconstruction by $O(\tau^{R-x})$.

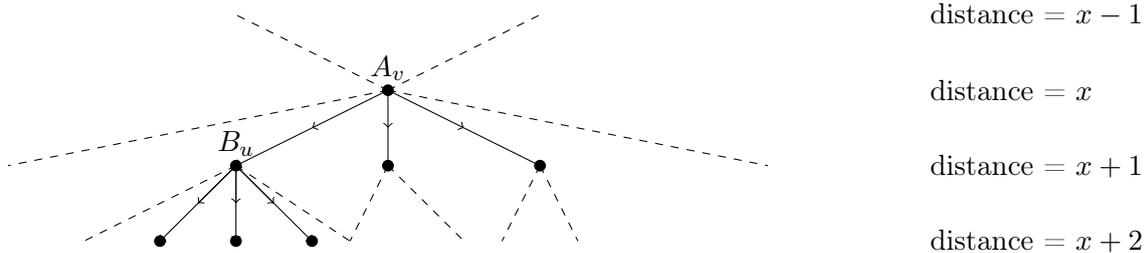
\begin{figure} \centering
  \begin{tikzpicture} \filldraw[black] (5,5) circle (2pt) node[anchor=south] {$A_v$}; \draw (12,5) node[anchor=west] {distance = $x$}; \draw [dashed] (5,5) -- (0,4); \filldraw[black] (3,4) circle (2pt) node[anchor=south] {$B_u$}; \filldraw[black] (5,4) circle (2pt); \filldraw[black] (7,4) circle (2pt); \draw [dashed] (5,5) -- (10,4); \draw [dashed] (5,5) -- (3,6); \draw [dashed] (5,5) -- (7,6); \draw (12,4) node[anchor=west] {distance = $x+1$}; \draw (12,6) node[anchor=west] {distance = $x-1$}; \draw [->] (5,5) -- (4,4.5); \draw [->] (4,4.5) -- (3,4); \draw [->] (5,5) -- (5,4.5); \draw (5,4.5) -- (5,4); \draw [->] (5,5) -- (6,4.5); \draw (6,4.5) -- (7,4); \draw (12,3) node[anchor=west] {distance = $x+2$}; \draw [dashed] (5,4) -- (6,3); \draw [dashed] (5,4) -- (4.5,3); \draw [dashed] (7,4) -- (6.5,3); \draw [dashed] (7,4) -- (7.5,3); \draw [dashed] (3,4) -- (4.5,3); \draw [dashed] (3,4) -- (1,3); \draw [->] (3,4) -- (3,3.5); \draw [->] (3,4) -- (3,3); \filldraw[black] (3,3) circle (2pt); \filldraw[black] (2,3) circle (2pt); \filldraw[black] (4,3) circle (2pt); \draw [->] (3,4) -- (2.5,3.5); \draw [->] (3,4) -- (2,3); \draw [->] (3,4) -- (3.5,3.5); \draw [->] (3,4) -- (4,3);
  \end{tikzpicture}
  \caption{If qudit $q^{(v)}$ is at distance $x$ from the centre of the HQECC, and qudit $q^{(u)}$ is a neighbouring qudit at distance $x+1$ from the centre, then an operator $A_v$ acting on $q^{(v)}$ can be pushed through $q^{(u)}$, so pushing an operator $A_v \otimes B_u$ where $B_u$ acts on $q^{(u)}$ through the tensor network will lead to a boundary operator with the same weight as pushing $A_v$ alone through the tensor network.}
  \label{fig-recon}
\end{figure}

If we consider a geometrically $k$-local operator in the bulk, where the deepest qudit the operator acts on is at distance $x$ from the centre,\footnote{Here by deepest operator we mean nearest the centre, so the minimum $x$.}
then the number of qudits needed for boundary reconstruction scales as $O(\tau^{R-x})$.
To see this consider an operator $A \otimes B$ where $A$ acts on qudit $q^{(v)}$ which is at distance $x$ from the centre, and $B$ acts on a neighbouring qudit $q^{(u)}$ at distance $x+1$, as in \cref{fig-recon}.
We can push $A$ through $q^{(u)}$ to reach the boundary, therefore $B$ can necessarily be reconstructed on a subset of the qudits required to reconstruct $A$.
Hence $B$ makes no difference to the number of qudits required for boundary reconstruction, and we only need to consider the deepest qudit a given operator acts on.
In general there may be more than one deepest qudit, however as $k$ is constant this can make at most a constant factor difference to the number of qudits needed for reconstruction.

The number of qudits at distance $x$ from the centre of the tensor network scales as $O(\tau^x)$, so we find that for $x = 0,\dots,R$ the boundary Hamiltonian has $O(\tau^x)$ operators of weight $O(\tau^{R-x})$.
All boundary operators can be chosen to be geometrically $O(\tau^{R-x})$ local (i.e.
the $O(\tau^{R-x})$ qudits which an operator act on are spread over an $O(\tau^{-x})$ fraction of the boundary).

\subsubsection{Full holographic duality} \label{main_theorem}

In this section we prove our main result: that using a HQECC and simulation techniques from Hamiltonian complexity it is possible to construct a full holographic duality between quantum many-body models in 3D hyperbolic space and models living on its 2D boundary.

We will require the following Lemma in the proof of the main theorem:

\begin{lemma} \label{propertiesHS}
Consider a HQECC constructed using Coxeter system $(W,S)$ and perfect tensor $T$.
Let $T^{(w)}$ denote the perfect tensor associated with element $w \in W$, and let $\mathcal{I}^{(w)}$ be the set of indices of $T^{(w)}$ which are contracted through faces $F_a^{(w)}$ for $ a \in \mathcal{D}_R(w)$.
Define:
  \begin{equation}
  \Pi_{\mathcal{C}^{(w)}} = \frac{1}{|\mathcal{S}^{(w)}|}\sum_{M \in \mathcal{S}^{(w)}} \overline{M}
  \end{equation}
where $\mathcal{S}^{(w)}$ is the stabilizer group of the QECC defined by viewing $T^{(w)}$ as an isometry from $\{\mathcal{I}^{(w)} \cup q^{(w)}\}$ to the complementary set of indices, and $\overline{M}$ is the boundary operator associated to the stabilizer $M$.

Let:
\begin{equation}
    H_S = \sum_{w \in W} \left( \identity - \Pi_{\mathcal{C}^{(w)}}\right)
\end{equation}
Then:
\begin{enumerate}
\item%
The kernel of $H_S$ is the code-subspace of the HQECC, $\mathcal{C}$.
\label{prop:i}

\item%
The smallest non-zero eigenvalue of $H_S$ is one.
\label{prop:ii}

\item%
Energy with respect to $H_S$ is equal to the number of logical qudits encoded by the HQECC which have a correctable error.
\label{prop:iii}

\item%
Eigenstates of $H_S$ with the same energy, but which pick up energy from errors on different logical qudits, are orthogonal.
\label{prop:4}

\end{enumerate}
\end{lemma}

\begin{proof}
1. The set $\left\{\cup_{w \in W} \mathcal{S}^{(w)}\right\}$ is a (non-minimal) generating set for the stabilizer group of the HQECC. Therefore, $\Pi_{\mathcal{C}^{(w)}}\ket{\psi} = \ket{\psi}$ for all $w \in W$ iff $\ket{\psi} \in \mathcal{C}$. So $H_S \ket{\psi} = 0$ iff $\ket{\psi} \in \mathcal{C}$.

2. This is immediate as each term in $H_S$ is a projector, so has eigenvalue zero or one.

3. Consider a boundary state $\ket{\phi}$.
If there is a correctable error affecting the state of encoded qudit $q^{(w)}$ then $ \exists M \in \mathcal{S}^{(w)}$ such that $\overline{M}\ket{\phi} \neq \ket{\phi}$.
Therefore $\Pi_{\mathcal{C}^{(w)}}\ket{\phi} \neq 1$.
Since $\Pi_{\mathcal{C}^{(w)}}$ is a projector, this gives $\Pi_{\mathcal{C}^{(w)}}\ket{\phi} = 0$. Therefore, $\ket{\phi}$ picks up energy +1 from the $\left( \identity - \Pi_{\mathcal{C}^{(w)}}\right)$ term in $H_S$.

If there is not a correctable error affecting the state of encoded qudit $q^{(w')}$ then $\overline{M}\ket{\phi} = \ket{\phi}$ for all $M \in \mathcal{S}^{(w)}$.
So $\Pi_{\mathcal{C}^{(w')}}\ket{\phi} = \ket{\phi}$, regardless of any errors affecting other encoded qudits.
Therefore $\ket{\phi}$ picks up zero energy from the $\left( \identity - \Pi_{\mathcal{C}^{(w')}}\right)$ term in $H_S$.

Therefore the energy with respect to $H_S$ counts the number of bulk qudits which have a correctable error.

4. Consider boundary states, $\ket{\psi}$, $\ket{\phi}$, which have correctable errors affecting the state of encoded qudits $q^{(w)}$ and $q^{(w')}$ respectively. We have:
\begin{equation}
 \Pi_{\mathcal{C}^{(w)}}\ket{\psi} = 0 \text{,          }
 \Pi_{\mathcal{C}^{(w)}}\ket{\phi} = \ket{\phi} \text{,          }
 \Pi_{\mathcal{C}^{(w')}}\ket{\psi} = \ket{\psi} \text{,           }
 \Pi_{\mathcal{C}^{(w')}}\ket{\phi} = 0
\end{equation}
Therefore $\bra{\psi}\ket{\phi} = \bra{\psi}\Pi_{\mathcal{C}^{(w)}}\ket{\phi}=\bra{\psi} \Pi_{\mathcal{C}^{(w')}}\ket{\phi} =0$.
\end{proof}

\begin{theorem}\label{main} Let $\hyper$ denote 3D hyperbolic space, and let $B_r(x)\subset\mathbb{H}^3$ denote a ball of radius $r$ centred at $x$.
  Consider any arrangement of $n$ qudits in $\hyper$ such that, for some fixed $r$, at most $k$ qudits and at least one qudit are contained within any $B_r(x)$.
  Let $L$ denote the minimum radius ball $B_L(0)$ containing all the qudits (which wlog we can take to be centred at the origin).
  Let $\Hbulk = \sum_Z h_Z$ be any local Hamiltonian on these qudits, where each $h_Z$ acts only on qudits contained within some $B_r(x)$.

  Then we can construct a Hamiltonian $\Hboundary$ on a 2D boundary manifold $\mathcal{M}\in\mathbb{H}^3$ with the following properties:
  \begin{enumerate}
  \item%
    $\mathcal{M}$ surrounds all the qudits, has diameter $O\left(\max(1,\frac{\ln(k)}{r}) L + \log\log n\right)$, and is homeomorphic to the Euclidean 2-sphere.
    \label{main:i}
  \item%
    The Hilbert space of the boundary consists of a triangulation of $\mathcal{M}$ by triangles of $O(1)$ area, with a qubit at the centre of each triangle, and a total of $O\left(n(\log n)^4\right)$ triangles/qubits.
    \label{main:ii}
  \item%
    Any local observable/measurement $M$ in the bulk has a set of corresponding observables/measurements $\{M'\}$ on the boundary with the same outcome.
    A local bulk operator $M$ can be reconstructed on a boundary region $A$ if $M$ acts within the greedy entanglement wedge of $A$, denoted $\mathcal{E}[A]$.
    \label{main:iii}
  \item%
    $\Hboundary$ consists of 2-local, nearest-neighbour interactions between the boundary qubits.
    Furthermore, $\Hboundary$ can be chosen to have full local $SU(2)$ symmetry; i.e.\ the local interactions can be chosen to all be Heisenberg interactions: $\Hboundary = \sum_{\langle i,j\rangle} \alpha_{ij} (X_iX_j + Y_iY_j + Z_iZ_j)$.
    \label{main:iv}
  \item%
    $\Hboundary$ is a $(\Delta_L,\epsilon,\eta)$-simulation of $\Hbulk$ in the rigorous sense of~\cite[Definition~23]{cubitt:2017}, with $\epsilon,\eta = 1/\poly(\Delta_L)$, $\Delta_L = \Omega\left(\norm{\Hbulk}^6\right)$, and where the maximum interaction strength $\Lambda = \max_{ij}\abs{\alpha_{ij}}$ in $\Hboundary$ scales as $\Lambda = O\left(\Delta_L^{\poly(n\log(n))}\right)$.
    \label{main:v}
  \end{enumerate}
\end{theorem}

\begin{proof}

  There are four steps to this simulation:

  \noindent\paragraph{Step 1}\\
  Simulate $\Hbulk$ with a Hamiltonian which acts on the bulk indices of a HQECC in $\hyper$ of radius $R = O\left(\max(1,\frac{\ln(k)}{r}) L\right)$.

  Note that in a tessellation of $\hyper$ by Coxeter polytopes the number of polyhedral cells in a ball of radius $r'$ scales as $O(\tau^{r'})$, where we are measuring distances using the word metric, $d(u,v) = l_S(u^{-1}v)$.
  If we want to embed a Hamiltonian $\Hbulk$ in a tessellation we will need to rescale distances between the qudits in $\Hbulk$ so that there is at most one qudit per polyhedral cell of the tessellation.
  If $\tau^{r'} = k$, then $\frac{r'}{r} = \frac{\ln(k)}{\ln(\tau) r} = O\left(\frac{\ln(k)}{r}\right)$.
  If $\frac{\ln(k)}{r} \geq 1$ then the qudits in $\Hbulk$ are more tightly packed than the polyhedral cells in the tessellation, and we need to rescale the distances between the qudits by a factor of $O\left(\frac{\ln(k)}{r}\right)$.
  If $\frac{\ln(k)}{r} < 1$ then the qudits in $\Hbulk$ are less tightly packed then the cells of the tessellation, and there is no need for rescaling.

  The radius, $R$, of the tessellation needed to contain all the qudits in $\Hbulk$ is then given by:
  \begin{equation}
    R =
    \begin{cases}
      O\left(\frac{\ln(k)}{r}L\right),& \text{if } \frac{\ln(k)}{r} \geq 1 \\
      O(L) & \text{otherwise}
    \end{cases}
  \end{equation}

  After rescaling there is at most one qudit per cell of the tessellation.
  There will be some cells of the tessellation which don't contain any qudits.
  We can put `dummy' qudits in these cells which don't participate in any interactions, so their inclusion is just equivalent to tensoring the Hamiltonian with an identity operator.
  We can upper and lower bound the number of `real' qudits in the tessellation.
  If no cells contain dummy qudits then the number of real qudits in the tesselation is given by $n_{\max} = N = O(\tau^R)$, where $N$ is the number of cells in the tessellation.
  By assumption there is at least one real qudit in a ball of radius $r'$, therefore the minimum number of real qudits in the tessellation scales as $n_{\min} = O\left(\frac{\tau^R}{\tau^{r'}} \right) = O(\tau^R) = O(N)$.
  Therefore $n = \Theta(\tau^R) = \Theta(N)$.

  If the tessellation of $\hyper$ by Coxeter polytopes is going to form a HQECC, the Coxeter polytope must have at least 7 faces.
  We show in \cref{example_1} that this bound is achievable, so we will wlog assume the tessellation we are using is by a Coxeter polytope with 7 faces.
  The perfect tensor used in the HQECC must therefore have 8 indices.
  Our method to construct perfect tensors can be used to construct perfect tensors with 8 indices for qudits of prime dimension $p \geq 11$.
  Qudits of general dimension $d$ can be incorporated by embedding qudits into a $d$-dimensional subspace of the smallest prime which satisfies both $p \geq d$ and $p \geq 11$.
  We then add one-body projectors onto the orthogonal complement of these subspaces, multiplied by some $\Delta_S' \geq |\Hbulk|$ to the embedded bulk Hamiltonian.
  The Hamiltonian, $\Hbulk'$ on the $n$ $p$-dimensional qudits is then a perfect simulation of $\Hbulk$.

 We can therefore simulate any $\Hbulk$ which meets the requirements stated in the theorem with a Hamiltonian which acts on the bulk indices of a HQECC in $\hyper$.

  Now consider simulating $\Hbulk$ with a Hamiltonian $H_B$ on the boundary surface of the HQECC.
  %The error correction properties of the HQECC ensure that $\Hbulk$ is redundantly encoded on the boundary, so there are a number of Hamiltonians we could choose.
%  Let $S^{(w)}$ be the stabilizer group of the $w^{th}$ tensor in the HQECC, viewed as an isometry from input indices to output indices.
%  Then:
%  \begin{equation}
%  \Pi_{\mathcal{C}^{(w)}} = \sum_{M \in S^{(w)}} \overline{M}
%  \end{equation}
%  (where $\overline{M}$ is the boundary operator associated to the stabilizer $M$) is the projector onto the code-subspace of the $w^{th}$ tensor. Define:
%  \begin{equation}
%    H_S = \sum_{w} \left( \identity - \Pi_{\mathcal{C}^{(w)}}\right)
%  \end{equation}
%  Then the nullspace of $H_S$ is the code-subspace, $\mathcal{C}$, of the HQECC. To see this note that $\left\{\cup_w S^{(w)}\right\}$ is a (non-minimal) generating set for the stabilizer group of the entire HQECC. Therefore, $\Pi_{\mathcal{C}^{(w)}}\ket{\psi} = \ket{\psi}$ for all $w$ in the HQECC iff $\ket{\psi} \in \mathcal{C}$. Each of the terms in $H_S$ has eigenvalues 0 and 1, so the smallest non-zero eigenvalue of $H_S$ is 1.
  Let:
  \begin{equation}
    H_B =  H' + \Delta_SH_S
  \end{equation}
  where $H_S$ is as defined in \cref{propertiesHS}, $H'$ satisfies $H'\Pi_{\mathcal{C}} = V(\Hbulk' \otimes \identity_{dummy})V^\dagger$, $V$ is the encoding isometry of the HQECC, $\Pi_{\mathcal{C}}$ is the projector onto the code-subspace of the HQECC and $\identity_{dummy}$ acts on the dummy qudits.
%  \begin{equation}
%  H_S = \sum_{w} \left( \identity - \Pi_{\mathcal{C}}^{(w)}\right)
%  \end{equation}
%  where the sum is over all tensors in the HQECC, and $\Pi_{\mathcal{C}}^{(w)}$ is the boundary operator which projects onto the code-subspace of the QECC defined by the $w^{th}$ tensor, viewed as an isometry from input indices to output indices.

Provided $\Delta_S \geq \norm{\Hbulk'}$, \cref{prop:i} and \cref{prop:ii} from \cref{propertiesHS} ensure that $H_B$ meets the conditions in~\cite{cubitt:2017} to be a perfect simulation of $\Hbulk'$ below energy $\Delta_S$, and (as simulations compose) a perfect simulation of $\Hbulk$.

%Now consider simulating $\Hbulk$ with a Hamiltonian $H_B$ on the boundary surface of the HQECC.
%The map $V(\Hbulk' \otimes \identity_{dummy})V^\dagger$, where $V$ is the encoding isometry of the HQECC and $\identity_{dummy}$ acts on the dummy qudits, meets the requirements from \cite{cubitt:2017} to be an encoding into the code-subspace of the HQECC.
%Therefore:
% \begin{equation}
%H_B =  H' + \Delta_SH_S
%\end{equation}
%where $H'$ satisfies $H'\Pi_{\mathcal{C}} = V(\Hbulk' \otimes \identity_{dummy})V^\dagger$, $\Pi_{\mathcal{C}}$ is the projector onto the code-subspace of the HQECC, and the nullspace of $H_S$ is the code-subpsace of the HQECC, is a perfect simulation of $\Hbulk'$ below energy $\Delta_S$ provided the smallest non-zero eigenvalue of $H_S$ is 1.
%
%Consider the Hamiltonian:
%  \begin{equation}
%  H_S = \sum_{w} \left( \identity - \Pi_{\mathcal{C}}^{(w)}\right)
%  \end{equation}
%where the sum is over all tensors in the HQECC, and $ \Pi_{\mathcal{C}}^{(w)} = \sum_{M \in S^{(w)}} M$ is t

  There is freedom in this definition as there are many $H'$ which satisfy the condition stated.
  We will choose an $H'$ where every bulk operator has been pushed directly out to the boundary, so that a 1-local bulk operator at radius $x$ corresponds to a boundary operator of weight $O(\tau^{R-x})$.
  We will also require that the Pauli rank of every bulk operator has been preserved (see \cref{pauli_rank} for proof we can choose $H'$ satisfying this condition).

  \noindent\paragraph{Step 2 \footnote{In steps 2 and 3 we are following the methods developed in~\cite{oliveira:2005}, replacing the qubit perturbation gadgets with qudit perturbation gadgets, and making use of the structure of the interaction graph on the boundary.
    }}\\
  Having constructed $H_B$ we now want to simulate it with a geometrically 2-local qudit Hamiltonian.
  To achieve this, we make use of the subdivision and 3-2 gadgets from \cref{perturbative_simulations}.

Consider simulating a single $k$-local interaction which is a tensor product of operators of the form $P_A + P_A^\dagger$ by 2-local interactions using these gadgets. The first step will be to simulate the interaction by two $\ceil{\frac{k}{2}}+1$- local interactions by applying the subdivision gadget. Then apply the subdivision gadget again to give four $O(\frac{k}{4})$-local interactions. Continue until all the interactions are 3-local. Finally use the 3-2 gadget to simulate on each 3-local interaction. This process requires $O(k)$ ancillas, and $O(\log(k))$ rounds of perturbation theory.

  The original qudits are in the centre of the polygon-cells which form the boundary.\footnote{The polygon cells are the faces of the tensor network which correspond to the uncontracted tensor indices.}
Place the ancilla qudits required for this simulation on the edges of the cells separating the sets of qudits they are interacting with (see \cref{ancilla_placement} for an example).

This process can be applied to each of the interactions in $H_B$ independently. $H_B$ contains $O(\tau^x)$ operators of weight $O(\tau^{R-x})$ for $x \in [0,R]$ (see \cref{ops}).
%  Breaking down a $k$-local operator which is a tensor product of operators of the form $P_A + P_A^\dagger$ to a 2-local operator using the subdivision and 3-2 gadgets requires $O(k)$ ancillas, and $O(\log(k))$ rounds of perturbation theory.
  Therefore applying this step to every interaction in $H_B$ will require a total of:
  \begin{equation}
    N_a = O( \sum_{x=0}^R \tau^x \tau^{R-x}) = O(R\tau^R) = O(n\log(n))
  \end{equation} ancilla qudits.
% The original qudits are in the centre of the polygon-cells which form the boundary.\footnote{The polygon cells are the faces of the tensor network which correspond to the uncontracted tensor indices.}
 % Place the ancilla qudits resulting from this step on the edges of the cells.
  Each edge will therefore contain $O(R) = O(\log(n))$ qudits.
  When breaking down the interactions to 2-local the ancillas are placed nearest the qudits they are interacting with, so none of the resulting 2-local interactions cross more than two of the cells which make up the boundary.

  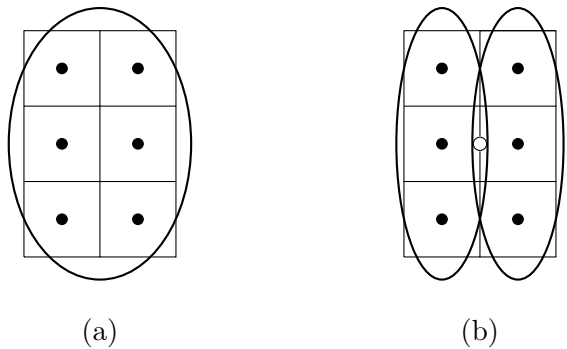
\begin{figure}
  \centering
  \begin{tikzpicture}
  \draw (0,0) -- (1,0) -- (1,1) -- (0,1) -- (0,0);
   \draw (1,1) -- (1,2) -- (0,2) -- (0,1);
   \draw (1,1) -- (2,1) -- (2,0) -- (1,0);
   \draw (2,1) -- (2,2) -- (1,2);
   \draw (2,2) -- (2,3) -- (1,3) -- (1,2);
\draw(0,2) -- (0,3) -- (1,3);
\filldraw[black] (0.5,0.5) circle (2pt);
\filldraw[black] (1.5,0.5) circle (2pt);
\filldraw[black] (0.5,2.5) circle (2pt);
\filldraw[black] (0.5,1.5) circle (2pt);
\filldraw[black] (1.5,1.5) circle (2pt);
\filldraw[black] (1.5,2.5) circle (2pt);
\draw[thick] (1,1.5) ellipse (1.2cm and 1.8cm);
\node at (1,-1) {(a)};

  \draw (5,0) -- (6,0) -- (6,1) -- (5,1) -- (5,0);
   \draw  (5,2) -- (5,1);
   \draw (6,1) -- (6,1.4);
   \draw (6,1.6) -- (6,2);
   \draw (5,2)--(6,2);
   \draw (6,1) -- (7,1) -- (7,0) -- (6,0);
   \draw (7,1) -- (7,2) -- (6,2);
   \draw (7,2) -- (7,3) -- (6,3) -- (6,2);
\draw(5,2) -- (5,3) -- (7,3);
\filldraw[black] (5.5,0.5) circle (2pt);
\filldraw[black] (6.5,0.5) circle (2pt);
\filldraw[black] (5.5,2.5) circle (2pt);
\filldraw[black] (5.5,1.5) circle (2pt);
\filldraw[black] (6.5,1.5) circle (2pt);
\filldraw[black] (6.5,2.5) circle (2pt);
\draw (6,1.5) circle (2.5pt);
\draw[thick] (5.5,1.5) ellipse (0.6cm and 1.8cm);
\draw[thick] (6.5,1.5) ellipse (0.6cm and 1.8cm);
\node at (6,-1) {(b)};
  \end{tikzpicture}
  \caption{(a) A 6-local interaction in a HQECC where the cells composing the boundary surfaces are square. (b) The 6-local interaction is simulated by two 4-local interactions, by introducing an ancilla qudit (denoted by the white vertex) which is placed on an edge separating the sets of qudits it is interacting with.} \label{ancilla_placement}
  \end{figure}

  As there are interactions with Pauli-weights which scale as $O(n)$ this step requires $O(\log(n))$ rounds of perturbation theory.
  By \cref{second_order}, the first round of perturbation theory will require interaction strengths of $\Delta_L = \Omega\left(\norm{\Hbulk}^{6}\right)$, while $r$ rounds of perturbation theory requires interaction strengths to scale as $\Delta_L^{6^r}$. Therefore, this step requires maximum interaction strengths scaling as $\Lambda = O\left(\Delta_L^{\poly(n)}\right)$.

  \noindent\paragraph{Step 3}\\
  Each of the ancillas introduced in step 2 has degree at most 6.
  The degree of the original qudits after step 2 is the same as their degree in the initial hypergraph, which can be calculated as:
  \begin{equation}
    d = \frac{\sum_{x=0}^R \tau^x \tau^{R-x}}{\tau^R} = R
  \end{equation}
  Therefore there are $O(\tau^R)$ qudits of degree $O(R)$, and $O(R\tau^R)$ qudits of degree $O(1)$.

  We reduce the degree of each vertex to at most $3(p-1)$ (where $p$ is the local dimension of each qudit) in the following manner:\footnote{A recent paper~\cite{Aharonov:2018} has derived a method to reduce the degree of a Hamiltonian, $H$, using only polynomial strength interactions.
    However, this method cannot be used here as it assumes $||H|| = O(\poly(n))$, whereas $||H_B|| = O(\exp(n))$.}
  \begin{enumerate}[i.]
  \item Use the subdivision gadget to localise each qudit with degree $O(R)$.
    This requires $O(R)$ ancilla qudits per cell of the boundary, so $O(R\tau^R)$ ancillas in total.
  \item Apply the triangle gadget to each qudit to reduce the degree to $3(p-1)$, by pairing edges of the form $P_a + P_a^\dagger$ in parallel.
    Reducing the degree of one $O(R)$ degree vertex in this manner requires $O(\log(R))$ rounds of perturbation theory, and $O(R)$ ancillas.
    Therefore applying this step to the entire graph requires $O(R\tau^R)$ ancillas.
  \end{enumerate}

  Once the degree of each qudit has been reduced there are $O(R)$ qudits in each of the cells of the boundary.

  Finally we need to remove all the crossings using crossing gadgets.
  Each interaction is constrained to 2 of the cells which make up the boundary surface, so we can consider each cell and its adjacent cells separately.\footnote{This will double-count some crossings as each cell will be included when considering its adjacent cells too, but as we are only interested in the asymptotic scaling this double-counting is not important.}
  There are $O(R)$ qudits, and hence at most $O(R^2)$ interactions in each cell, including contributions from adjacent cells.
  Therefore there are at most $O(R^4)$ crossings.
  We use the subdivision gadget to localise each crossing,\footnote{This step can be skipped for edges with only one crossing, where each qudit involved in the crossing interactions has degree at most $3(p-1)$.}
  then apply the crossing gadget in parallel to each localised crossing.
  This requires $O(R^4)$ ancillas per cell of the boundary, so requires $O(R^4\tau^R) = O(n\log(n)^4)$ ancilla qudits.
  These ancilla qudits are placed within the corresponding face of the boundary surface.

  This step required $O(\log \log n)$ rounds of perturbation theory, so we require the maximum interaction strength $\Lambda$ to scale as $\Lambda = O\left(\Delta_L^{\poly(n\log(n))}\right)$

  Label the Hamiltonian resulting from this step as $H_B'$.

  \noindent\paragraph{Step 4}\\
  Finally, if we want the boundary Hamiltonian to have full local $SU(2)$ symmetry, we can simulate $H_B'$ with a qubit $\{XX+YY+ZZ\}$-Hamiltonian on a $2d$ lattice, $\Hboundary$.

  First use the technique from Lemma 21 in~\cite{cubitt:2017} to simulate $H_B'$ with a qubit Hamiltonian by simulating each $p$-dimensional qudit with $\ceil*{\log_2p}$ qubits.
  The resulting Hamiltonian is given by:
  \begin{equation}
    H_B'' = \mathcal{E}(H_B') + \Delta \sum_{i=0}^{n'} P_i
  \end{equation}
  where $\mathcal{E}(M) = VMV^\dagger$, $V = W^{\otimes n'}$,\footnote{$n'$ is the total number of qudits in $H_B'$.}
  $W$ is an isometry $W: \mathbb{C}^d \rightarrow \left(\mathbb{C}^2 \right)^{\otimes \ceil*{\log_2 p}}$, and $P = \identity - WW^\dagger$.
  This requires $n' \ceil*{\log_2p} = O(n')$ qubits.

  The operators in $H_B''$ are at most $2 \ceil*{\log_2 p}$-local, and the qubits have degree at most $3(p-1)\ceil*{\log_2 p}$.

  Next we use the technique from~\cite[Lemma~22]{cubitt:2017} to simulate $H_B''$ with a real Hamiltonian.
  This is a perfect simulation, which requires $2N' = O(n')$ qubits, it increases the locality of the interactions to at most $4 \ceil*{\log_2 p} = O(1)$ and doesn't change the degree of the qubits.

  Using the technique from~\cite[Lemma~39]{cubitt:2017} we then simulate the real Hamiltonian with a Hamiltonian containing no $Y$ operators.
  This involves adding an ancilla qubit for every interaction in the Hamiltonian.
  As each qubit is involved in a fixed number of interactions, this only requires $O(n')$ ancilla qubits, so the total number of qubits involved in the Hamiltonian is still $O(n')$.
  The locality of each interaction in the Hamiltonian is increased by 1.
  This requires $O(1)$ rounds of perturbation theory.

  The qubit subdivision and 3-2 perturbation gadgets from~\cite{oliveira:2005} can then be used to reduce the Hamiltonian containing no $Y$s to a 2-local Pauli interaction with no $Y$s, leaving a Hamiltonian of the form $\sum_{i>j} \alpha_{ij}A_{ij} + \sum_k\left(\beta_kX_k + \gamma_k Z_k \right)$ where $A_{ij}$ is one of $X_iX_j$, $X_iZ_j$, $Z_iX_j$ or $Z_iZ_j$~\cite[Lemma~39]{cubitt:2017}.
  This requires $O(1)$ rounds of perturbation theory, and $O(n')$ ancilla qubits.

  Next we use the subspace perturbation gadget from~\cite[Lemma~41]{cubitt:2017} to simulate the Hamiltonian of the form $\sum_{i>j} \alpha_{ij}A_{ij} + \sum_k\left(\beta_kX_k + \gamma_k Z_k \right)$ with a $\{XX+YY+ZZ\}$-Hamiltonian.
  This requires encoding one logical qubit in four physical qubits, so introduces $O(N')$ ancilla qubits, and requires $O(1)$ rounds of perturbation theory.

  Finally, we can simulate the general $\{XX+YY+ZZ\}$-Hamiltonian with a $\{XX+YY+ZZ\}$-Hamiltonian on a triangulation of the boundary surface of the HQECC using the perturbation gadgets from~\cite{Piddock:2017}.
  These perturbation gadgets are generalisations of the fork, crossing and subdivision gadgets from~\cite{oliveira:2005} which use a pair of mediator qubits, rather than a single ancilla qubit, so that all interactions in the final Hamiltonian are of the form $\{XX+YY+ZZ\}$.
  Following the method in~\cite{Piddock:2017}, first reduce the degree of all vertices in the interaction graph to 3 using the subdivision and fork gadgets.
  This requires $O(1)$ ancillas and $O(1)$ rounds of perturbation theory per qubit, and can be done to all qubits in the Hamiltonian in parallel.
  Next remove all the crossings.
  The qudit Hamiltonian $H_B'$ had no crossings, and our simulation of $H_B'$ with a $\{XX+YY+ZZ\}$-Hamiltonian will have introduced $O(1)$ crossings per qudit in $H_B'$, so $O(n')$ crossings across the entire interaction graph.
  The crossings are localised using the subdivision gadget, then removed using the crossing gadget.
  This requires $O(n')$ ancilla qubits.

  Step 4 therefore requires a total of $O(n') = O\left(n \log(n)^4\right)$ qubits.
  The scaling of the interaction strengths in the Hamiltonian is unchanged by this final step as it only required $O(1)$ rounds of perturbation theory.

  Each qubit has degree at most 3, so we can construct a triangulation of the boundary surface with a qubit in the centre of each triangle.
  This is not an $O(1)$ triangulation, but if we increase the diameter of our boundary manifold to $O\left(\max(1,\frac{\ln(k)}{r}) L + \log\log n\right)$ then we can construct an $O(1)$ triangulation with a qubit in the centre of each triangle (this follows because we are working in hyperbolic space).
  This surface will be homeomorphic to a sphere as boundary surface of the HQECC is homeomorphic to a sphere by \cref{2_Sphere}.

  The final Hamiltonian, $\Hboundary$, is a $(\Delta_L,\epsilon,\eta)$-simulation of $\Hbulk$ with full local SU(2) symmetry.
  The total number of qubits required scales as $O\left(n\log(n)^4\right)$, and the interaction strengths in the Hamiltonian scale as $\max_{ij}|\alpha_{ij}| = O\left(\Delta_L^{\poly(n \log(n))} \right)$ where $\Delta_L = \Omega\left(\norm{\Hbulk}^6\right)$.
  The perturbation gadget techniques require that $\epsilon, \eta = 1 / \poly(\Delta_L)$.\footnote{In steps 2 and 3 we assume that all operators are Pauli rank 2 operators of the form $P_A + P_A^\dagger$.
    We have shown that the HQECC preserves the Pauli rank of operators (\cref{pauli_rank}), so accounting for operators of general form will only increase the overheads calculated by a constant factor.}

  It is immediate from the definition of the greedy entanglement wedge~\cite[Definition~8]{Pastawski:2015} that bulk local operators in $\mathcal{E}(A)$ can be reconstructed on $A$.
  The boundary observables / measurements $\{M'\}$ corresponding to a bulk observable / measurement $M$ have the same outcome because simulations preserve the outcome of all measurements.

\end{proof}

Note that the fact that $\Hboundary$ is a $(\Delta_L,\epsilon,\eta)$-simulation of $\Hbulk$ immediately implies, by \cref{physical-properties}, that the partition function, time dynamics, and all measurement outcomes of the boundary are the same as that of the bulk, up to $O(1/\poly(\epsilon,\eta))$ errors which can be made as small as desired by increasing $\Delta_L$.

\subsection{HQECC constructed from pentagonal prisms} \label{example_1}

The proof of \cref{main} does not require any particular HQECC -- all it requires is that one exists.
Here and in \cref{example_2} we provide examples of two pairs of Coxeter group and tensor which can be used to construct a HQECC.
There are many more which could be constructed.

First we construct a HQECC using a perfect tensor, and a non-uniform Coxeter polytope.
The Coxeter polytope, $P_1$, we use is a pentagonal prism.
It is described by the Coxeter diagram, $\Sigma(P_1)$, shown in \cref{pentagonal_prism}.
The elliptic subdiagrams of $\Sigma(P_1)$ are shown in \cref{prism_elliptic}.\footnote{Elliptic subdiagrams are subdiagrams containing $J \subseteq S$ such that $W_J$ is finite.}

The maximum $|J|$ such that $W_J$ is finite is three, so $|\mathcal{D}_R(w)| \leq 3 \; \forall w \in W$.
Clearly if we construct a perfect tensor with 8 legs, and place one tensor in each polyhedral-cell in a tessellation of $\mathbb{H}^3$ by pentagonal prisms then the tensor network will be a HQECC.
Details of the tensor are given in \cref{perf_sec}.

The growth rate of the Coxeter group is 3.13.\footnote{The growth rate was calculated using CoxIterWeb~\cite{Guglielmetti}, a web applet which computes invariants of Coxiter groups.}

\begin{figure} \centering
  \begin{tikzpicture} \draw [dashed] (-0.5,0) -- (1,0); \draw (1,0) -- (2,1); \draw (1,0) -- (2,-1); \draw [dashed] (2,1) -- (3.5,1); \draw [dashed] (2,-1) -- (3.5,-1); \draw [dashed] (3.5,1) -- (4.5,0); \draw [dashed] (3.5,-1) -- (4.5,0); \draw [dashed] (2,-1) -- (2,1); \filldraw[black] (-0.5,0) circle (2pt) node[anchor=east] {$s_g$}; \filldraw[black] (1,0) circle (2pt) node[anchor=south] {$s_f$}; \filldraw[black] (2,1) circle (2pt) node[anchor=south] {$s_a$}; \filldraw[black] (2,-1) circle (2pt) node[anchor=north] {$s_d$}; \filldraw[black] (3.5,1) circle (2pt) node[anchor=south] {$s_c$}; \filldraw[black] (3.5,-1) circle (2pt) node[anchor=north] {$s_b$}; \filldraw[black] (4.5,0) circle (2pt) node[anchor=west] {$s_e$};
  \end{tikzpicture}
  \caption{Coxeter diagram, $\Sigma(P_1)$, for the the Coxeter group associated with the pentagonal prism.
    The vertices of the graph are labelled with the corresponding generator of the Coxeter group.
    The faces $F_f$ and $F_g$ of the Coxeter polyhedra (corresponding to the generators $s_f$ and $s_g$) are the pentagonal faces of the prism, the faces $F_a-F_e$ are the quadrilateral faces.}
  \label{pentagonal_prism}
\end{figure}
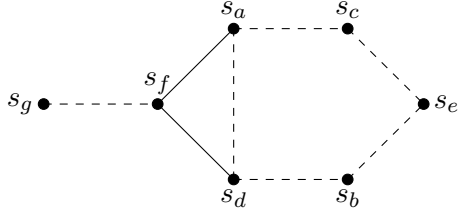

\begin{table}
  \begin{tabular}{c|c|c}
    Name & Diagram & Generating sets \\  \hline
    $A_1$ &
            \begin{tikzpicture} \filldraw[black] (0,0) circle (2pt) node[anchor=west] {};
            \end{tikzpicture}  & $\{s_a\}, \{s_b\}, \{s_c\}, \{s_d\}, \{s_e\}, \{s_f\}, \{s_g\}$
    \\

    \hline

         &

                   &  $ \{s_a, s_b\}, \{s_a, s_e\},\{s_a, s_g\},\{s_b, s_c\},$ \\ $A_1 \times A_1$ &
                                                                                                     \begin{tikzpicture} \filldraw[black] (0,0) circle (2pt) node[anchor=west] {}; \filldraw[black] (1,0) circle (2pt) node[anchor=west] {};
                                                                                                     \end{tikzpicture} &  $\{s_b, s_f\}, \{s_b, s_g\},\{s_c, s_d\}, $ \\ & & $ \{s_c, s_f\}, \{s_c, s_g\},\{s_d, s_e\}, $ \\ & & $ \{s_d, s_g\}, \{s_e, s_f\}, \{s_e, s_g\}$

    \\ \hline

    $A_2$ (equivalently $I_2^{(3)}$)
         &
           \begin{tikzpicture} \filldraw[black] (0,0) circle (2pt) node[anchor=west] {}; \filldraw[black] (1,0) circle (2pt) node[anchor=west] {}; \draw (0,0) -- (1,0);
           \end{tikzpicture}
                   & $ \{s_a, s_f\},\{s_d, s_f\}$

    \\ \hline
    $A_1 \times A_1 \times A_1$
         &
           \begin{tikzpicture} \filldraw[black] (0,0) circle (2pt) node[anchor=west] {}; \filldraw[black] (1,0) circle (2pt) node[anchor=west] {}; \filldraw[black] (2,0) circle (2pt) node[anchor=west] {};
           \end{tikzpicture}
                   & $\{s_a, s_e, s_g\} , \{s_a, s_b, s_g\} , \{s_c, s_d, s_g\} , $ \\ & & $ \{s_d, s_e, s_g\}, \{s_b, s_c, s_f\} , \{s_b, s_c, s_g\}$

    \\ \hline

    $A_2 \times A_1$
         &
           \begin{tikzpicture} \filldraw[black] (0,0) circle (2pt) node[anchor=west] {}; \filldraw[black] (1,0) circle (2pt) node[anchor=west] {}; \filldraw[black] (2,0) circle (2pt) node[anchor=west] {}; \draw (0,0) -- (1,0);
           \end{tikzpicture}

                   & $\{s_a, s_b, s_f\},\{s_a, s_e, s_f\}, \{s_c, s_d, s_f\} , \{s_d, s_e, s_f\} $

  \end{tabular}
  \caption{Elliptic subdiagrams of $\Sigma(P_1)$.}
  \label{prism_elliptic}
\end{table}

\subsubsection{Perfect tensor} \label{perf_sec}

We use the procedure set out in \cref{appendix_5} to construct a perfect tensor with the required properties.

Construct a AME$(8,11)$ stabilizer state via a classical Reed Solomon code with $n=8,k=4$ over $\mathbb{Z}_{11}$ defined by the set $S = \{1,2,3,4,5,6,7,8\} \in \mathbb{Z}_{11}$.
The generator matrix is given by:
\begin{equation} \label{non-standard}
  G = \left(
    \begin{array}{c c c c c c c c}
      1 & 1 & 1 & 1 & 1 & 1 & 1 & 1 \\
      1 & 2 & 3 & 4 & 5 & 6 & 7 & 8 \\
      1 & 4 & 9 & 5 &  3 & 3 & 5 & 9 \\
      1 & 8 & 5 & 9 &  4 & 7 & 2 & 6
    \end{array}
  \right)
\end{equation}

In standard form this becomes:

\begin{equation}
  G = \left(
    \begin{array}{@{}c c c c  c c c c@{}}
      1 & 0 & 0 & 0 & 10 & 7 & 1 & 2 \\
      0 & 1 & 0 & 0 & 4 & 4 & 3 & 4 \\
      0 & 0 & 1 & 0 & 5 & 2 & 10 & 4 \\
      0 & 0 & 0 & 1 & 4 & 10 & 9 & 2
    \end{array}
  \right)
\end{equation}

Giving a parity check matrix:
\begin{equation}
  H = \left(
    \begin{array}{@{}c c c c  c c c c@{}}
      1 & 7 & 6 & 7 &1 & 0 & 0 & 0  \\
      4 & 7 & 9 & 1 & 0 & 1 & 0 & 0 \\
      10 & 8 & 1 & 2 & 0 & 0 & 1 & 0  \\
      9 & 7 & 7 & 9 & 0 & 0 & 0 & 1
    \end{array}
  \right)
\end{equation}

The stabilizer generators of the AME$(8,11)$ stabilizer state are then given by:
\begin{equation}
  M = \left(
    \begin{array}{@{}c c c c  c c c c | c c c c  c c c c@{}}
      1 & 0 & 0 & 0 & 10 & 7 & 1 & 2 &0 &  0 &  0 &  0 &   0 & 0 &  0 & 0\\
      0 & 1 & 0 & 0 & 4 & 4 & 3 & 4 &0 &  0 &  0 &  0 &   0 & 0 &  0 & 0 \\
      0 & 0 & 1 & 0 & 5 & 2 & 10 & 4  &0 &  0 &  0 &  0 &   0 & 0 &  0 & 0 \\
      0 & 0 & 0 & 1 & 4 & 10 & 9 & 2 &0 &  0 &  0 &  0 &   0 & 0 &  0 & 0 \\
      0 &  0 &  0 &  0 &   0 & 0 &  0 & 0 &1 & 7 & 6 & 7 &1 & 0 & 0 & 0  \\
      0 &  0 &  0 &  0 &   0 & 0 &  0 & 0 & 4 & 7 & 9 & 1 & 0 & 1 & 0 & 0 \\
      0 &  0 &  0 &  0 &   0 & 0 &  0 & 0 & 10 & 8 & 1 & 2 & 0 & 0 & 1 & 0  \\
      0 &  0 &  0 &  0 &   0 & 0 &  0 & 0 & 9 & 7 & 7 & 9 & 0 & 0 & 0 & 1
    \end{array}
  \right)
\end{equation}

The tensor which describes the stabilizer state is a perfect tensor.

\subsection{HQECC based on the order-4 dodecahedral honeycomb} \label{example_2}

There are only four compact, regular honeycombings of $\hyper$, and all the honeycombings are by polyhedra with even number of faces, so to use any of them in a HQECC would require a pseudo-perfect tensor.
Here we use the order-4 dodecahedral honeycomb.

The Coxeter polytope, $P_2$, we use is right angled dodecahedron.
It is described by the Coxeter diagram, $\Sigma(P_2)$, shown in \cref{dodecahedral_honeycomb}.
The elliptic subdiagrams of $\Sigma(P_2)$ are shown in \cref{dodec_ellip}.

The maximum $|J|$ such that $W_J$ is finite is three, so $|\mathcal{D}_R(w)| \leq 3 \; \forall w \in W$.
Clearly if we construct a pseudo-perfect tensor with 13 legs, and place one tensor in each polyhedral-cell in a tessellation of $\mathbb{H}^3$ by right-angled dodecahedra then the tensor network will be a HQECC.
Details of the tensor are given in \cref{pseudo_sec}.

The growth rate of the Coxeter group is 7.87.\footnote{The growth rate was calculated using CoxIterWeb~\cite{Guglielmetti}, a web applet which computes invariants of Coxiter groups.}

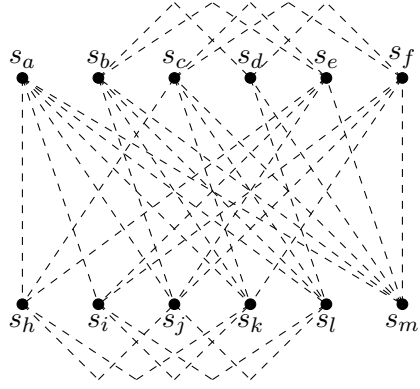
\begin{figure} \centering
  \begin{tikzpicture} \draw[dashed] (0,2) -- (0,-1); \draw[dashed] (0,2) -- (1,-1); \draw[dashed] (0,2) -- (2,-1); \draw[dashed] (0,2) -- (3,-1); \draw[dashed] (0,2) -- (4,-1); \draw[dashed] (0,2) -- (5,-1);

    \draw[dashed] (5,-1) -- (1,2); \draw[dashed] (5,-1) -- (2,2); \draw[dashed] (5,-1) -- (3,2); \draw[dashed] (5,-1) -- (4,2); \draw[dashed] (5,-1) -- (5,2);

    \draw[dashed](1,2) -- (2,3); \draw[dashed](2,3) -- (3,2);

    \draw[dashed](1,2) -- (2.5,3); \draw[dashed](4,2) -- (2.5,3);

    \draw[dashed](1,2) -- (4,-1); \draw[dashed](1,2) -- (3,-1); \draw[dashed](1,2) -- (2,-1);

    \draw[dashed](2,2) -- (3,3); \draw[dashed](3,3) -- (4,2);

    \draw[dashed](2,2) -- (3.5,3); \draw[dashed](3.5,3) -- (5,2);

    \draw[dashed](2,2) -- (0,-1); \draw[dashed](2,2) -- (4,-1); \draw[dashed](2,2) -- (3,-1);

    \draw[dashed](3,2) -- (4,3); \draw[dashed](5,2) -- (4,3);

    \draw[dashed](3,2) -- (4,-1);

    \draw[dashed](4,2) -- (0,-1); \draw[dashed](4,2) -- (1,-1); \draw[dashed](4,2) -- (2,-1);

    \draw[dashed](5,2) -- (1,-1); \draw[dashed](5,2) -- (2,-1); \draw[dashed](5,2) -- (3,-1);

    \draw[dashed](0,-1) -- (1,-2); \draw[dashed](1,-2) -- (2,-1);

    \draw[dashed](0,-1) -- (1.5,-2); \draw[dashed](1.5,-2) -- (3,-1);

    \draw[dashed](1,-1) -- (2,-2); \draw[dashed](2,-2) -- (3,-1);

    \draw[dashed](1,-1) -- (2.5,-2); \draw[dashed](2.5,-2) -- (4,-1);

    \draw[dashed](2,-1) -- (3,-2); \draw[dashed](3,-2) -- (4,-1);

    \filldraw[black] (0,2) circle (2pt) node[anchor=south] {$s_a$}; \filldraw[black] (1,2) circle (2pt) node[anchor=south] {$s_b$}; \filldraw[black] (2,2) circle (2pt) node[anchor=south] {$s_c$}; \filldraw[black] (3,2) circle (2pt) node[anchor=south] {$s_d$}; \filldraw[black] (4,2) circle (2pt) node[anchor=south] {$s_e$}; \filldraw[black] (5,2) circle (2pt) node[anchor=south] {$s_f$}; \filldraw[black] (0,-1) circle (2pt) node[anchor=north] {$s_h$}; \filldraw[black] (1,-1) circle (2pt) node[anchor=north] {$s_i$}; \filldraw[black] (2,-1) circle (2pt) node[anchor=north] {$s_j$}; \filldraw[black] (3,-1) circle (2pt) node[anchor=north] {$s_k$}; \filldraw[black] (4,-1) circle (2pt) node[anchor=north] {$s_l$}; \filldraw[black] (5,-1) circle (2pt) node[anchor=north] {$s_m$};
  \end{tikzpicture}
  \caption{Coxeter diagram, $\Sigma(P_2)$, for the the Coxeter group associated with order-4 dodecahedral honeycomb.
    Each pentagonal face of a dodecahedral cell has a dihedral angle of $\frac{\pi}{2}$ with the five faces it intersects, and diverges from the other six faces of the dodecahedron.}
  \label{dodecahedral_honeycomb}
\end{figure}

\begin{table}
  \begin{tabular}{c|c|c}
    Name & Diagram & Generating sets \\  \hline
    $A_1$ &
            \begin{tikzpicture} \filldraw[black] (0,0) circle (2pt) node[anchor=west] {};
            \end{tikzpicture}  & $ \{s_a\}, \{s_b\}, \{s_c\}, \{s_d\}, \{s_e\}, \{s_f\}, $ \\ & & $\{s_g\}, \{s_h\}, \{s_i\}, \{s_j\}, \{s_k\}, \{s_l\}$
    \\

    \hline

         &

                   &  $ \{s_a,s_b\},  \{s_a,s_c\}, \{s_a,s_d\}, \{s_a,s_e\}, $ \\ & & $ \{s_a,s_f\},  \{s_b,s_c\},  \{s_b,s_f\}, \{s_b,s_g\}$ \\ & &
                                                                                                                                                     $\{s_b,s_h\}, \{s_c,s_d\}, \{s_c,s_h\},\{s_c,s_i\}, $ \\ & & $ \{s_d,s_e\} ,  \{s_d,s_i\}, \{s_d,s_j\},  \{s_e,s_f\} , $

    \\ $A_1 \times A_1$ &
                          \begin{tikzpicture} \filldraw[black] (0,0) circle (2pt) node[anchor=west] {}; \filldraw[black] (1,0) circle (2pt) node[anchor=west] {};
                          \end{tikzpicture} &

                                              $\{s_e,s_j\}, \{s_e,s_k\},  \{s_f,s_k\},\{s_f,s_g\},  $ \\ & & $ \{s_g,s_l\} ,  \{s_g,s_h\}, \{s_g,s_k\},  \{s_h,s_l\}$ \\
         & & $\{s_h,s_i\}, \{s_i,s_j\},  \{s_i,s_l\},\{s_j,s_l\},  $ \\ & & $  \{s_j,s_k\} ,  \{s_k,s_l\}$

    \\ \hline

         & & $\{s_a,s_b,s_c\}, \{s_a,s_c,s_d\},  \{s_a,s_d,s_e\},  \{s_a,s_e,s_f\}$ \\ & & $\{s_a,s_f,s_b\},   \{s_b,s_f,s_g\},  \{s_b,s_c,s_h\},  \{s_b,s_g,s_h\}$ \\  & & $\{s_c,s_h,s_i\},  \{s_c,s_d,s_i\},  \{s_d,s_e,s_j\},   \{s_d,s_i,s_j\},$
    \\
    $A_1 \times A_1 \times A_1$
         &
           \begin{tikzpicture} \filldraw[black] (0,0) circle (2pt) node[anchor=west] {}; \filldraw[black] (1,0) circle (2pt) node[anchor=west] {}; \filldraw[black] (2,0) circle (2pt) node[anchor=west] {};
           \end{tikzpicture}
                   & $\{s_e,s_f,s_k\}, \{s_e,s_j,s_k\},  \{s_f,s_g,s_k\},  \{s_g,s_k,s_l\}$ \\ & & $\{s_g,s_h,s_l\},   \{s_h,s_i,s_l\},  \{s_i,s_j,s_l\}, \{s_j,s_k,s_l\}$
  \end{tabular}
  \caption{Elliptic subdiagrams of $\Sigma(P_2)$.}
  \label{dodec_ellip}
\end{table}

\subsubsection{Pseudo-perfect tensor} \label{pseudo_sec}

We use the procedure set out in \cref{appendix_5} to construct a pseudo-perfect tensor with the required properties.

Construct a AME$(13,13)$ stabilizer state via a classical Reed Solomon code with $n=13,k=6$ over $\mathbb{Z}_{13}$ defined by the set $S = \{0,1,2,3,4,5,6,7,8,9,10,11,12\} \in \mathbb{Z}_{13}$.
The generator matrix is given by:
\begin{equation}
  G = \left(
    \begin{array}{c c c c c c c c c c c c c}
      1 & 1 & 1 & 1 & 1 & 1 & 1 & 1 & 1 & 1 & 1 & 1 & 1 \\
      0 & 1 & 2 & 3 & 4 & 5 & 6 & 7 & 8 & 9 & 10 & 11 & 12 \\
      0 & 1 & 4 & 9 & 3 & 12 & 10 & 10 & 12 & 3 & 9 & 4 & 1 \\
      0 & 1 & 8 & 1 & 12 & 8 & 8 & 5 & 5 & 1 & 12 & 5 & 12 \\
      0 & 1 & 3 & 3 & 9 & 1 & 9 & 9 & 1 & 9 & 3 & 3 & 1 \\
      0 & 1 & 6 & 9 & 10 & 5 & 2 & 11 & 8 & 3 & 4 & 7 & 12
    \end{array}
  \right)
\end{equation}

In standard form this becomes:

\begin{equation}
  G = \left(
    \begin{array}{@{}c c c c c c c c c c c c c @{}}
      1 & 0 & 0 & 0 & 0 & 0 & 12 & 7 & 5 & 9 & 4 & 8 & 6 \\
      0 & 1 & 0 & 0 & 0 & 0 & 6 & 9 & 3 & 3 & 11 & 8 &11 \\
      0 & 0 & 1 & 0 & 0 & 0 & 11 & 7 & 6 & 8 & 11& 1 & 7 \\
      0 & 0 & 0 & 1 & 0 & 0 & 7 & 1 & 11 & 8 & 6 & 7 & 11 \\
      0 & 0 & 0 & 0 & 1 & 0 & 11 & 8 & 11 & 3 & 3 & 9 & 6 \\
      0 & 0 & 0 & 0 & 0 & 1 & 6 & 8 & 4 & 9 & 5 & 7 & 12
    \end{array}
  \right)
\end{equation}

Giving a parity check matrix:
\begin{equation}
  H = \left(
    \begin{array}{@{}c c c c c c c c c c c c c@{}}
      1 & 7 & 2 & 6 & 2 & 7& 1 & 0 & 0 & 0 & 0 & 0 & 0\\
      6 & 4 & 6 & 12&5 & 5 & 0 & 1 & 0 & 0 & 0 & 0 & 0 \\
      8 & 10 & 7 & 2 & 2 &9 & 0 & 0 & 1 & 0 & 0 & 0 & 0 \\
      4 & 10 & 5 & 5 & 10 & 4 & 0 & 0 & 0 & 1 & 0 & 0 & 0 \\
      9 & 2 & 2 & 7 & 10 & 8 & 0 & 0 & 0 & 0 & 1 & 0 & 0 \\
      5 & 5 & 12 & 6 & 4 & 6 & 0 & 0 & 0 & 0 & 0 & 1 & 0 \\
      7&  2 & 6 & 2& 7 &1& 0 & 0 & 0 & 0 & 0 & 0 & 1
    \end{array}
  \right)
\end{equation}

The stabilizer generators of the AME$(13,13)$ state are given by:
\begin{equation} \label{gen-pseudo}
  M = \left(
    \begin{array}{@{}c|c@{}}
      G & 0 \\
      0 & H
    \end{array}
  \right)
\end{equation} %\todo[inline]{Sort out alignment}
% \begin{equation} \label{gen-pseudo}
%   M = \left(
% \begin{array}{@{}c c c c c c c c c c c c c | c c c c c c c c c c c c c@{}} % 1 & 1 & 1 & 1 & 1 & 1 & 1 & 1 & 1 & 1 & 1 & 1 & 1 & 0 & 0 & 0 & 0 & 0 & 0 & 0 & 0 & 0 & 0 & 0 & 0 & 0\\
% 0 & 1 & 2 & 3 & 4 & 5 & 6 & 7 & 8 & 9 & 10 & 11 & 12 & 0 & 0 & 0 & 0 & 0 & 0 & 0 & 0 & 0 & 0 & 0 & 0 & 0\\
% 0 & 1 & 4 & 9 & 3 & 12 & 10 & 10 & 12 & 3 & 9 & 4 & 1 & 0 & 0 & 0 & 0 & 0 & 0 & 0 & 0 & 0 & 0 & 0 & 0 & 0\\
% 0 & 1 & 8 & 1 & 12 & 8 & 8 & 5 & 5 & 1 & 12 & 5 & 12 & 0 & 0 & 0 & 0 & 0 & 0 & 0 & 0 & 0 & 0 & 0 & 0 & 0\\
% 0 & 1 & 3 & 3 & 9 & 1 & 9 & 9 & 1 & 9 & 3 & 3 & 1 & 0 & 0 & 0 & 0 & 0 & 0 & 0 & 0 & 0 & 0 & 0 & 0 & 0\\
% 0 & 1 & 6 & 9 & 10 & 5 & 2 & 11 & 8 & 3 & 4 & 7 & 12 & 0 & 0 & 0 & 0 & 0 & 0 & 0 & 0 & 0 & 0 & 0 & 0 & 0\\
% 0 & 0 & 0 & 0 & 0 & 0 & 0 & 0 & 0 & 0 & 0 & 0 & 0 & 1 & 7 & 2 & 6 & 2 & 7& 1 & 0 & 0 & 0 & 0 & 0 & 0\\
% 0 & 0 & 0 & 0 & 0 & 0 & 0 & 0 & 0 & 0 & 0 & 0 & 0 & 6 & 4 & 6 & 12&5 & 5 & 0 & 1 & 0 & 0 & 0 & 0 & 0 \\
% 0 & 0 & 0 & 0 & 0 & 0 & 0 & 0 & 0 & 0 & 0 & 0 & 0 & 8 & 10 & 7 & 2 & 2 &9 & 0 & 0 & 1 & 0 & 0 & 0 & 0 \\
% 0 & 0 & 0 & 0 & 0 & 0 & 0 & 0 & 0 & 0 & 0 & 0 & 0 & 4 & 10 & 5 & 5 & 10 & 4 & 0 & 0 & 0 & 1 & 0 & 0 & 0 \\
% 0 & 0 & 0 & 0 & 0 & 0 & 0 & 0 & 0 & 0 & 0 & 0 & 0 & 9 & 2 & 2 & 7 & 10 & 8 & 0 & 0 & 0 & 0 & 1 & 0 & 0 \\
% 0 & 0 & 0 & 0 & 0 & 0 & 0 & 0 & 0 & 0 & 0 & 0 & 0 & 5 & 5 & 12 & 6 & 4 & 6 & 0 & 0 & 0 & 0 & 0 & 1 & 0 \\
% 0 & 0 & 0 & 0 & 0 & 0 & 0 & 0 & 0 & 0 & 0 & 0 & 0 & 7& 2 & 6 & 2& 7 &1& 0 & 0 & 0 & 0 & 0 & 0 & 1
% \end{array}
% \right)
% \end{equation}

The tensor which describes the AME($13,13$) stabilizer state is a pseudo-perfect tensor.

\section*{Acknowledgements}
The authors would like to thank Eva Silverstein for very helpful discussions about holographic dualities, and Fernando Pastawski for explaining holographic tensor network constructions, during the 2017 KITP programme on Quantum Physics of Information.
This research was supported in part by the National Science Foundation under Grant No. NSF PHY-1748958.
TK is supported by the EPSRC Centre for Doctoral Training in Delivering Quantum Technologies [EP/L015242/1].
TC is supported by the Royal Society.

\appendix
\section{(Pseudo-)perfect tensors and absolutely maximally entangled states} \label{appendix_1}

Perfect and pseudo-perfect tensors are closely related to the concept of absolutely maximally entangled (AME) states, which are maximally entangled across all bipartitions. More formally:

\begin{definition}[Absolutely maximally entangled states, definition 1 from~\cite{Helwig:2013}]
An AME state is a pure state, shared among $n$ parties $P = \{1,...,n\}$, each having a system of dimension $q$. Hence, $\ket{\Phi} \in \mathcal{H}_1 \otimes ... \otimes \mathcal{H}_n$, with the following equivalent properties:
\begin{enumerate}[i.]
\item $\ket{\Phi}$ is maximally entangled for any possible bipartition. This means that for any bipartition of $P$ into disjoint sets $A$ and $B$ with $A \cup B = P$, and without loss of generality $m = |B| \leq |A| = n-m$, the state $\ket{\Phi}$ can be written in the form:
\begin{equation}
\ket{\Phi} = \frac{1}{\sqrt{d^m}} \sum_{\bm{k} \in \mathbb{Z}_q^m} \ket{k_1}_{B_1}\ket{k_2}_{B_2}...\ket{k_m}_{B_m} \ket{\phi(\bm{k})}_A
\end{equation}
with $\bra{\phi(k)}\ket{\phi(k')} = \delta_{kk'}$
\item The reduced density matrix of every subset of parties $A \subset P$ with $|A| = \floor*{\frac{n}{2}}$ is maximally mixed, $\rho_A = q^{- \floor*{\frac{n}{2}}} \identity_{q^{- \floor*{\frac{n}{2}}}}$.
\item The reduced density matrix of every subset of parties $A \subset P$ with $|A| \leq \frac {n}{2}$ is maximally mixed.
\item The von Neumann entropy of every subset of parties $A \subset P$ with $|A| = \floor*{\frac{n}{2}}$ is maximal, $S(A) = \floor*{\frac{n}{2}} \log{q}$.
\item The von Neumann entropy of every subset of parties $A \subset P$ with $|A| \leq \frac{n}{2}$ is maximal, $S(A) = |A| \log{q}$.
\end{enumerate}
These are all necessary and sufficient conditions for a state to be absolutely maximally entangled. We denote such state as an AME(n,q) state.
\end{definition}

The connection between perfect tensors and AME states was noted in~\cite{Pastawski:2015}, and separately in~\cite{Goyeneche:2015} (where perfect tensors are referred to as multi-unitary matrices). Here we generalise the arguments from~\cite{Goyeneche:2015} to encompass the case of pseudo-perfect tensors.

A $t$-index tensor, where each index ranges over $q$ values, describes a pure quantum state of $t$ $q$-dimensional qudits:
\begin{equation}
\ket{\psi} = \sum_{a_1a_2...a_t \in \mathbb{Z}_q^t} T_{a_1a_2...a_t} \ket{a_1a_2...a_t}
\end{equation}

A necessary and sufficient condition for $\ket{\psi}$ to be an AME state is that the reduced density matrix of any set of particles $A$ such that $|A| \leq \floor*{\frac{t}{2}}$ is maximally mixed. The reduced density matrix $\rho_A$ can be calculated as $\rho_A = MM^\dagger$, where $M$ is a $|A| \times |A^c|$ matrix formed by reshaping $T$. Therefore, the state $\ket{\psi}$ is an AME state if and only if the tensor $T$ is an isometry from any set of indices $A$ to the complementary set of indices $A^c$ with $|A| \leq |A^c|$.

If $t$ is even (odd) this implies that $T$ is a perfect (pseudo-perfect) tensor. Therefore an AME state containing an even (odd) number of qudits can be described by a perfect (pseudo-perfect) tensor, and every perfect (pseudo-perfect) tensor describes an AME state on an even (odd) number of qudits.

\section{(Pseudo-)perfect tensors and quantum error correcting codes} \label{appendix_2}
An $[n,k,d]_q$ quantum error correcting code (QECC) encodes $k$ $q$-dimensional qudits into $n$ $q$-dimensional qudits, such that $d-1$ located errors (or $\frac{d-1}{2}$ unlocated errors) can be corrected. The quantum Singleton bound states that $n-k \geq 2(d-1)$. A QECC that saturates the quantum Singleton bound is known as a quantum maximum distance separable (MDS) code.

Previous work has established that every AME$(2m,q)$ state is the purification of a quantum MDS code~\cite{Helwig:2013,Helwig:2012}.\footnote{The original proof actually demonstrates that AME$(2m,q)$ states are the purification of a threshold quantum secret sharing (QSS) scheme, however every pure QSS scheme is equivalent to a quantum MDS code~\cite{Cleve:1999} so the result follows immediately.} Furthermore, viewing the perfect tensor which describes an AME$(2m,q)$ state as a linear map from 1 leg to $2m-1$ legs, it is  the encoding isometry of the quantum MDS code encoding one logical qudit~\cite{Pastawski:2015}.

We can generalise the proof in~\cite{Helwig:2012} to further characterise the connection between (pseudo-)perfect tensors and QECC:

\begin{theorem}
Every AME$(t,q)$ state is the purification of a $[t-k,k,\floor*{\frac{t}{2}}-k+1]_q$ QECC for $1 \leq k \leq \floor*{\frac{t}{2}}$\footnote{The proof of this theorem is a straightforward generalisation of~\cite[Theorem 2]{Helwig:2012}}.
\end{theorem}
\begin{proof}

Let $\ket{\Phi}$ be an AME$(t,q)$ state, and let $m =  \floor*{\frac{t}{2}}$. For any partition of the state into disjoint sets $L$, $A$ and $B$ such that $|L| = k \leq m$, $|A| = m - k$ and $|B| = \ceil*{\frac{t}{2}}$ we can write:

\begin{equation} \label{eq}
\ket{\Phi} = \sqrt{\frac{1}{q^m}} \sum_{\bm{i} \in \mathbb{Z}_q^k, \bm{j} \in \mathbb{Z}_q^{m-k}}\ket{i_1...i_k}_L\ket{j_1...j_{m-k}}_A \ket{\phi(\bm{i,j})}_B
\end{equation}
%where we now have that $L$ is a fixed set of qudits such that $|L| = k \leq m$, and $A$ and $B$ are any sets of qudits disjoint from $L$ and each other such that $|A| = m-k$, $|B| = \ceil*{\frac{t}{2}}$.
The set $A \cup B$ are the physical qudits.
Define the basis states of a QECC as:
\begin{equation}
\begin{multlined}
\ket{\Phi_{\bm{i}}} = \sqrt{q^k}_L\bra{i_1...i_k}\ket{\Phi} \\
 = \sum_{\bm{j} \in \mathbb{Z}_q^{m-k}} \sqrt{\frac{1}{q^{m-k}}} \ket{j_1...j_{m-k}}_A \ket{\phi(\bm{i,j})}_B
 \end{multlined}
\end{equation}
Encode a logical state in the physical qudits as:
\begin{equation}\label{basis}
\ket{a} = \sum_{\bm{i} \in \mathbb{Z}_q^k} a_{\bm{i}} \ket{\bm{i}} \rightarrow \sum_{\bm{i} \in \mathbb{Z}_q^k} a_{\bm{i}} \ket{\Phi_{{\bm{i}}}}
\end{equation}

Now consider tracing out $m-k$ of the physical qudits. Since the sets $A$ and $B$ in \cref{eq} are arbitrary, we can always choose that the qudits we trace out are in the set $A$. The qudits we are left with are then in the state:
\begin{equation}
\rho_B = \frac{1}{q^{m-k}} \sum_{\bm{i,i',j}} a_{\bm{i}} a^*_{\bm{i'}} \ket{\phi(\bm{i,j})}\bra{\phi(\bm{i',j})}
\end{equation}
We can recover the logical state by performing the unitary operation:
\begin{equation}
U_B \ket{\phi(\bm{i,j})}_B =
\begin{cases}
\ket{i_1}_{B_{1}}...\ket{i_k}_{B_{k}} \ket{j_1}_{B_{k+1}}...\ket{j_{m-k}}_{B_{m}} & \text{for $t$ even}\\
\ket{i_1}_{B_{1}}...\ket{i_k}_{B_{k}} \ket{j_1}_{B_{k+1}}...\ket{j_{m-k}}_{B_{m}} \ket{0}_{B_{m+1}} & \text{for $t$ odd}
\end{cases}
\end{equation}

which will give:
\begin{equation}
U_B \rho_B U_b^\dagger = \ket{a}\bra{a} \otimes \frac{1}{q^{m-k}}\sum_{\bm{j}} \ket{\bm{j}}\bra{\bm{j}}
\end{equation}

Therefore, any set of $t-m$ qudits contains all the information about the logical state. By the no-cloning theorem, any set of $m$ qudits contains no information about the logical state, so the QECC can correct exactly $m -k  = \floor*{\frac{t}{2}} -k$ erasure errors. This gives $d = \floor*{\frac{t}{2}} -k +1$.
\end{proof}

Therefore, an AME$(2m,q)$ state is the purification of a quantum MDS code with  parameters $[2m-k,k,m-k+1]_q$; while an AME$(2m+1,q)$ state is the purification of a QECC with parameters $[2m+1-k,k,m-k+1]_q$. The parameters in the AME$(2m+1,q)$ case do not saturate the Singleton bound, so it is not an MDS code, but it is an optimal QECC.\footnote{The terms MDS quantum code and optimal quantum code are sometimes used interchangeably. Here, by an optimal quantum code we mean either an MDS code, or a code for which $n-k$ is odd so which cannot saturate the Singleton bound, but for which the distance $d$ is maximal given this constraint.}

If we consider the (pseudo-)perfect tensor, $T$, which describes an AME$(t,q)$ state $\ket{\Phi}$ we have:
\begin{equation}
\ket{\Phi} = \sum_{\bm{i} \in \mathbb{Z}_q^k,\bm{j}\in \mathbb{Z}_q^{t-k}} T_{\bm{i,j}} \ket{\bm{i}}_L \ket{\bm{j}}_P
\end{equation}
where $L$ and $P$ are the sets of logical and physical qudits in the corresponding QECC, $|L| = k$, $|P| = t-k$. The basis states for the QECC are then:
\begin{equation}
\ket{\Phi_{\bm{i}}} = \sum_{\bm{j} \in \mathbb{Z}_q^{t-k}} T_{\bm{i,j}} \ket{\bm{j}}_P
\end{equation}
and the encoding isometry is:
\begin{equation}
V = \sum_{\bm{i} \in \mathbb{Z}_q^k,\bm{j}\in \mathbb{Z}_q^{t-k}} T_{\bm{i,j}} \ket{\bm{j}}\bra{\bm{i}}
\end{equation}
So, viewed as an isometry from $k$ legs to $t-k$ legs a (pseudo-)perfect tensor is the encoding isometry of a $[t-k,k,\floor*{\frac{t}{2}}-k+1]_q$ QECC.

%\subsection{Existence of (pseudo-)perfect tensors}
%It is possible to construct $t$-index (pseudo-)perfect tensors for arbitrarily large $t$ by increasing $q$~\cite{Aharanov:1997, Grassl:2004, Helwig:2013}. In our constructions we are specifically interested in (pseudo-)perfect tensors which describe stabilizer states / codes. These can be constructed for arbitrary $t$ provided $q$ is chosen appropriately~\cite{Helwig:2013a}. A brief outline of the construction is given in \cref{stab_perf}.

\section{Qudit stabilizer codes and states} \label{appendix_3}
We restrict our attention to qudits of dimension $p$ where $p$ is an odd prime.

\subsection{Generalised Pauli group}
The generalised Pauli operators on $p$ dimensional qudits are defined as:
\begin{equation}
X = \sum_{j=0}^{p-1}\ket{j+1}\bra{j}
\end{equation}
\begin{equation}
Z = \sum_{j=0}^{p-1} \omega^{j} \ket{j}\bra{j}
\end{equation}
where $\omega = e^{\frac{2\pi i}{p}}$. The generalised Pauli operators obey the relations $X^p = Z^p = \identity$ and $XZ = \omega ZX$.

The Pauli group on $n$ qudits is given by $\mathcal{G}_{n,p} = \langle \omega^a X^{\bm{b}} Z^{\bm{c}}\rangle$ where  $a \in \mathbb{Z}_p$, $\bm{b,c} \in \mathbb{Z}_p^n$. Two elements $\omega^a X^{\bm{b}} Z^{\bm{c}}$ and $\omega^{a'} X^{\bm{b'}} Z^{\bm{c'}}$ commute if and only if $\bm{b'}\cdot \bm{c} = \bm{b}\cdot \bm{c'}$, where all addition is mod $p$.

\subsection{Qudit stabilizer codes} \label{codes_section}
A stabilizer code $\mathcal{C}$ on $n$ qudits is a $p^k$-dimensional subspace of the Hilbert space given by:
\begin{equation}
\mathcal{C} = \{\ket{\psi} \mid M\ket{\psi} = \ket{\psi} \forall M \in S\}
\end{equation}
 where $S$ is an Abelian subgroup of $\mathcal{G}_{n,p}$ that does not contain $\omega \identity$.

The projector onto $\mathcal{C}$ is given by~\cite{Gheorghiu:2014}:
\begin{equation}
\Pi = \frac{1}{|S|} \sum_{M \in S}M
\end{equation}
where $|S| = p^{n-k}$. $S$ is an elementary Abelian $p$-group, so this implies that a minimal generating set for $S$ contains $n-k$ elements~\cite{finite_groups}.

The minimum weight of a logical operator in an $[n,k,d]$ stabilizer code is $d$. This is also the minimum weight of any operator that is not in the stabilizer, but which commutes with every element of the stabilizer.

A stabilizer code with $k=0$ is a stabilizer state.

\section{Stabilizer (pseudo-)perfect tensors} \label{appendix_4}
A stabilizer (pseudo-)perfect tensor describes stabilizer AME states. This implies, using the method in~\cite{Gottesman:1997} for generating short qubit stabilizer codes from longer ones, that the QECCs described by the tensors are stabilizer codes:

\begin{theorem} \label{stab_tensor}
If a (pseudo-)perfect tensor, $T$, with $t$ legs describes a stabilizer AME$(t,p)$ state, then the $[t-k,k,\floor*{\frac{t}{2}}-k+1]_p$ QECCs (for $1 \leq k \leq \floor*{\frac{t}{2}}$) described by the tensor are stabilizer codes. The stabilizers of the code are given by the stabilizers of the AME state which start with $I^{\otimes k}$, restricted to the last $t-k$ qudits.
\end{theorem}
\begin{proof}
Consider an AME$(t,p)$ stabilizer state with stabilizer $S$:
\begin{equation}
\ket{\Phi} = \sqrt{\frac{1}{p^m}} \sum_{\bm{i} \in \mathbb{Z}_p^m} \ket{i_1...i_m}_A \ket{\phi(\bm{i})}_B
\end{equation}
where $|A| = m = \floor*{\frac{t}{2}}$, $|B| = \ceil*{\frac{t}{2}}$.

We have that $M \ket{\Phi} = \ket{\Phi}$ for all $M \in S$, where $|S| = p^t$ so a minimal generating set for $S$ contains $t$ elements. We can always pick a generating set for $S$ so that $M_1$ and $M_2$ begin with $X$ and $Z$ respectively, and $M_3$ to $M_t$ begin with $\identity$~\cite{Gottesman:1997}. Define $M_j'$ to be $M_j$ restricted to the last $t-1$ qudits, where $j = 1,\dots,t$.

Consider the codespace, $\mathcal{C}$, for a $[t-1,1,\floor*{\frac{t}{2}}]_p$ error correcting code described by $T$:
\begin{equation}
\ket{\psi} = \sum_{i_1 \in \mathbb{Z}} a_{i_1} \ket{\Phi_{i_1}}
\end{equation}
where $\sum_i a_i^2 = 1$ and $\ket{\Phi_{i_1}} = \sqrt{p} \bra{i_1}\ket{\Phi}$. If we act on $\ket{\psi} \in \mathcal{C}$ with $M_j'$ we find $M_j' \ket{\psi} = \ket{\psi}$ for $j = 3,\dots,t$, $M_j'\ket{\psi} = \ket{\psi'} \in \mathcal{C}$ where $\ket{\psi'} \neq \ket{\psi}$ for $j=1,\dots,2$.

The group $S'$ generated by $M_j'$ for $j=3,\dots,t$ contains $|S'| = p^{t-2}$ elements, and it stabilizes the $[t-1,1,\floor*{\frac{t}{2}}]_p$ code described by $T$ with codespace $\mathcal{C}$.

This procedure for discarding two stabilizer generators from a $[n,k,d]$ code to obtain an $[n-1,k+1,d-1]$ code is always possible provided $d>1$~\cite{Gottesman:1997}. So we can repeat the procedure $\floor*{\frac{t}{2}} - 1$ times, demonstrating that the $[t-k,k,\floor*{\frac{t}{2}}-k+1]_p$ QECCs for $1 \leq k \leq \floor*{\frac{t}{2}}$ described by the perfect tensor are stabilizer codes.
\end{proof}

We also require that all the QECC used in our construction map logical Pauli operators to physical Pauli operators. It is known that for qubit stabilizer codes a basis can always be chosen so that this is true~\cite{Gottesman:1997}, and the same group-theoretic proof applies to qudit stabilizer codes.\footnote{The discussion in~\cite{Gottesman:1997} actually shows that there is an automorphism between $\mathcal{G}_{k,2}$ and $N(S)/S$, where $N(S)$ is the normalizer of $S$ in $\mathcal{G}_{n,2}$. As $N(S) \in \mathcal{G}_{n,2}$ this is sufficient. The discussion in~\cite{Gottesman:1997} can be extended to qudits of prime dimension by replacing phase factors of 4 with factors of $p$, and dimension factors of 2 with factors of $p$.}
The physical Pauli operators we obtain using this method are not given by acting on the logical Pauli operators with the encoding isometry, but they have the same action in the code subspace. So, we have that for qudit stabilizer codes it is always possible to pick a basis where $V P V^\dagger = P'VV^\dagger$ where $P$ is a $k$-qudit Pauli operator, $P'$ is an $n$-qudit Pauli operator, and $V$ is the encoding isometry of the QECC.

In our holographic QECC we do not have complete freedom to pick a basis, so we also need to show that we can pick this basis consistently. In order to show this we will require two lemmas about qudit stabilizer codes.

\begin{lemma} \label{sub_group}
The smallest subgroup, $G$, of the Pauli group $\mathcal{G}_{n,p}$ such that $\forall P \in \mathcal{G}_{n,p}$, $P \neq \identity^{\otimes n}$, $\exists M \in G$ where $MP \neq PM$ is the entire Pauli group.
\end{lemma}
\begin{proof}
Consider the following process for constructing a set $A$ element by element such that $\forall P \in \mathcal{G}_{n,p}$, $P \neq \identity^{\otimes n}$, $\exists M \in A$ where $MP \neq PM$:

\begin{enumerate}
\item Select an arbitrary element of $\mathcal{G}_{n,p}$, $P^{(1)} = \omega^{a^{(1)}} X^{\bm{b}^{(1)}} Z^{\bm{c}^{(1)}}$.
\item Pick an element $P^{(1')} = \omega^{a^{(1')}} X^{\bm{b}^{(1')}} Z^{\bm{c}^{(1')}}$ such that $\bm{b}^{(1')}\cdot \bm{c}^{(1)} \neq \bm{b}^{(1)}\cdot \bm{c}^{(1')}$. This ensures that $P^{(1')}$ does not commute with $P^{(1)}$, and $P^{(1')}$ is our first element of $A$.
\item Pick an arbitrary element, $P^{(i)}$ of $\mathcal{G}_{n,p}$ which commutes with every element of $A$ and is not the identity.
\item Choose any element, $P^{(i')}$ of $\mathcal{G}_{n,p}$ which does not commute with $P^{(i)}$, and add it to $A$.
\item Repeat steps (3) and (4) until $\forall P \in \mathcal{G}_{n,p}$, $P \neq \identity^{\otimes n}$, $\exists M \in A$ such that  $MP \neq PM$.
\end{enumerate}

When we construct $A$, every element $P^{(i')} = \omega^{a^{(i')}} X^{\bm{b}^{(i')}} Z^{\bm{c}^{(i')}}$ which we add to $A$ is independent from every element already in $A$. To see this note that by assumption there is some $P^{(i)}$ which commutes with every element in $A$, but does not commute with $P^{(i')}$. If $P^{(i')}$ was not independent from the other elements of $A$, and $\bm{b}^{(i')} = \sum_{k} \bm{b}^{(k)}$ and $\bm{c}^{(i')} = \sum_k \bm{c}^{(k)}$ where $P^{(k)} \in A$ for all $k$, then $\bm{b}^{(i')}\cdot \bm{c}^{(i)} =\sum_{k} \bm{b}^{(k)}\cdot \bm{c}^{(i)} = \sum_k \bm{b}^{(i)}\cdot \bm{c}^{(k)}$ and $\bm{b}^{(i)}\cdot \bm{c}^{(i')} = \sum_k \bm{b}^{(i)}\cdot \bm{c}^{(k)}$ so  $\bm{b}^{(i')}\cdot \bm{c}^{(i)} = \bm{b}^{(i)}\cdot \bm{c}^{(i')}$, and $P^{(i')}$ would commute with $P^{(i)}$, contradicting our initial assumption.

We need to determine the minimum number of elements in $A$ when this process terminates.

Suppose we have repeated steps (3) and (4) $m$ times, so that $|A| = m$. If there is an element $P^{(k)} = \omega^{a^{(k)}} X^{\bm{b}^{(k)}} Z^{\bm{c}^{(k)}}$ of $\mathcal{G}_{n,p}$ which commutes with every element of $A$ then $\bm{b}^{(i')}\cdot \bm{c}^{(k)} = \bm{b}^{(k)}\cdot \bm{c}^{(i')}$ for all $P^{(i')} \in A$. $P^{(k)}$ is described by $2n$ degrees of freedom: $b_1^{(k)},...b_n^{(k)}$ and $c_1^{(k)},...c_n^{(k)}$, and there are $m$ homogeneous equations which $P^{(k)}$ needs to satisfy.\footnote{The equations are homogeneous as all constant terms are equal to zero. Homogeneity of the equations ensures that the equations are not inconsistent.} Provided $m < 2n$, the set of equations is underdetermined, and we can always choose a $P^{(k)}$ which commutes with every element of $A$ and is not the identity. If $m=2n$ then the solution to the equations is uniquely determined, and is the identity. At this point we cannot continue with the process, so it terminates with $|A| = 2n$.

Therefore, the smallest set $A$ of elements of $\mathcal{G}_{n,p}$ such that $\forall P \in \mathcal{G}_{n,p}$, $P \neq \identity^{\otimes n}$, $\exists M \in A$ where $MP \neq PM$ contains $2n$ elements. At this stage $A$ is not a group because all the elements of $A$ are independent so the set isn't closed. Any $2n$ independent elements of $\mathcal{G}_{n,p}$ generate the entire group, so the smallest group $G$ which contains every element of $A$ is $\mathcal{G}_{n,p}$ itself.
\end{proof}

\begin{lemma} \label{choose}
In an $[n,k,d]_p$ stabilizer code, the action of an encoded Pauli operator $\overline{P}$ on any $d-1$ physical qudits can be chosen to be any element of $\mathcal{G}_{d-1,p}$.
\end{lemma}
\begin{proof}
The encoded Pauli operator $\overline{P}$ is not unique, and the different possible physical operators are related by elements of the stabilizer. We therefore need to show that the stabilizer $S$ restricted to any set of $d-1$ qudits is the entire Pauli group $\mathcal{G}_{d-1,p}$.

An $[n,k,d]_p$ stabilizer code can correct all Pauli errors of weight $d-1$ and less. A correctable error doesn't commute with some element of the stabilizer, so for any Pauli operator of weight $d-1$ there $\exists M \in S$ such that $[M,P] \neq 0$. The result follows immediately from \cref{sub_group}.
\end{proof}

\begin{theorem} \label{pauli_rank}
If there exists a basis such that the QECC described by a (pseudo-)perfect tensor, $T$, from qudit $l$ to $t-1$ qudits maps Pauli operators to Pauli operators, then all other QECC described by $T$ which include qudit $l$ in the logical set also map Pauli operators to Pauli operators in that basis.
\end{theorem}
\begin{proof}
Let the AME state described by $T$ be given by:
\begin{equation}
\ket{\Phi} = \sum_{a \in \mathbb{Z}_p} \sum_{\bm{b} \in \mathbb{Z}_p^{k-1}} \sum_{\bm{c} \in \mathbb{Z}_p^{m-k}} \ket{a}_l \ket{b_1...b_{k-1}}_L \ket{c_1...c_{m-k}}_A \ket{\phi(a,\bm{b},\bm{c})}_B
\end{equation}
where $m = \floor*{\frac{t}{2}}$, and $|B| = \ceil*{\frac{t}{2}}$.

The basis states of the $[t-1,1,\floor*{\frac{t}{2}}]_p$ QECC from qudit $l$ to other $t-1$ qudits are given by:
\begin{equation}
\ket{\Phi_a} = \sum_{\bm{b} \in \mathbb{Z}_p^{k-1}} \sum_{\bm{c} \in \mathbb{Z}_p^{m-k}} \ket{b_1...b_{k-1}}_L \ket{c_1...c_{m-k}}_A \ket{\phi(a,\bm{b},\bm{c})}_B
\end{equation}

and the encoding isometry is given by:
\begin{equation}
V = \sum_a \ket{\Phi_a}\bra{a}
\end{equation}

The basis states of a $[t-k,k,\floor*{\frac{t}{2}}-k+1]_p$ QECC from a set $L$ qudits (where $l \in L$, $|L| = k$) to $t-k$ qudits is given by:
\begin{equation}
\ket{\Phi_{a,\bm{b}}} = \sum_{\bm{c} \in \mathbb{Z}_p^{m-k}} \ket{c_1...c_{m-k}}_A \ket{\phi(a,\bm{b},\bm{c})}_B
\end{equation}

and the encoding isometry is given by:
\begin{equation}
V' = \sum_a \ket{\Phi_{a,\bm{b}}}\bra{a,\bm{b}}
\end{equation}

By assumption we have:
\begin{equation}
VP_1V^\dagger = QVV^\dagger
\end{equation}
where $P_1 \in \mathcal{G}_{1,p}$ and $Q \in \mathcal{G}_{n,p}$, $n = t-1$.

Therefore:
\begin{equation}
\sum_{aa'} \ket{\Phi_a}\bra{a} P_1 \ket{a'}\bra{\Phi_{a'}} = Q\sum_{a} \ket{\Phi_a}\bra{\Phi_{a}}
\end{equation}

Consider the action of $V'$ on $P_1 \otimes P_2 \in \mathcal{G}_{k,p}$:

\begin{equation}
\begin{split}
V'(P_1 \otimes P_2)V'^\dagger & = \sum_{a,a',\bm{b,b'}}\ket{\Phi_{a,\bm{b}}}\bra{a,\bm{b}}(P_1 \otimes P_2 )\ket{a,\bm{b}}\bra{\Phi_{a,\bm{b}}} \\
& = \sum_{a,a',\bm{b,b'}} \bra{\bm{b}}\ket{\Phi_a} \bra{a}P_1\ket{a'} \bra{\bm{b}}P_2 \ket{\bm{b'}} \bra{\Phi_{a'}}\ket{\bm{b'}} \\
& = \sum_{a,a',\bm{b,b'}} \bra{\bm{b}}P_2 \ket{\bm{b'}} \bra{\bm{b}}\left(  \ket{\Phi_a}\bra{a} P_1 \ket{a'}\bra{\Phi_{a'}}  \right) \ket{\bm{b'}} \\
& = \sum_{a,\bm{b,b'}} \bra{\bm{b}}P_2 \ket{\bm{b'}} \bra{\bm{b}}\left( Q \ket{\Phi_a}\bra{\Phi_{a}}\right) \ket{\bm{b'}} \\
& = \sum_{a,\bm{b,b',b''}} \bra{\bm{b}}P_2 \ket{\bm{b'}}  \bra{\bm{b}}Q_2'\ket{\bm{b''}} Q' \ket{\Phi_{a,\bm{b''}}}\bra{\Phi_{a,\bm{b'}}} \\
& = \sum_{a,\bm{b}} Q' \ket{\Phi_{a,\bm{b}}}\bra{\Phi_{a,\bm{b}}} \text{   if }P_2 = Q_2' \\
& = V'V'^\dagger P'^{(n')}
\end{split}
\end{equation}
where and $Q_2'$ and $Q'$ indicate $Q$ restricted to the first $k-1$ and remaining $n-k-1$ qudits respectively ($Q = Q_2' \otimes Q'$).

Therefore, if $Q$ acts as $P_2$ on the first $k-1$ qudits, then $P_1 \otimes P_2$ maps to a Pauli under $V'$. The operator $Q$ is not unique, and from \cref{choose} we know that its action on $\floor*{\frac{t}{2}}-1$ qudits can be chosen to be any element of $\mathcal{G}_{\floor{\frac{t}{2}}-1,p}$. So we can choose that $Q$ acts as $P_2$ on the first $k-1$ qudits, for $k \leq \floor*{\frac{t}{2}}$.
\end{proof}

\section{Existence of (pseudo-)perfect stabilizer tensors}\label{appendix_5}

\subsection{Classical coding theory}
A classical linear $[n,k]_d$ code, $\mathcal{C}_{cl}$, encodes $k$ $d$-dimensional dits of information in $n$ dits. It can be described by a generator matrix $G^T: \mathbb{Z}_p^k \rightarrow \mathbb{Z}_d^n$, where information is encoded as $\bm{x} \rightarrow G^T \bm{x}$ for $\bm{x}\in \mathbb{Z}_d^k$. Equivalently, $\mathcal{C}_{cl}$ admits a description as the kernel of a parity check matrix $H: \mathbb{Z}_d^n \rightarrow \mathbb{Z}_d^{n-k}$. Consistency of the two descriptions implies $HG^T\bm{x} = 0$, $\forall \bm{x} \in \mathbb{Z}_d^n$, and hence the rows of $H$ are orthogonal to the rows of $G$.

The minimum distance $\delta$ of a classical code is defined as the minimum Hamming distance between any two code words. It is bounded by the classical Singleton bound, $\delta \leq n - k + 1$. Codes which saturate the classical Singleton bound are referred to as classical MDS codes.

Reed-Solomon codes are a class of classical MDS codes~\cite{Reed:1960}.\footnote{Reed Solomon codes can be defined over any finite field, but we only require the definition of Reed Solomon codes over $\mathbb{Z}_p$ for our construction.}

\begin{definition}
Let $p$ be a prime, and let $k,n$ be integers such that $k < n \leq p$. For a set $S = \{\alpha_1,\alpha_2,...,\alpha_n\} \in \mathbb{Z}_p$, the Reed-Solomon code over $\mathbb{Z}_p$ is defined as:
\begin{equation}
\mathcal{C}_{RS}[n,k] = \{\left(P(\alpha_1),P(\alpha_2),...,P(\alpha_n)  \right) \in \mathbb{Z}_p^n \mid P(X) \in \mathbb{Z}_p[X], \deg(P) \leq k-1 \}
\end{equation}
where $\mathbb{Z}_p[X]$ is the polynomial ring in $X$ over $\mathbb{Z}_p$.\footnote{The polynomial ring in $X$ over $\mathbb{Z}_p$, $\mathbb{Z}_p[X]$, is the set of polynomials $P(X) = a_0 + a_1X+a_2X^2+...+a_mX^m$ where $a_i \in \mathbb{Z}_p$.}
\end{definition}

To encode a message $\bm{a} = (a_0,a_1,...,a_{k-1}) \in \mathbb{Z}_p^{k}$ in the Reed-Solomon code define the polynomial:
\begin{equation}
P_{\bm{a}}(X) = a_0 + a_1 X + a_2X^2 + ...+a_{k-1}X^{k-1}
\end{equation}
and construct the codeword $\left(P_{\bm{a}}(\alpha_1), P_{\bm{a}}(\alpha_2),...,P_{\bm{a}}(\alpha_n) \right) \in \mathcal{C}_{RS}[n,k] $.

Reed-Solomon codes are linear codes, with generating matrix:
\begin{equation}
G = \left(
\begin{array}{c c c c c }
1 & 1 & 1 & .\hdots & 1 \\
\alpha_1 & \alpha_2 & \alpha_3 & ... & \alpha_n \\
\alpha_1^2 & \alpha_2^2 & \alpha_3^2 & ... & \alpha_n^{2} \\
\vdots & \vdots & \vdots & \ddots & \vdots \\
 \alpha_1^{k-1} & \alpha_2^{k-1} & \alpha_3^{k-1} & ... & \alpha_n^{k-1}
\end{array}
\right)
\end{equation}

The generator matrix can be put into standard form $G = [I_k | P]$ (where $P$ is a $k \times (n-k)$ matrix) using Gauss-Jordan elimination over the field $\mathbb{Z}_p$. The parity check matrix is then given by $H = [P^T|I_{n-k}]$.

\subsection{Constructing AME stabilizer states} \label{stab_perf}
An AME$(t,p)$ stabilizer state $\ket{\Phi}$ can be constructed from a classical $[t,\floor*{\frac{t}{2}}]_p$ MDS code with $\delta = \ceil*{\frac{t}{2}}+1$. The state is given by~\cite{Helwig:2013a}:
\begin{equation}
\ket{\Phi} = \frac{1}{d^{\frac{l}{2}}} \sum_{\bm{x}\in \mathbb{Z}_p^l} \ket{G \bm{x}}
\end{equation}
and has stabilizers $X^{G\bm{y}}$ for all $\bm{y} \in \mathbb{Z}_p^l$, and $Z^{\bm{y}}$ where $\bm{y}^T = \bm{z}^TH$ for all $\bm{z} \in \mathbb{Z}_p^m$. The full set of stabilizers is given by the generator matrix~\cite{Helwig:2013a}:
\begin{equation}
M = \left(
\begin{array}{@{}c|c@{}}
G & 0 \\
0 & H
\end{array}
\right)
\end{equation}
where $(\bm{\alpha} \mid \bm{\beta}) \equiv X^{\bm{\alpha}} \cdot Z^{\bm{\beta}}$ for $\bm{\alpha}, \bm{\beta} \in \mathbb{Z}_p^t$.

Reed-Solomon codes can be constructed for any $k,n$ satisfying $k < n  \leq p$~\cite{Reed:1960}, so by increasing $p$ this construction can provide AME$(t,p)$ stabilizer states for arbitrarily large $t$. By \cref{stab_tensor} the tensor which describes the AME$(t,p)$ stabilizer states will be a stabilizer (pseudo-)perfect tensor.

This construction is not optimised to minimise $p$ for a given $t$. It is possible to construct generalised Reed-Solomon codes which exist for $n = p+1$~\cite{Seroussi:1986}, which if used in this construction will give (pseudo-)perfect stabilizer tensors acting on lower dimensional qudits for certain values of $t$. There are also methods for constructing stabilizer perfect tensors for which $p \propto \sqrt{t}$ using cyclic and constacyclic classical MDS codes~\cite{Grassl:2015a}, but this method is significantly more involved than the one presented here, and does not work for pseudo-perfect tensors. In our construction there is no benefit to minimising $p$, so we have selected the simplest, most universal method for constructing (pseudo-)perfect stabilizer tensors.

\section{Perturbative simulations} \label{appendix_6}
In this appendix we collect some results regarding perturbative techniques, and introduce the new qudit perturbation gadgets which are used in this paper. All the perturbation gadgets we introduce are qudit generalisations of the qubit gadgets from~\cite{oliveira:2005}.

Let $\mathcal{H}$ be a Hilbert space decomposed as $\mathcal{H} = \mathcal{H}_- \oplus \mathcal{H}_+$. Let $\Pi_\pm$ be the projectors onto $\mathcal{H}_\pm$. For arbitrary operator $M$ define $M_{++} = \Pi_+M\Pi_+$, $M_{--} = \Pi_-M\Pi_-$, $M_{+-} = \Pi_+M\Pi_-$, and $M_{-+} = \Pi_-M\Pi_+$.

Consider an unperturbed Hamiltonian $H = \Delta H_0$, where $H_0$ is block-diagonal with respect to the split $\mathcal{H} = \mathcal{H}_- \oplus \mathcal{H}_+$, $(H_0)_{--} = 0$, $\lambda_{\text{{min}}}\left((H_0)_{++} \right) \geq 1$.

We will use \cref{second_order,third_order} from \cref{perturbative_simulations} to construct qudit perturbation gadgets, which generalise qubit gadgets from~\cite{oliveira:2005}. In our analysis we assume without loss of generality that every interaction is a Pauli-rank 2 interaction of the form $P_a + P_a^\dagger$.

\noindent\paragraph{Qudit subdivision gadget}\\
The subdivision gadget is used to simulate a $k$-local interaction by interactions which are at most $\ceil*{\frac{k}{2}}+1$-local.
We want to simulate the Hamiltonian:
\begin{equation}
H_{\text{target}} = H_{\text{else}} + (P_{A}\otimes P_{B} + P_{A}^\dagger\otimes P_{B}^\dagger)
\end{equation}
Let $\tilde{H} = H+V$ where:
\begin{equation}
H = \Delta \Pi_+
\end{equation}
\begin{equation}
V = H_1 + \Delta^{\frac{1}{2}}H_2
\end{equation}
where:
\begin{equation}
\Pi_+ = \ket{1}\bra{1}_w+\ket{2}\bra{2}_w+...+\ket{p-1}\bra{p-1}_w
\end{equation}
\begin{equation}
H_1 = H_{\text{else}} + 2\identity
\end{equation}
\begin{equation}
H_2 = -P_{A}\otimes X_w - P_A^\dagger \otimes X_w^\dagger + P_B\otimes X_w^\dagger + P_B^\dagger\otimes X_w
\end{equation}
The degenerate ground space of $H$ has the mediator qubit $w$ in the state $\ket{0}\bra{0}$ so $\Pi_- = \ket{0}\bra{0}_w$. This gives:

\begin{equation}
(H_{1})_{--}= \left(H_{\text{else}} + 2\identity \right) \otimes \ket{0}\bra{0}_w
\end{equation}
and:
\begin{equation}
(H_2)_{-+} = -P_A \otimes \ket{0}\bra{p-1}_w - P_A^\dagger \otimes \ket{0}\bra{1}_w + P_B \otimes \ket{0}\bra{1}_w + P_B^\dagger \otimes \ket{0}\bra{p-1}_w
\end{equation}
Therefore:
\begin{equation}
(H_2)_{-+}H_0^{-1}(H_2)_{+-} = \left(P_A \otimes P_B + P_A^\dagger \otimes P_B^\dagger - 2\identity  \right) \otimes \ket{0}\bra{0}_w
\end{equation}

If we define an isometry $W$ by $W\ket{\psi}_A = \ket{\psi}_A\ket{0}_w$ then:
\begin{equation}
 || W H_{\text{target}} W^\dagger - (H_1)_{--} + (H_2)_{-+}H_0^{-1}(H_2)_{+-}|| =  0
\end{equation}
Therefore \cref{second_order_eq} is satisfied for all $\epsilon \geq 0$. So, provided a $\Delta$ is picked which satisfies the conditions of \cref{second_order}, $\tilde{H}$ is a $(\frac{\Delta}{2},\eta,\epsilon)$-simulation of $H_{\text{target}}$.

%The self energy is then given by:
%\begin{equation}
%\Sigma_-(z) = \left( H_{\text{else}}' - \frac{\Delta}{z - \Delta} \left(P_A \otimes P_B + P_A^\dagger \otimes P_B^\dagger - 2\identity  \right) \right) \otimes\ket{0}\bra{0}_w + O\left(\frac{||V||^3}{(z- \Delta)^2} \right)
%\end{equation}
%Expand the self-energy about $z = 0$ and let $H_{\text{eff}} = H_{\text{target}} \otimes \ket{0}\bra{0}_w$:
%\begin{equation}
%||\Sigma_-(z) - H_{\text{eff}}|| = O\left(\frac{|z|r^2}{\Delta} \right) + O\left(\frac{||V||^2}{\Delta^2} \right) + O\left(\frac{|z| ||V||^3}{\Delta^3} \right)
%\end{equation}
%where $r = \max(||P_A||,||P_B||)$.
%
%Following the same analysis as in~\cite{oliveira:2005}, if we choose:
%\begin{equation}
%\Delta \geq \frac{\left(||H_{\text{else}}|| + \sqrt{2}r \right)^6}{\epsilon^2}
%\end{equation}
%then:
%\begin{equation}
%||\Sigma_-(z) - H_{\text{eff}}|| = O(\epsilon)
%\end{equation}
%so the $jth$ eigenvalue of $\tilde{H}$ is $\epsilon$-close to the $jth$ eigenvalue of $H_{\text{eff}}$.

\noindent\paragraph{Qudit 3-2 gadget}\\
The Hamiltonian we want to simulate is:
\begin{equation}
H_{\text{target}} = H_{\text{else}} + P_A\otimes P_B\otimes P_C + P_A^\dagger \otimes P_B^\dagger \otimes P_C^\dagger
\end{equation}
Let $\tilde{H} = H+V$ where:
\begin{equation}
H = \Delta \Pi_+
\end{equation}
\begin{equation}
V = H_1 + \Delta^{\frac{1}{3}} H_1' + \Delta^{\frac{2}{3}}H_2
\end{equation}
where:
\begin{equation}
\Pi_+ = \ket{1}\bra{1}_w + ... + \ket{p-1}\bra{p-1}_w
\end{equation}
\begin{equation}
\begin{split}
H_1 = H_{\text{else}} + \frac{1}{2}\left(P_A^2 \otimes P_C + (P_A^\dagger)^2\otimes P_C^\dagger + P_B^2 \otimes P_C + (P_B^\dagger)^2\otimes P_C^\dagger \right) \\
  + \frac{1}{2\sqrt{2}} \left[\left(-P_A+P_B\right)^2 \left(-P_A^\dagger+P_B^\dagger \right) + \left(-P_A^\dagger+P_B^\dagger\right)^2 \left(-P_A+P_B \right) \right]
\end{split}
\end{equation}
\begin{equation}
H_1' = - P_A\otimes P_B^\dagger - P_A^\dagger \otimes P_B
\end{equation}
\begin{equation}
H_2 =  \left(P_C \otimes \ket{p-1}\bra{1}_w + P_C^\dagger \otimes \ket{1}\bra{p-1}_w \right) + \frac{1}{\sqrt{2}} \left( (-P_A+P_B) \otimes X_w + (-P_A^\dagger+P_B^\dagger)\otimes X_w^\dagger \right)
\end{equation}

%\begin{equation}
%\begin{split}
%V = H_{\text{else}} + V_{\text{extra}} - \Delta^{\frac{2}{3}} \left(P_C \otimes \ket{p-1}\bra{1}_w + P_C^\dagger \ket{1}\bra{p-1}_w \right) \\
%+ \frac{\Delta^{\frac{2}{3}}}{\sqrt{2}} \left( (-P_A+P_B) \otimes X_w + (-P_A^\dagger+P_B^\dagger)\otimes X_w^\dagger \right)
%\end{split}
%\end{equation}
%where $ \Delta \Pi_+ = \ket{1}\bra{1}_w + ... + \ket{p-1}\bra{p-1}_w$, and the correction term $V_{\text{extra}}$ is given by:
%\begin{equation}
%\begin{split}
%  V_{\text{extra}} = \Delta^{\frac{1}{3}} \left(P_A\otimes P_B^\dagger + P_A^\dagger \otimes P_B  \right) + \left(P_A^2 \otimes P_C + (P_A^\dagger)^2\otimes P_C^\dagger + P_B^2 \otimes P_C + (P_B^\dagger)^2\otimes P_C^\dagger \right) \\
%  + \frac{1}{\sqrt{2}} \left[\left(P_A+P_B)\right)^2 \left(P_A^\dagger+P_B^\dagger \right) + \left(P_A^\dagger+P_B^\dagger)\right)^2 \left(P_A+P_B \right) \right]
%\end{split}
%\end{equation}
This gives:
\begin{equation}
(H_1)_{--} = H_1 \otimes \ket{0}\bra{0}_w
\end{equation}
\begin{equation}
(H_1')_{--} = H_1' \otimes \ket{0}\bra{0}_w
\end{equation}
\begin{equation}
\begin{split}
(H_2)_{++} = - \left(P_C \otimes \ket{p-1}\bra{1}_w + P_C^\dagger \otimes \ket{1}\bra{p-1}_w \right)
+ \frac{1}{\sqrt{2}} \left[ \left(-P_A+P_B \right)\otimes Q_w + \left(-P_A^\dagger+P_B^\dagger \right)\otimes R_w \right]
\end{split}
\end{equation}
where $Q = \ket{2}\bra{2}+...+\ket{p-1}\bra{p-1}$ and $R = \ket{1}\bra{1}+...+\ket{p-2}\bra{p-2}$.

\begin{equation}
(H_2)_{-+} = \frac{1}{\sqrt{2}}\left[ \left(-P_A + P_B \right)\otimes \ket{0}\bra{p-1}_w + \left(-P_A^\dagger + P_B^\dagger \right)\otimes \ket{0}\bra{1}_w \right]
\end{equation}
If we define an isometry $W$ by $W\ket{\psi}_A = \ket{\psi}_A\ket{0}_w$ then:
\begin{equation}
|| W H_{\text{{target}}} W^\dagger - (H_1)_{--} + (H_2)_{-+}H_0^{-1}(H_2)_{++}H_0^{-1}(H_2)_{+-}|| = 0
\end{equation}
Therefore \cref{third_order_eq1} is satisfied for all $\epsilon \geq 0$.

We also have:
\begin{equation}
(H_2)_{-+}H_0^{-1}(H_2)_{+-} = - P_A\otimes P_B^\dagger - P_A^\dagger \otimes P_B = (H_1')_{--}
\end{equation}
As required by \cref{third_order_eq2}. So, provided a $\Delta$ is picked which satisfies the conditions of \cref{third_order}, $\tilde{H}$ is a $(\frac{\Delta}{2},\eta,\epsilon)$-simulation of $H_{\text{target}}$.

\noindent\paragraph{Qudit crossing gadget}\\
We want to generate the Hamiltonian:
\begin{equation}
H_{\text{target}} = H_{\text{else}} + \alpha_{ad}\left(P_A\otimes P_D + P_A^\dagger \otimes P_D^\dagger\right) + \alpha_{bc}\left(P_B\otimes P_C +  P_B^\dagger \otimes P_C^\dagger\right)
\end{equation}

Set $\tilde{H} = H+V$ where:
\begin{equation}
H = \Delta \Pi_+
\end{equation}
\begin{equation}
V = H_1 + \Delta^{\frac{1}{2}}H_2
\end{equation}

where:
\begin{equation}
\Pi_+ = \ket{1}\bra{1}_w + ... + \ket{p-1}\bra{p-1}_w
\end{equation}
\begin{equation}
\begin{split}
H_1 = H_{\text{else}} + [\alpha_{ad}\alpha_{bc}(P_A\otimes P_B^\dagger + P_A^\dagger \otimes P_B) - \alpha_{ad}(P_A\otimes P_C^\dagger + P_A^\dagger \otimes P_C) \\
- \alpha_{bc}(P_B\otimes P_D^\dagger + P_B^\dagger \otimes P_D) + (P_C\otimes P_D^\dagger + P_C^\dagger \otimes P_D) \\
+ \identity(\alpha_{ad}^2 + \alpha_{bc}^2+2)]
\end{split}
\end{equation}
\begin{equation}
\begin{split}
H_2 = \frac{1}{\sqrt{2}}[-\alpha_{ad} (P_A\otimes X_w + P_A^\dagger \otimes X_w^\dagger) - \alpha_{bc}  (P_B\otimes X_w + P_B^\dagger \otimes X_w^\dagger) \\
+ (P_C\otimes X_w^\dagger + P_C^\dagger \otimes X) + (P_D\otimes X_w^\dagger + P_D^\dagger \otimes X_w)  ]
\end{split}
\end{equation}

%\begin{equation}
%\begin{split}
%V = \sqrt{\frac{\Delta}{2}} [ -\alpha_{ad} (P_A\otimes X_w + P_A^\dagger \otimes X_w^\dagger) - \alpha_{bc}  (P_B\otimes X_w + P_B^\dagger \otimes X_w^\dagger) \\
%+ (P_C\otimes X_w^\dagger + P_C^\dagger \otimes X) + (P_D\otimes X_w^\dagger + P_D^\dagger \otimes X_w)  ] + H_{\text{else}}'
%\end{split}
%\end{equation}
%where:
%\begin{equation}
%\begin{split}
%H_{\text{else}}' = H_{\text{else}} + [\alpha_{ad}\alpha_{bc}(P_A\otimes P_B^\dagger + P_A^\dagger \otimes P_B) - \alpha_{ad}(P_A\otimes P_C^\dagger + P_A^\dagger \otimes P_C) \\
%- \alpha_{bc}(P_B\otimes P_D^\dagger + P_B^\dagger \otimes P_D) + (P_C\otimes P_D^\dagger + P_C^\dagger \otimes P_D) \\
%+ \identity(\alpha_{ad}^2 + \alpha_{bc}^2+2)]
%\end{split}
%\end{equation}
Then:
\begin{equation}
(H_1)_{--} = H_1\otimes \ket{0}\bra{0}_w
\end{equation}
\begin{equation}
\begin{split}
(H_2)_{-+} = \sqrt{\frac{1}{2}}[ -\alpha_{ad}(P_A \otimes \ket{0}\bra{p-1}_w + P_A^\dagger \ket{0}\bra{1}_w) -\alpha_{bc}(P_B \otimes \ket{0}\bra{p-1}_w + P_B^\dagger \ket{0}\bra{1}_w) \\
+ (P_C \otimes \ket{0}\bra{1}_w + P_C^\dagger \ket{0}\bra{p-1}_w) + (P_D \otimes \ket{0}\bra{1}_w + P_D^\dagger \ket{0}\bra{p-1}_w)]
\end{split}
\end{equation}
If define an isometry $W$ by $W\ket{\psi}_A = \ket{\psi}_A\ket{0}_w$ then:
\begin{equation}
    || W H_{\text{{target}}} W^\dagger - (H_1)_{--} + (H_2)_{-+}H_0^{-1}(H_2)_{+-}|| = 0
\end{equation}

Therefore $\cref{second_order_eq}$ is satisfied for all $\epsilon \geq 0$. So, provided $\Delta$ is chosen to satisfy the conditions of \cref{second_order}, $\tilde{H}$ is a $(\frac{\Delta}{2},\eta,\epsilon)$-simulation of $H_{\text{target}}$.

%\begin{equation}
%\begin{split}
%\Sigma_-(z) = \left( H_{\text{else}}' + \frac{\Delta}{z - \Delta} \left[ -\alpha_{ad} \left(P_A\otimes P_D + P_A^\dagger \otimes P_D^\dagger \right) -\alpha_{bc} \left(P_B\otimes P_C + P_B^\dagger \otimes P_C^\dagger \right) \right] \right) \otimes\ket{0}\bra{0}_w \\
%+  \frac{\Delta}{z - \Delta} \left[H_{\text{else}}' - H_{\text{else}} \right] \otimes\ket{0}\bra{0}_w + O\left(\frac{||V||^3}{(z- \Delta)^2} \right)
%\end{split}
%\end{equation}
%Expanding $\Sigma_-(z)$ about $z=0$ and setting $H_{\text{eff}} = H_{\text{target}}\otimes \ket{0}\bra{0}_w$ gives:
%\begin{equation}
%||\Sigma_-(z) - H_{\text{eff}}|| = O\left(\frac{|z|r^2}{\Delta} \right) + O\left(\frac{||V||^2}{\Delta^2} \right) + O\left(\frac{|z| ||V||^3}{\Delta^3} \right)
%\end{equation}
%where $r = \max(P_A,P_B,P_C,P_D)$.
%
%For sufficiently large $\Delta$ we can make $\Sigma_-(z)$ arbitrarily close to $H_{\text{eff}}$ for the required range of $|z|$.

\noindent\paragraph{Qudit fork gadget}\\
We want to generate the Hamiltonian:
\begin{equation}
H_{\text{target}} = H_{\text{else}} + \alpha_{ab}\left( P_A \otimes P_B + P_A^\dagger \otimes P_B^\dagger \right) + \alpha_{ac}\left( P_A \otimes P_C + P_A^\dagger \otimes P_C^\dagger \right)
\end{equation}
Let $\tilde{H} = H+V$ where:
\begin{equation}
H = \Delta \Pi_+
\end{equation}
\begin{equation}
V = H_1 + \Delta^{\frac{1}{2}}H_2
\end{equation}
where:

\begin{equation}
H_1 = H_{\text{else}} + \alpha_{ab}\alpha_{ac}\left(P_B \otimes P_C^\dagger + P_B^\dagger \otimes P_C \right) + \identity \left(1 + \alpha_{ab}^2 + \alpha_{ac} + \alpha_{ac}^2 \right)
\end{equation}

\begin{equation}
\begin{split}
H_2 = \frac{1}{\sqrt{2}}[ -(P_A \otimes X_w + P_A^\dagger \otimes X^\dagger_w) + \alpha_{ab}(P_B \otimes X_w^\dagger + P_B^\dagger \otimes X_w) \\
+ \alpha_{ac}(P_C \otimes X_w^\dagger + P_C^\dagger \otimes X_w) ]
\end{split}
\end{equation}

%\begin{equation}
%\begin{split}
%V = H_{\text{else}}'+\sqrt{\frac{\Delta}{2}} (-(P_A \otimes X_w + P_A^\dagger \otimes X^\dagger_w) + \alpha_{ab}(P_B \otimes X_w^\dagger + P_B^\dagger \otimes X_w) \\
%+ \alpha_{ac}(P_C \otimes X_w^\dagger + P_C^\dagger \otimes X_w))
%\end{split}
%\end{equation}
%where:
%\begin{equation}
%H_{\text{else}}' = H_{\text{else}} + \alpha_{ab}\alpha_{ac}\left(P_B \otimes P_C^\dagger + P_B^\dagger \otimes P_C \right) + \identity \left(1 + \alpha_{ab}^2 + \alpha_{ac}^2 \right)
%\end{equation}
Then:
\begin{equation}
(H_1)_{--} = H_1\otimes \ket{0}\bra{0}_w
\end{equation}
\begin{equation}
\begin{split}
(H_2)_{-+} = \frac{1}{\sqrt{2}} [-(P_A \otimes \ket{0}\bra{p-1}_w + P_A^\dagger \otimes \ket{0}\bra{1}_w) + \alpha_{ab}(P_B \otimes \ket{0}\bra{1}_w + P_B^\dagger \otimes \ket{0}\bra{p-1}_w)  \\
+\alpha_{ac}(P_C \otimes \ket{0}\bra{1}_w + P_C^\dagger \otimes \ket{0}\bra{p-1}_w)]
\end{split}
\end{equation}

If define an isometry $W$ by $W\ket{\psi}_A = \ket{\psi}_A\ket{0}_w$ then:
\begin{equation}
    || W H_{\text{{target}}} W^\dagger - (H_1)_{--} + (H_2)_{-+}H_0^{-1}(H_2)_{+-}|| = 0
\end{equation}

Therefore $\cref{second_order_eq}$ is satisfied for all $\epsilon \geq 0$. So, provided $\Delta$ is chosen to satisfy the conditions of \cref{second_order}, $\tilde{H}$ is a $(\frac{\Delta}{2},\eta,\epsilon)$-simulation of $H_{\text{target}}$.

%Therefore:
%\begin{equation}
%\Sigma_-(z) = H_{\text{else}}'\otimes \ket{0}\bra{0}_w - \frac{\Delta}{z - \Delta}H_\text{target}\otimes \ket{0}\bra{0} +  \frac{\Delta}{z - \Delta}(H_{\text{else}}' - H_{\text{else}})\otimes \ket{0}\bra{0} + O\left(\frac{||V||^3}{(z- \Delta)^2} \right)
%\end{equation}
%Expanding $\Sigma_-(z)$ about $z=0$ and setting $H_{\text{eff}} = H_{\text{target}}\otimes \ket{0}\bra{0}_w$ gives:
%\begin{equation}
%||\Sigma_-(z) - H_{\text{eff}}|| = O\left(\frac{|z|r^2}{\Delta} \right) + O\left(\frac{||V||^2}{\Delta^2} \right) + O\left(\frac{|z| ||V||^3}{\Delta^3} \right)
%\end{equation}
%where $r = \max(P_A,P_B,P_C)$.
%
%For sufficiently large $\Delta$ we can make $\Sigma_-(z)$ arbitrarily close to $H_{\text{eff}}$ for the required range of $|z|$.

\section{Translational invariance in the boundary model} \label{appendix_ti}

In general the boundary model which results from pushing a translationally invariant bulk Hamiltonian through the HQECC will not be translationally invariant, but for particular choices of tessellation and (pseudo-)perfect tensor the boundary model will exhibit block translational invariance.

To see how this comes about consider the example discussed in \cref{example_2}.
First consider the symmetry of the honeycombing of $\mathbb{H}^3$.
The tessellation is the order-4 dodecahedral honeycomb.
The symmetry group of the dodecahedron is the icosahedral symmetry group, which is the Coxeter group $H_3$ with Coxeter diagram given in \cref{icosahedral-group}.
The rotation subgroup of this group is the alternating group $A_5$, and contains rotations by $\frac{2\pi}{5}$ about centres of pairs of opposite faces, rotations by $\pi$ about centres of pairs of opposite edges, and rotations by $\frac{2\pi}{3}$ about pairs of opposite vertices.
The symmetry group of the entire tessellation is the Coxeter group $\overline{BH}_3$, which has Coxeter diagram given in \cref{bh3-group}.\footnote{This is not the Coxeter diagram given for the tessellation in \cref{example_2}.
  In general a Coxeter group can have many different Coxeter diagrams depending on which presentation is used.
  In \cref{example_2} we used the presentation corresponding to reflections in the faces of the dodecahedron.
  Here we are using the Coxeter diagram which makes the link between $H_3$ and $\overline{BH}_3$ explicit.}
Clearly $H_3 < \overline{BH}_3$.
Therefore the symmetry group of the tessellation contains all of the rotational symmetries of the dodecahedron itself.

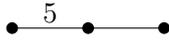
\begin{figure} \centering
  \begin{tikzpicture} \filldraw [black] (0,0) circle (2pt); \filldraw [black] (1,0) circle (2pt); \filldraw [black] (2,0) circle (2pt); \draw (0,0) -- (1,0); \draw (1,0)--(2,0); \node at (0.5,0.2) {5};
  \end{tikzpicture}
  \caption{Coxeter diagram for the icosahedral symmetry group $H_3$.}
  \label{icosahedral-group}
\end{figure}

\begin{figure} \centering
  \begin{tikzpicture} \filldraw [black] (0,0) circle (2pt); \filldraw [black] (1,0) circle (2pt); \filldraw [black] (2,0) circle (2pt); \filldraw [black] (3,0) circle (2pt); \draw (0,0) -- (1,0); \draw (1,0)--(2,0); \draw (2,0)--(3,0); \node at (0.5,0.2) {5}; \node at (2.5,0.2) {4};
  \end{tikzpicture}
  \caption{Coxeter diagram for the group $\overline{BH}_3$.}
  \label{bh3-group}
\end{figure}
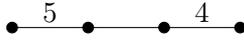

Cutting off the tessellation at some finite radius will not break the rotational symmetry. Therefore, so long as we can align the perfect tensors within the dodecahedral cells of the tessellations in such a way that the tensors don't break the symmetry, the HQECC will have the same rotational symmetry as the dodecahedron.

Ignoring the central tensor for now, it is clear that for the remaining tensors in the network it is possible to align them in such a way that rotational symmetry about at least one axis
is preserved.
To see this consider starting with an empty tessellation (of finite radius).
Pick an arbitrary cell in the tessellation, $P^{(w)}$, and place the pseudo-perfect tensor in that cell in an arbitrary orientation.
Now pick an axis of rotation, and consider rotating the tessellation by the minimum rotation about that axis which is in $H_3$.
This sends $P^{(w)}$ to $P^{(w')}$, and the resulting tensor in $P^{(w')}$ will have some particular orientation.
Place a tensor with this orientation in $P^{(w')}$.
We can now repeat this process, placing tensors in every cell which is equivalent to $P^{(w)}$ under rotation about this axis.
Then pick another empty cell in the tessellation, and repeat the process, keeping the axis of rotation the same.
We are guaranteed to be able to complete the process consistently as rotations about the same axis commute, and there are no conditions on how tensors in neighbouring cells have to be connected.

Now consider the central tensor.
Rotating the HQECC doesn't send the central tensor to another tensor in the network, it permutes 12 of the indices of the bulk tensor (leaving the final index, the bulk logical index, unchanged).
The stabilizer generators of the pseudo-perfect tensor used in the HQECC are given in \cref{gen-pseudo}.
Viewed as a isometry from any one index to the other twelve indices the pseudo-perfect tensor is the encoding isometry of a $[12,1,6]_{13}$ QECC.
Reed-Solomon codes are cyclic codes, so the pseudo-perfect tensor is symmetric under cyclic permutations of the 13 indices.
Which index we chose as the logical index is therefore not important.

Reading off from \cref{non-standard} one of the stabilizer generators of the AME(13,13) state is $X^{\otimes 13}$.
Therefore using the process described in~\cite{Gottesman:1997} for generating new stabilizer codes from old stabilizer codes we can construct a logical $X$ operator for the $[12,1,6]_{13}$ code as:
\begin{equation}
  \bar{X} = X^{\otimes 12}
\end{equation}

In order to construct a logical $Z$ operator we need to find an operator which commutes with every element of the stabilizer such that $\bar{X}\bar{Z} = \omega \bar{Z}\bar{X}$.
One such operator is:
\begin{equation}
  \bar{Z} = \left(Z^{12} \right)^{\otimes 12}
\end{equation}

Both the encoded $X$ and $Z$ operators on the central bulk index can be realised using operators which are symmetric under any permutation of the contracted tensor indices.
Therefore so can any operator we push through the central bulk tensor, so the central tensor does not break the rotational symmetry of the HQECC.

Since the HQECC (including the tensors) can be constructed to preserve rotational symmetry about at least one axis, a rotation about that axis will send the entire HQECC, including the boundary, to itself.
Therefore the boundary exhibits a form of `block translational invariance' - the Hamiltonian is a repeating pattern.

The existence of translationally invariant universal quantum Hamiltonians is an open question (in the classical case it has been shown that translationally invariant universal Hamiltonians do exist~\cite{kohler:2018}).
If translationally invariant universal quantum models were found it may be possible to construct a HQECC where the boundary Hamiltonian exhibits full translational invariance.

\bibliography{References_ads}
\bibliographystyle{abbrv}
\end{document}